\documentclass[letterpaper,11pt]{article}

\usepackage{cite}
\usepackage{tikz}\usetikzlibrary{matrix,decorations.pathreplacing,positioning}
\usepackage{hyperref}
\usepackage{amsmath,amsthm,amssymb,mathabx,bm,relsize}

\usepackage{wasysym}

\usepackage{enumitem}
\usepackage{authblk}

\usepackage{titling}
\usepackage[top=2.2cm, left=2.3cm, bottom=2.2cm, right=2.3cm]{geometry}

\definecolor{myg}{RGB}{220,220,220}
 
\theoremstyle{definition}
\newtheorem{theorem}{Theorem}[subsection]
\newtheorem{corollary}[theorem]{Corollary}
\newtheorem{proposition}[theorem]{Proposition}
\newtheorem{lemma}[theorem]{Lemma}
\newtheorem{definition}[theorem]{Definition}
\newtheorem{example}[theorem]{Example}
\newtheorem{notation}[theorem]{Notation}
\newtheorem{remark}[theorem]{Remark}

\newtheorem{theoremapp}{Theorem}[section]

\newtheorem{lemmaapp}[theoremapp]{Lemma}

\hypersetup{
  colorlinks   = true, %Colours links instead of ugly boxes
  urlcolor     = blue, %Colour for external hyperlinks
  linkcolor    = blue, %Colour of internal links
  citecolor   = red    %Colour of citations
}

\newcommand{\numberset}{\mathbb}
\newcommand{\N}{\numberset{N}}

\newcommand{\R}{\numberset{R}}

\newcommand{\F}{\numberset{F}}

\newcommand{\mC}{\mathcal{C}}

\newcommand{\mG}{\mathcal{G}}

\newcommand{\mA}{\mathcal{A}}
\newcommand{\mN}{\mathcal{N}}
\newcommand{\mB}{\mathcal{B}}
\newcommand{\mF}{\mathcal{F}}

\newcommand{\mD}{\mathcal{D}}
\newcommand{\mU}{\mathcal{U}}

\newcommand{\mX}{\mathcal{X}}
\newcommand{\mY}{\mathcal{Y}}
\newcommand{\dH}{d_\textnormal{H}}

\newcommand{\mE}{\mathcal{E}}
\renewcommand{\mG}{\mathcal{G}}
\newcommand{\mV}{\mathcal{V}}
\newcommand{\mR}{\mathcal{R}}

\newcommand{\inn}{\textnormal{in}}
\newcommand{\out}{\textnormal{out}}
\newcommand{\Id}{\textnormal{Id}}
\newcommand{\rest}{\textnormal{cp}}
\newcommand{\Enc}{\textnormal{Enc}}
\newcommand{\Dec}{\textnormal{Dec}}
\newcommand{\CA}{\textnormal{C}_1}
\newcommand{\CAz}{\textnormal{C}_0}
\newcommand{\CAzr}{\textnormal{C}_{0,\rest}}
\newcommand{\reg}{\mR_1}
\newcommand{\regz}{\mR_0}
\newcommand{\regzr}{\mR_{0,\rest}}
\newcommand{\HH}{\textnormal{H}}
\newcommand{\adv}{\textnormal{\textbf{A}}}
\newcommand{\dto}{\dashrightarrow}

\newcommand{\vertprod}{\renewcommand{\arraystretch}{0.3}\begin{array}{c}\times\\\raisebox{1ex}{\vdots}\\\times\end{array}}
\newcommand{\concat}{\RHD} %requires wasysym package
\newcommand\HP[3]{\HH_{#1,#2}\langle #3 \rangle}
\newcommand\restHP[3]{\overline{\HH}_{#1,#2}\langle #3 \rangle}
\newcommand\HB[3]{\HH_{\bm{#1},\bm{#2}}\langle \bm{#3} \rangle}
\newcommand\AP[3]{\adv_{#1,#2}\langle #3 \rangle}
\newcommand\AB[3]{\adv_{\bm{#1},\bm{#2}}\langle \bm{#3} \rangle}
\newcommand\HHB[2]{\HH_{\bm{#1}}\langle \bm{#2} \rangle}
\newcommand\HHP[2]{\HH_{#1}\langle #2 \rangle}

\newlength{\mynodespace}
\title{\textbf{\huge{Adversarial Network Coding}}}

\author{Alberto Ravagnani\thanks{The author was
partially supported by the Swiss National Science Foundation through
grant n. P2NEP2\_168527.} \ and Frank R. Kschischang}
  
 \affil{The Edward S. Rogers
Sr.\ Department\ of Electrical and Computer Engineering \\
University of Toronto, Toronto, ON M5S 3G4, Canada\\
 Email: \texttt{\{ravagnani,frank\}@ece.utoronto.ca}.}

\date{}

\begin{document}
\maketitle

\begin{abstract}
A combinatorial framework for adversarial network coding is presented.
Channels are described by specifying the possible
actions that one or more (possibly coordinated) adversaries may take.
Upper bounds on three notions of capacity---the one-shot capacity, the
zero-error capacity, and the compound zero-error capacity---are obtained
for point-to-point channels, and generalized to corresponding capacity
regions appropriate for multi-source networks.  A key result of this
paper is a general method by which bounds on these capacities in
point-to-point channels may be ported to networks.  This technique is
illustrated in detail for Hamming-type channels with multiple
adversaries operating on specific coordinates, which correspond, in the
context of networks, to multiple adversaries acting on specific network
edges.  Capacity-achieving coding schemes are described for some of the
considered adversarial models. 
\end{abstract}

%%%%%%%%%%%%%%%%%%%%%%%%%%%%%%%%%%%%%%%%

\vspace{3em}

\section{Introduction}\label{intr}

In this
paper, we propose a mathematical framework for adversarial
network coding, introducing combinatorial tools and techniques for the 
analysis of
communication networks under adversarial models.  The actions that one
or more adversaries may take are described using fan-out sets, thereby
allowing for a wide variety of possible communication scenarios.  The
framework applies to single and multi-source networks, and to single and
multiple adversaries of various kinds.  Three notions of capacity,
called ``one-shot'', ``zero-error'', and ``compound zero-error'' are
defined.  For networks, we study three corresponding notions of capacity
region, establishing bounds and describing capacity-achieving schemes.

Since 2002, the problem of correcting errors caused by an adversary in
the context of network coding has been an active research area.  The
fundamentals of error correction within network coding schemes were
originally investigated in~\cite{CaiYeung,YeungCai} by Cai and Yeung,
who studied errors and erasures in the framework of single-source
networks, and established the network analogues of the Singleton and the
Hamming bounds.  Other bounds were derived in~\cite{YeungCai} for
{regular} networks~\cite[Definition 4]{YeungCai}.  Bounds and
error-correcting code constructions also appeared in~\cite{CaiYeung,YangYeung,YangNgaiYeung,YangYeung2,Zhang,YangYeungZhang,Matsumoto}. 

Various adversarial models have been investigated in the context of
network coding.  For example, Byzantine attacks are studied in~\cite{byzantine} and~\cite{nutmanlangberg}, in which omniscient
adversaries, secret-sharing models, and adversaries of limited
eavesdropping power are considered.  Adversaries who can control some of
the network's vertices are investigated in~\cite{kosut14,WSK}.
End-to-end approaches to error control in random and coherent network
coding were proposed in~\cite{KK1,SKK,onmetrics}, along with efficient
coding and decoding schemes based on rank-metric and subspace codes.
Other models were proposed in~\cite{jaggi2005,ho2008,randomlocations}. 
Adversaries controlling a channel state within probabilistic channel models
were studied in \cite{kos_prob}.

Fan-out set descriptions of adversaries 
in point-to-point channels
were proposed in~\cite[Section~III]{onmetrics},
investigating connections 
between such descriptions 
and the concept of correction capability of a code.

The problem of error correction in the context of multi-source random
linear network coding was recently addressed in~\cite{epfl1,epfl2,epfl3,MANIAC}, 
and capacity-achieving schemes were given in~\cite{MANIAC}.

Although a wide variety of network adversarial models have been studied
by various authors (mainly focusing on single-source linear network
coding), a unified combinatorial treatment of adversarial network
channels seems to be absent.  One of the goals of this paper is to fill
this gap.  In particular, this paper makes the following contributions.

Whereas noisy channels are described within a theory of ``probability,''
adversarial channels are described within a theory of ``possibility.''
Accordingly, throughout this paper, we take a combinatorial approach to
defining and studying channels, rather than a probabilistic one.  This
approach is inspired by Shannon's work~\cite{shannon_zero} on the
zero-error capacity of a channel, and motivated by the fact that a network adversary may not
be specified in general via random variables and probability distributions. In Section~\ref{firstsection} we
define codes, one-shot capacity, channel products, and the zero-error
capacity of a channel, and we describe how this approach relates to
Shannon's work.  In Section~\ref{secoperations}, we define two
fundamental channel operations: ``concatenation'' and ``union,'' and we
establish their main algebraic properties. We also show how they relate
to each other and to the notions of capacity.

In Section~\ref{sechammingtype} we study certain adversarial channels
called ``Hamming-type channels,'' whose input alphabet is a cartesian
product of the form $\mA^s$.  We consider multiple adversaries who can
corrupt or erase the components of an element $x \in \mA^s$, according
to certain restrictions, and explicitly compute the one-shot capacity,
the zero-error capacity and the compound zero-error capacity of these
channels (see Subsection~\ref{subsectioncompound} for the definitions).
This extends a number of classical results in Coding Theory.

The study of networks starts with Section~\ref{secnetwork}.  In contrast
to previous approaches (e.g.,~\cite{CaiYeung,YeungCai}), our framework
allows for multi-source networks with a wide variety of adversarial
models.  We propose three notions of capacity region of a multi-source
adversarial network, which we call the ``one-shot capacity region'', the
``zero-error capacity region'', and the ``compound zero-error capacity
region'' (see Subsection~\ref{subcapacityregions} for precise
definitions). Most previous work in adversarial network coding
 implicitly
focus on one-shot models (cf. also Remark \ref{rmkisdifferent}), while to our best knowledge zero-error and
compound zero-error adversarial models have not so far been investigated.

The centerpiece of this paper is Section~\ref{secportinglemma}, in which
we show that any upper bounds for the capacities of Hamming-type
channels can be ported to the networking context in a systematic manner.
Using the channel operations defined in Section~\ref{secoperations}, we
show that this ``porting technique'' applies to all three notions of
capacity region mentioned above, in the general context of multi-source
networks.  Moreover, this method does not require the underlying
network to be regular, in the sense of~\cite[Definition 4]{YeungCai}. 

These theoretical results are then applied to concrete networking
contexts in Section~\ref{secbounds}, where we study multiple
adversaries, each with possibly different error and erasure powers,
having access to prescribed subsets of the network edges.  The
adversaries are in principle allowed to coordinate with each other.  We
derive upper bounds for the three capacity regions of such adversarial
networks, extending certain results of~\cite{YeungCai} to multiple
adversaries (restricted or not) and multi-source networks.

In Section~\ref{secconstructions}, we give capacity-achieving schemes
for some of the adversarial scenarios investigated in
Section~\ref{secbounds}.  Incorporating ideas from~\cite{MANIAC} 
and~\cite{koettermedard}, we show that linear network coding suffices to
achieve any integer point of the capacity regions associated with some
simple adversaries. We then adapt these communication schemes to
compound models. Finally, we show that for some adversarial networks
capacity cannot be achieved with linear network coding.

Other classes of adversaries (such as rank-metric adversaries and
multiple adversaries having access to overlapping sets of the network
edges) are briefly discussed in Section~\ref{extensions}. 
Finally,
Section~\ref{secconcl} is devoted to conclusions and a discussion of
open problems.

%%%%%%%%%%%%%%%%%%%%%%%%%%%%%%%%%%%

\section{Adversarial Channels} \label{firstsection}

In this section we define and study point-to-point adversarial channels,
one-shot codes, and one-shot capacity. We then describe the channel
product construction, and define the zero-error capacity of an
adversarial channel, as proposed by Shannon in~\cite{shannon_zero}.
Although most of this material is well known, this section serves to
establish notation that will be used throughout the paper.

\subsection{Channels, Codes, Capacity}

An adversarial channel is described by an input alphabet $\mX$, an
output alphabet $\mY$, and a collection $\{\Omega(x) : x \in \mX\}$ of
subsets of $\mY$, one for each $x \in \mX$. The set $\Omega(x) \subseteq
\mY$ is interpreted as the \textbf{fan-out set} of $x$, i.e., the set of
all symbols $y \in \mY$ that the adversary can cause to be received when
the input symbol $x \in \mX$ is transmitted. This motivates the
following definition.

\begin{definition} An (\textbf{adversarial}) \textbf{channel} is a map
$\Omega: \mX \to 2^\mY \setminus \{\emptyset\}$, where $\mX$ and $\mY$
are finite non-empty sets called \textbf{input} and \textbf{output}
\textbf{alphabet}, respectively.  We denote such an adversarial channel
by $\Omega: \mX {\dto} \mY$, and say that $\Omega$ is
\textbf{deterministic} if $|\Omega(x)|=1$ for all $x \in \mX$.  A
deterministic channel can be identified naturally with the function $\mX
\to \mY$ that associates to $x \in \mX$ the unique element
$y \in \Omega(x)$.
\end{definition}

In this paper we will restrict our attention to channels whose input and
output alphabets are finite. 

\begin{definition}
Let $\Omega: \mX \dto \mY$ be a channel. A (\textbf{one-shot}) \textbf{code} 
for $\Omega$ is a non-empty 
subset $\mC \subseteq \mX$. 
We say that $\mC$ is \textbf{good} for $\Omega$ when 
$\Omega(x) \cap \Omega(x') = \emptyset$ for all $x,x' \in \mC$ with
$x \neq x'$.
\end{definition}

In words, a good code for a channel $\Omega: \mX \dto \mY$ is a
selection of input symbols from $\mX$ whose fan-out sets are pairwise
disjoint. If channel inputs are restricted to a good code, it is
impossible for an adversary to cause confusion at the receiver about the
transmitted symbol.

\begin{definition}\label{defC}
The (\textbf{one-shot}) \textbf{capacity} of a channel $\Omega: \mX \dto \mY$ 
 is the base-2 logarithm 
of the largest cardinality of a good code for $\Omega$, i.e.,
$$\CA(\Omega) := \max\{ \log_2 |\mC| : \mbox{$\mC \subseteq \mX$ is good 
for $\Omega$} \}.$$
\end{definition}

Clearly, the capacity of any channel $\Omega: \mX \dto \mY$
satisfies $0 \le \CA(\Omega) \le \min\{\log_2 |\mX|, 
\log_2 |\mY|\}$.

\begin{example}\label{exid}
Let $\mX,\mY$ be finite non-empty sets with $\mY \supseteq \mX$. 
The \textbf{identity channel} 
$\textnormal{Id} : \mX \dto \mY$ is defined by 
 $\Omega(x):=\{x\}$ for all $x \in \mX$.
We have $\CA(\textnormal{Id})= \log_2|\mX|$.
\end{example}

\begin{example} \label{firstexample}
Let $\mX:=\F_2^4$. Consider an adversary who is capable of corrupting
at most one of the components of any $x \in \F_2^4$. The action of the 
adversary is described by the channel 
$\HH: \F_2^4 \dto \F_2^4$ defined by $\HH(x):=\{y \in \F_2^4 : \dH(x,y) \le 1\}$
for all $x \in \F_2^4$,
where $\dH$ is the Hamming distance. The code
$\mC=\{(0000), (1111)\}$ is good for $\HH$, and there is no good code
with larger cardinality. Therefore we have $\CA(\HH)=1$.
\end{example}

%A combinatorial bound for the capacity of a channel is the following.
%The proof is simple and left to the reader.
%
%\begin{proposition}
%Let $\mC \subseteq \mX$ be good for the channel 
%$\Omega: \mX \dto \mY$.
%Then 
%$$|\mC| \cdot \min_{x \in \mC} |\Omega(x)| \le {|\mY|}.$$
%\end{proposition}
%%\begin{proof}
%%By definition of a good code, for all $x,x' \in \mC$ with $x \neq x'$ 
%we have $\Omega(x) \cap \Omega(x') = \emptyset$. Therefore 
%%\begin{equation*}
%%|\mY| \ge \sum_{x \in \mC} |\Omega(x)| \ge |\mC| 
%\cdot \min_{x \in \mC} |\Omega(x)|. \qedhere
%%\end{equation*}
%%\end{proof}

Channels with the same input and output alphabets can be compared as follows.

\begin{definition} \label{deffiner}
Let $\Omega_1,\Omega_2: \mX \dto \mY$ be channels.
We say that $\Omega_1$ is
\textbf{finer} than $\Omega_2$ (in symbols, $\Omega_1 \le \Omega_2$ or $\Omega_2 \ge \Omega_1$) when 
$\Omega_1(x) \subseteq \Omega_2(x)$ 
for all $x \in \mX$.
\end{definition}

If $\Omega_1 \le \Omega_2$, then every code that is good for 
$\Omega_2$ is good for $\Omega_1$ as well. Therefore we have 
 $\CA(\Omega_1) \ge \CA(\Omega_2)$.

\subsection{Products of Channels}

Adversarial channels $\Omega_1$ and $\Omega_2$ can be naturally combined 
with each other via a product construction,
giving rise to a third channel denoted by $\Omega_1 \times \Omega_2$.

\begin{definition}
The \textbf{product} of channels $\Omega_1 : \mX_1 \dto \mY_1$ and 
$\Omega_2 : \mX_2 \dto \mY_2$
is the channel
$\Omega_1 \times \Omega_2 : \mX_1 \times \mX_2 \dto  \mY_1 \times \mY_2$
defined by   
$$(\Omega_1 \times \Omega_2)(x_1,x_2):= \Omega_1(x_1) \times 
\Omega_2(x_2), \mbox{ for all 
$(x_1,x_2) \in \mX_1 \times \mX_2$}.$$
\end{definition}

The following result shows two important properties of the channel product.

\begin{proposition}\label{boundprod}
Let $\Omega_1, \Omega_2, \Omega_3$ be channels. Then
\begin{enumerate}\setlength\itemsep{0em}
\item $(\Omega_1 \times \Omega_2 )\times \Omega_3 = \Omega_1 
\times (\Omega_2 \times \Omega_3)$,\label{boundprod1}
\item $\CA(\Omega_1 \times \Omega_2) \ge \CA(\Omega_1) + \CA(\Omega_2)$. 
\label{boundprod2}
\end{enumerate}
\end{proposition}
\begin{proof}
The first property is straightforward. To see the second, observe that 
if $\mC_1$ and $\mC_2$ are good codes 
for $\Omega_1$ and $\Omega_2$ (respectively), then   
$\mC_1 \times \mC_2$ is good for $\Omega_1 \times \Omega_2$.
\end{proof}

The associativity of the channel product (part~\ref{boundprod1} of 
Proposition~\ref{boundprod}) allows expressions 
such as $\Omega_1 \times \Omega_2 \times \Omega_3$ to be written 
without danger of ambiguity.

\begin{definition}
Let $n \ge 1$ be an integer. The \textbf{$n$-th power} of a channel 
$\Omega : \mX \dto \mY$ is the channel
$$\Omega^n := \underbrace{\Omega \times \cdots \times 
\Omega}_{n \text{ times}} :
\mX^n \dto \mY^n.$$
\end{definition}

The $n$-th power of $\Omega : \mX \dto \mY$ models 
$n$ uses of $\Omega$. Recall that the elements of 
$\Omega(x)$ represent the outputs that an adversary can produce from the input 
$x \in \mX$. If the channel $\Omega$ is used $n$ times, then we have
$\Omega^n(x_1,...,x_n)= \Omega(x_1) \times \cdots \times \Omega(x_n)$.

\begin{remark}
In general, the lower bound of part~\ref{boundprod2} of 
Proposition~\ref{boundprod} is not tight, i.e., the capacity of the 
product channel $\Omega_1 \times \Omega_2$ can be strictly larger than the 
sum of the capacities of the 
channels $\Omega_1$ and $\Omega_2$.
The following example, which will be used repeatedly in the paper,
illustrates this point. 
\end{remark}

\begin{example} \label{fee}
Let $\HH: \F_2^4 \dto \F_2^4$ be the channel of 
Example~\ref{firstexample}, which has capacity
 $\CA(\HH)=1$. The code 
$\mC:=\{(00000000), (00011101),  (10100111), (11010110), (11101000)\} 
\subseteq \F_2^4 \times \F_2^4$
is good for the product channel $\HH \times \HH$. Therefore
$\CA(\HH \times \HH) \ge \log_2(5)>2 = \CA(\HH)+\CA(\HH)$.

For $x =(x_1,...,x_8)\in \F_2^4 \times \F_2^4$,  let 
$x^1:=(x_1,...,x_4)$ and $x^2:=(x_5,...,x_8)$. Then
a code $\mC \subseteq \F_2^4 \times \F_2^4$ is good for $\HH^2$ if and 
only if $\dH(x^1,y^1) \ge 3$ or $\dH(x^2,y^2) \ge 3$
for all $x,y \in \mC$
with $x \neq y$.

We conclude the example by showing a structural property of any good code
$\mC$ for $\HH^2$ with $|\mC|=5$. The property will be needed later  in 
 Example~\ref{ex_isdifferent}.

Let $\mC \subseteq \F_2^4 \times \F_2^4$ be any good code for $\HH^2$ 
with $|\mC|=5$.
We claim that there are no two codewords of 
$\mC$ that coincide in the first four components.
To see this, denote by $x,y,z,t,u$ the elements of $\mC$, and assume by 
contradiction that, say,
$x^1=y^1$. Without loss of generality, we may assume $x=0$ (and thus $y^1=0$).
Then the vectors $z^1,t^1,u^1$ must have Hamming weight at least 3. 
Indeed, if, say, $z^1$
has Hamming weight smaller than 3, then $\{x^2,y^2,z^2\} \subseteq \F_2^4$ is a 
code of cardinality
$3$ and minimum Hamming distance $3$, contradicting the fact that 
$\CA(\HH)=1$.
On the other hand, since $z^1,t^1,u^1$ have Hamming weight at least 3, we have 
$\dH(z^1,t^1), \dH(z^1,u^1), \dH(t^1,u^1) \le 2$. Since $\mC$ is good 
for $\HH^2$,  $\{z^2,t^2,u^2\}$ must be 
a code of cardinality
$3$ and minimum Hamming distance 3, again contradicting
the fact that $\CA(\HH)=1$. 
\end{example}

\subsection{Adjacency Structure of a Channel} \label{subadj}
In this subsection we introduce the adjacency function of a channel, and 
propose a definition of isomorphic channels. In particular, we relate 
the fan-out set description of channels adopted in this paper with the graph-theoretic 
approach taken by Shannon in~\cite{shannon_zero}.

\begin{definition} \label{defisoch}
The \textbf{adjacency function} $\alpha_\Omega: \mX \times \mX \to \{0,1\}$
 of  a channel 
 $\Omega: \mX \dto \mY$ is defined, for all $x,x' \in \mX$, by 
$$\alpha_\Omega(x,x'):= \left\{ 
\begin{array}{ll} 1 & \mbox{ if } \Omega(x) \cap \Omega(x') \neq \emptyset, \\
0 & \mbox{ otherwise.}
 \end{array}\right.$$
We say that channels $\Omega_1: \mX_1 \dto \mY_1$ and 
$\Omega_2: \mX_2 \dto \mY_2$
 are \textbf{isomorphic} (in symbols, $\Omega_1 \cong \Omega_2$) if there
  exists a bijection $f: \mX_1 \to \mX_2$ such that 
 $\alpha_{\Omega_1}(x,x') = \alpha_{\Omega_2}(f(x),f(x'))$  for all  
 $(x,x') \in \mX_1 \times \mX_2$.
\end{definition}

The adjacency function of 
$\Omega$ captures the ``ambiguity relations'' among the 
input symbols of $\Omega$. 
Channels $\Omega_1$ and $\Omega_2$ are isomorphic if their input symbols have the 
same ambiguity relations, for some identification of their input alphabets.

The isomorphism class of a channel $\Omega: \mX \dto \mY$ can be represented 
via a graph $\mG$ as follows. Up to a suitable bijection, the vertices of $\mG$ are the elements of
 $\mathcal{V}=\{0,1,...,|\mX|-1\}$, and $(x,x') \in \mV \times \mV$ is an edge 
of $\mG$ if and only if $\alpha_\Omega(x,x')=1$. Therefore $\alpha_\Omega$ is 
precisely the adjacency matrix of the graph 
$\mG$.
This is the way channels are 
described and studied by Shannon in~\cite{shannon_zero}. Note that,
although every vertex of $\mG$ is by definition adjacent to itself, loops are 
usually not shown in the graph description.

\begin{example}[The ``pentagon channel''] \label{expentagon}
Let $\mX=\mY:=\{0,1,2,3,4\}$, and let $\Omega: \mX \dto \mY$ be the channel defined by
$$\Omega(0):=\{0,1\}, \qquad \Omega(1):=\{1,2\}, \qquad \Omega(2):=\{2,3\}, \qquad \Omega(3):=\{3,4\},
\qquad \Omega(4):=\{4,0\}.$$
The five fan-out sets of $\Omega$ are represented as in 
Figure~\ref{fig:1}(a), and
a graph representation of the isomorphism class of 
$\Omega$ is depicted in Figure~\ref{fig:1}(b). The channel $\Omega$ was first introduced  and studied
 by Shannon in~\cite{shannon_zero}.
\end{example}

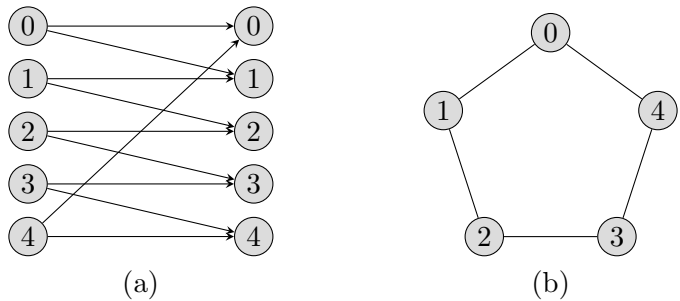
\begin{figure}[htbp]
\centering
\setlength{\tabcolsep}{1cm}
\begin{tabular}{cc}
\begin{tikzpicture}[>=stealth,
vertex/.style={shape=circle,draw,inner sep=2pt},
nnode/.style={shape=circle,fill=myg,draw,inner sep=2pt}]
\foreach \i in {0,1,...,4}
{
   \node[nnode] (L\i) at (0,-0.7*\i) {\i};
   \node[nnode] (R\i) at (3,-0.7*\i) {\i};
   \draw[->] (L\i) -- (R\i);
}
\draw[->] (L0)--(R1);
\draw[->] (L1)--(R2);
\draw[->] (L2)--(R3);
\draw[->] (L3)--(R4);
\draw[->] (L4)--(R0);
\end{tikzpicture} &
\begin{tikzpicture}[>=stealth,
vertex/.style={shape=circle,draw,inner sep=2pt},
nnode/.style={shape=circle,fill=myg,draw,inner sep=2pt}]
\foreach \i in {0,1,...,4}
{
   \node[nnode] (\i) at (90+\i*72:1.5) {\i};
}
\draw (0) -- (1);
\draw (1) -- (2);
\draw (2) -- (3);
\draw (3) -- (4);
\draw (4) -- (0);
\end{tikzpicture}\\ (a) & (b)
\end{tabular}
\caption{The ``pentagon channel:'' (a) fan-out sets, (b) graph representation.} \label{fig:1}
\end{figure}

One can show that channel isomorphism is an equivalence relation. Moreover, the
following properties of isomorphic channels hold. The proof simple and left to
the reader.

\begin{proposition} Let $\Omega_1, \Omega_2, \Omega_3, \Omega_4$ be channels. \label{connect} \label{proprisoproduct}
\begin{enumerate} \setlength\itemsep{0em}
\item If $\Omega_1 \cong \Omega_2$ and $\Omega_3 \cong \Omega_4$, then 
$\Omega_1 \times \Omega_3 \cong \Omega_2 \times \Omega_4$.
\item  \label{four} If $\Omega_1 \cong \Omega_2$, then $\Omega_1^n \cong \Omega_2^n$ for all $n\in \N_{\ge 1}$.
\item If $\Omega_1 \cong \Omega_2$, then
$\CA(\Omega_1)=\CA(\Omega_2)$.
\end{enumerate}
\end{proposition}

\subsection{Zero-Error Capacity}

In this subsection we define the zero-error capacity of an adversarial channel, and relate it to
the one-shot capacity.
Given the graph-theoretic description of channels illustrated in Subsection~\ref{subadj}, 
the following results essentially already appear in Shannon's paper~\cite{shannon_zero}.
We state them in the language of fan-out sets for convenience.
See \cite{zeroerrorinvited} for a general reference on Zero-Error Information Theory.

\begin{definition}
The \textbf{zero-error capacity} of a channel $\Omega$ is 
the number $$\CAz(\Omega):= \sup \left\{ \CA(\Omega^n)/n : n \in \N_{\ge 1}\right\}.$$
\end{definition}

\begin{remark}
It is easy to see that $\CAz(\Omega)$ is a non-negative number for every channel 
$\Omega$, i.e., that $0 \le \CAz(\Omega) < +\infty$. A less immediate (though intuitive) property 
is that supremum in the definition of zero-capacity is in fact a limit.
This can be shown using Fekete's Lemma for superadditive 
sequences (see~\cite{fekete} or~\cite[Section 1.9]{steele}). 
%To prove this, observe that by Proposition
%\ref{boundprod} the sequence $\{a_n : n \in \N_{\ge 1}\}$ defined by 
%$n \mapsto a_n:=\CA(\Omega^n)$ is superadditive, i.e.,
%$a_{n+m} \ge a_n + a_m$ for all $n,m \in \N_{\ge 1}$.
%Therefore by Fekete's Lemma for superadditive sequences 
%(see~\cite{fekete} or~\cite[Section 1.9]{steele}) we conclude
%$$\sup \{a_n/n : n \in \N_{\ge 1}\} = \lim_{n \to \infty} a_n/n,$$
%as desired.
\end{remark}

As one may expect, the zero error-capacity of a channel only depends on its 
isomorphism class.
This fact follows from property~\ref{four} of 
Proposition~\ref{connect}.

\begin{proposition}
Let $\Omega_1, \Omega_2$ be channels. If $\Omega_1 \cong \Omega_2$, then 
$\CAz(\Omega_1)=\CAz(\Omega_2)$.
\end{proposition}

The next result shows how one-shot capacity and zero-error capacity relate to each other.

\begin{proposition} \label{relate}
Let $\Omega: \mX \dto \mY$ be a channel. The following hold.
\begin{enumerate} \setlength\itemsep{0em}
\item $\CA(\Omega^n) \ge n \cdot \CA(\Omega)$ for all $n \ge 1$. 
Thus $\CAz(\Omega) \ge \CA(\Omega)$.
\item If $\CA(\Omega)=0$, then $\CAz(\Omega)=0$.
\end{enumerate}
\end{proposition}

\begin{proof}
\begin{enumerate} \setlength\itemsep{0em}
\item Let $\mC \subseteq \mX$ be good for $\Omega$ with 
$\log_2 |\mC|=\CA(\Omega)$, and let $n \in \N_{\ge 1}$. Then $\mC^n$ 
is good for $\Omega^n$. 
Thus $\CA(\Omega^n)/n \ge \CA(\Omega)$.
Since $n$ is arbitrary, this implies in particular that 
$\CAz(\Omega) \ge \CA(\Omega)$.
\item Since $\CA(\Omega)=0$, for all $x,x' \in \mX$ we have 
$\Omega(x) \cap \Omega(x') \neq \emptyset$.
 Assume by way of contradiction that $\CAz(\Omega)>0$. Then 
 there exists $n \in \N_{\ge 1}$
  with
$\CA(\Omega^n)>0$. In particular, there exists a good code 
$\mC \subseteq \mX^n$ 
for $\Omega^n$ with $|\mC| \ge 2$. Let $(x_1,...,x_n), (x_1',...,x_n') \in \mC$ with 
$(x_1,...,x_n) \neq (x_1',...,x_n')$.
Then $$\emptyset = (\Omega(x_1) \times \cdots \times \Omega(x_n)) 
\cap (\Omega(x_1') \times \cdots \times \Omega(x_n'))
= \prod_{k=1}^n \left( \Omega(x_k) \cap \Omega(x_k')\right) \neq \emptyset,$$
a contradiction. Therefore it must be that $\CAz(\Omega)=0$, as claimed. \qedhere 
\end{enumerate}
\end{proof}

\begin{remark}
In general, the zero-error capacity of a channel $\Omega$ is \emph{strictly} 
larger than its one-shot capacity. 
For example, let $\HH$ be the channel of  
Example~\ref{firstexample}. We showed that 
$\CA(\HH^2)/2 > \CA(\HH)$, which implies that
 $\CAz(\HH)>\CA(\HH)$.
\end{remark}

The zero-error capacity of a channel is a combinatorial invariant which is in general 
very difficult to compute.  For example, let $\Omega$ be the pentagon channel
of Example~\ref{expentagon}. In~\cite{shannon_zero} Shannon showed that 
$\log_2 \sqrt{5} \le \CAz(\Omega) \le \log_2 {5}-1$. The exact value 
of $\CAz(\Omega)$ was computed only
twenty-three years later by Lov\'asz in~\cite{lovasz}, using 
sophisticated Graph Theory techniques. The result is $\CAz(\Omega)=\log_2 \sqrt{5}$.

%%%%%%%%%%%%%%%%%%%%%%%%%%%%%%%%%%%

\section{Operations with Channels}
\label{secoperations}
In this section we introduce two operations with channels, namely, concatenation and union, 
showing  how they relate to each other and to the channel product.
These two constructions will play an important role throughout the paper 
in the study of several classes of adversarial point-to-point and network channels.

\subsection{Concatenation of Channels}

We start with channel concatenation.

\begin{definition}
Let $\Omega_1: \mX_1 \dto \mY_1$ and $\Omega_2: \mX_2 \dto \mY_2$ be channels, with $\mY_1 \subseteq \mX_2$.
The \textbf{concatenation}
of $\Omega_1$ and $\Omega_2$ is the channel 
$\Omega_1 \concat \Omega_2  : \mX_1 \dto \mY_2$ defined by 
$$(\Omega_1 \concat \Omega_2)(x):= \!\!\!\! \bigcup_{y \in \Omega_1(x)} \!\!\! \Omega_2(y)
\text{ for all } x \in \mX_1.$$
\end{definition}

This operation models the situation where the output of 
$\Omega_1$ is taken as the input to $\Omega_2$ without any
intermediate processing.

\begin{example}
Let $\Omega:\F_2^4 \to \F_2^4$ be the channel introduced in Example~\ref{firstexample}.
Then for all $x \in \F_2^4$ we have 
$(\Omega \concat \Omega)(x)= \{y \in \F_2^4 : \dH(x,y) \le 2\}$.
\end{example}

\begin{remark}
The isomorphism class of 
$\Omega_1 \concat \Omega_2$ is not determined in general by the 
isomorphism classes of $\Omega_1$ and $\Omega_2$, as the following example shows. This crucial difference 
between channel product and channel concatenation 
 motivates the choice of
the fan-out sets language in this paper.
\end{remark}

\begin{example}
Let $\mX=\mY:=\{0,1,2\}$. Define the adversarial channels
$\Omega_1,\Omega_2 : \mX \dto \mY$ by
$\Omega_1(0):=\Omega_1(1):=\{0,1\}$, $\Omega_1(2):=\{2\}$,
$\Omega_2(0):=\{0\}$,  
$\Omega_2(1):=\Omega_2(2):=\{1,2\}$.
It is easy to see that $\Omega_1 \cong \Omega_2$. 
However,  
$\Omega_1 \concat \Omega_1 \not\cong \Omega_2 \concat \Omega_1$, 
$\CA(\Omega_1 \concat \Omega_1) \neq \CA(\Omega_2 \concat \Omega_1)$, and 
$\CAz(\Omega_1 \concat \Omega_1) \neq \CAz(\Omega_2 \concat \Omega_1)$.
\end{example}

The following result is the analogue of Proposition~\ref{boundprod} for channel concatenation.

\begin{proposition} \label{genconc} Let $\Omega_1, \Omega_2,\Omega_3$ be channels.
Then:
\begin{enumerate} \setlength\itemsep{0em}

\item \label{genconc1}  $(\Omega_1 \concat \Omega_2 )\concat \Omega_3 = \Omega_1 \concat (\Omega_2 \concat \Omega_3)$,
\item \label{genconc2} $\CA(\Omega_1 \concat \Omega_2) \le
 \min \{\CA(\Omega_1), \CA(\Omega_2)\}$,
\end{enumerate}
provided that all of the above concatenations are defined.
\end{proposition}

\begin{proof}
Property~\ref{genconc1} easily follows from the definition of concatenation.
To see  property~\ref{genconc2}, suppose that $\Omega_1: \mX_1 \dto \mY_1$ and $\Omega_2: \mX_2 \dto \mY_2$ are channels with $\mY_1 \subseteq \mX_2$.
Let $\mC \subseteq \mX_1$ be good for $\Omega_1 \concat \Omega_2$ with $\log_2|\mC|=\CA(\Omega_1 \concat \Omega_2)$.
We will show that $\mC$ is good for $\Omega_1$. Assume that there exist
$x,x' \in \mC$ with $\Omega_1(x) \cap \Omega_1(x') \neq \emptyset$, and let $y \in \Omega_1(x) \cap \Omega_1(x')$.
Then we have 
$\emptyset \neq \Omega_2(y) \subseteq (\Omega_1 \concat \Omega_2)(x) \cap (\Omega_1 \concat \Omega_2)(x)$,
a contradiction. Therefore $\mC$ is good for $\Omega_1$, and so
$\CA(\Omega_1) \ge \log_2|\mC| = \CA(\Omega_1 \concat \Omega_2)$.
On the other hand, for every $x \in \mC$ we can select $y_x \in \Omega_1(x) \subseteq \mY_1$. It is easy to check 
 that $\mC':=\{y_x : x \in \mC\} \subseteq \mX_2$ is a good code for $\Omega_2$ with the same cardinality as $\mC$.
Therefore $\CA(\Omega_2) \ge \CA(\Omega_1 \concat \Omega_2)$.
\end{proof}

Note that the associativity of channel concatenation  allows  expressions such 
as $\Omega_1 \concat \Omega_2 \concat \Omega_3$  to be written
without danger of ambiguity. 

The following result provides an identity showing that the product of the concatenation of channels is the concatenation of their products. This property and its corollary will be needed later for the analysis of certain classes of network channels
 (see Lemma~\ref{portinglemma1}). The proof can be found in Appendix~\ref{Appe}, and
 the result is illustrated in Figure~\ref{diagramblocks}.
 
\begin{proposition} \label{concprod}
Let $n,m  \in \N_{\ge 1}$, and let 
$\Omega_{k,i}$ be channels, for $1 \le k \le n$ and $1 \le i \le m$.
Then
$$\prod_{k=1}^n \left( \Omega_{k,1} \concat \cdots \concat \Omega_{k,m} \right) = \left( \prod_{k=1}^n \Omega_{k,1}\right)
\concat \cdots \concat   \left( \prod_{k=1}^n \Omega_{k,m}\right),$$
provided that all the above concatenations are defined.
\end{proposition}

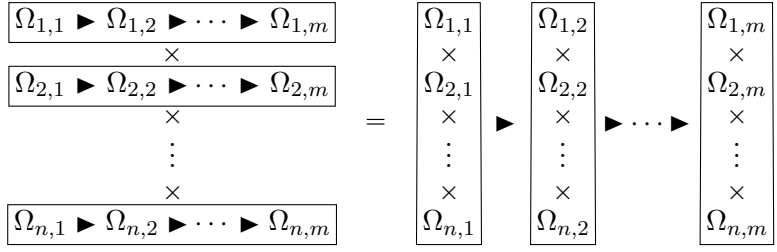
\begin{figure}[htbp]
\centering
\begin{tikzpicture}[
channel/.style={inner sep=2pt},
relation/.style={inner sep=1pt}
]
\matrix[row sep=2ex,column sep=0ex]
{
%% ROW 1
\node[channel] (11) {$\Omega_{1,1}$}; &
\node[relation]      {$\concat$}; &
\node[channel] (12) {$\Omega_{1,2}$}; &
\node[relation]      {$\concat$}; &
\node[relation]      {$\cdots$}; &
\node[relation]      {$\concat$}; &
\node[channel] (1M) {$\Omega_{1,m}$}; &\\
%% ROW 2
\node[channel] (21) {$\Omega_{2,1}$}; &
\node[relation]      {$\concat$}; &
\node[channel] (22) {$\Omega_{2,2}$}; &
\node[relation]      {$\concat$}; &
\node[relation]      {$\cdots$}; &
\node[relation]      {$\concat$}; &
\node[channel] (2M) {$\Omega_{2,m}$}; &\\
%% ROWS 3-
 & & &  & & & & \\
 & & &  & & & & \\
 & & &  & & & & \\
\node[channel] (N1) {$\Omega_{n,1}$}; &
\node[relation]      {$\concat$}; &
\node[channel] (N2) {$\Omega_{n,2}$}; &
\node[relation]      {$\concat$}; &
\node[relation]      {$\cdots$}; &
\node[relation]      {$\concat$}; &
\node[channel] (NM) {$\Omega_{n,m}$}; &\\
};
\draw (11.north west) -- (1M.north east) -- (1M.south east) -- (11.south west) -- cycle;
\draw (21.north west) -- (2M.north east) -- (2M.south east) -- (21.south west) -- cycle;
\draw (N1.north west) -- (NM.north east) -- (NM.south east) -- (N1.south west) -- cycle;
\coordinate (A1) at (11.south west);
\coordinate (A2) at (2M.north east);
\node at (barycentric cs:A1=0.5,A2=0.5) {$\times$};
\coordinate (B1) at (21.south west);
\coordinate (B2) at (NM.north east);
\node[align=center] at (barycentric cs:B1=0.5,B2=0.5) {$\vertprod$};

\matrix[xshift=15em,row sep=2ex,column sep=1.8ex]
{
%% ROW 1
\node[channel] (R11) {$\Omega_{1,1}$}; &
\node[relation]      {}; &
\node[channel] (R12) {$\Omega_{1,2}$}; &
\node[relation]      {}; &
\node[relation]      {}; &
\node[relation]      {}; &
\node[channel] (R1M) {$\Omega_{1,m}$}; &\\
%% ROW 2
\node[channel] (R21) {$\Omega_{2,1}$}; &
\node[relation]      {}; &
\node[channel] (R22) {$\Omega_{2,2}$}; &
\node[relation]      {}; &
\node[relation]      {}; &
\node[relation]      {}; &
\node[channel] (R2M) {$\Omega_{2,m}$}; &\\
%% ROWS 3-
 & & &  & & & & \\
 & & &  & & & & \\
 & & &  & & & & \\
\node[channel] (RN1) {$\Omega_{n,1}$}; &
\node[relation]      {}; &
\node[channel] (RN2) {$\Omega_{n,2}$}; &
\node[relation]      {}; &
\node[relation]      {}; &
\node[relation]      {}; &
\node[channel] (RNM) {$\Omega_{n,m}$}; &\\
};
\draw (R11.north west) -- (R11.north east) -- (RN1.south east) -- (RN1.south west) -- cycle;
\draw (R12.north west) -- (R12.north east) -- (RN2.south east) -- (RN2.south west) -- cycle;
\draw (R1M.north west) -- (R1M.north east) -- (RNM.south east) -- (RNM.south west) -- cycle;
\node at (barycentric cs:R11=0.5,R21=0.5) {$\times$};
\node[align=center] at (barycentric cs:R21=0.5,RN1=0.5) {$\vertprod$};
\node at (barycentric cs:R12=0.5,R22=0.5) {$\times$};
\node[align=center] at (barycentric cs:R22=0.5,RN2=0.5) {$\vertprod$};
\node at (barycentric cs:R1M=0.5,R2M=0.5) {$\times$};
\node[align=center] at (barycentric cs:R2M=0.5,RNM=0.5) {$\vertprod$};
\coordinate (C1) at (R11.north east);
\coordinate (C2) at (RN2.south west);
\node at (barycentric cs:C1=0.5,C2=0.5) {$\concat$};
\coordinate (D1) at (R12.north east);
\coordinate (D2) at (RNM.south west);
\node at (barycentric cs:D1=0.5,D2=0.5) {$\concat\cdots\concat$};
\coordinate (E1) at (1M.north east);
\coordinate (E2) at (RN1.south west);
\node at (barycentric cs:E1=0.5,E2=0.5) {$=$};
\end{tikzpicture}
\caption{Illustrating Proposition~\ref{concprod}: a product
of concatenations is a concatenation of products.}
\label{diagramblocks}
\end{figure}

\begin{corollary} \label{propprel5}
Let $\Omega$ be a channel, and let $n \ge 1$ be an integer. Assume that 
$\Omega_1,...,\Omega_n$ and $\Omega'_1,...,\Omega'_n$ are channels for which
the concatenation $\Omega_k \concat \Omega \concat \Omega'_k$ is defined for all
$1 \le k \le n$. Then 
\begin{equation*}
\prod_{k=1}^n (\Omega_k \concat \Omega \concat \Omega'_k) = \left(\prod_{k=1}^n \Omega_k\right) \concat \Omega^n
\concat \left(\prod_{k=1}^n \Omega'_k\right).
\end{equation*}
\end{corollary}

We can now establish the 
zero-error analogue of part~\ref{genconc2} of 
Proposition~\ref{genconc}.

\begin{proposition} \label{DPI2}
Let $\Omega_1: \mX_1 \dto \mY_1$ and $\Omega_2: \mX_2 \dto \mY_2$ be channels, with $\mY_1 \subseteq \mX_2$.
 We have 
$\CAz(\Omega_1 \concat \Omega_2) \le \min \{\CAz(\Omega_1), \CAz(\Omega_2)\}$.
\end{proposition}

\begin{proof}
Fix any integer $n\ge 1$. Combining Proposition~\ref{genconc} with Proposition~\ref{concprod} one
obtains 
$$\CA((\Omega_1 \concat \Omega_2)^n) =
 \CA(\Omega_1^n \concat \Omega_2^n) \le   \min \{\CA(\Omega_1^n), \CA(\Omega_2^n)\} 
= n \cdot \min \left\{\CA(\Omega_1^n)/n, \CA(\Omega_2^n)/n\right\}.$$
Therefore
\begin{eqnarray*}
\sup \left\{ \frac{\CA((\Omega_1 \concat \Omega_2)^n)}{n} : n \in \N_{\ge 1}\right\}
&\le&  \sup \left\{ \min \left\{\frac{\CA(\Omega_1^n)}{n}, \frac{\CA(\Omega_2^n)}{n}\right\} 
: n\in \N_{\ge 1}\right\} \\
&\le& \min \left\{ \sup \left\{\frac{\CA(\Omega_1^n)}{n} : n \in \N_{\ge 1} \right\}, 
\sup\left\{ \frac{\CA(\Omega_2^n)}{n} : n \in \N_{\ge 1}\right\} \right\}\\
&=& \min\{\CAz(\Omega_1), \CAz(\Omega_2)\},
\end{eqnarray*} 
as claimed.
\end{proof}

\subsection{Union of Channels} \label{subunion}
We now define
the union of a family of channels having the same input and output alphabets. 
This channel operation will be used later in 
Sections~\ref{sechammingtype} and 
\ref{secportinglemma}
to study compound adversarial models.

\begin{definition} \label{defunionc}
Let $\{\Omega_i\}_{i \in I}$ be a family of channels, where $I$ is a finite index set and $\Omega_i: \mX \dto \mY$
for all $i \in I$. The \textbf{union} of the family $\{\Omega_i\}_{i \in I}$ is the channel 
denoted as and defined by
$$\bigcup_{i \in I} \Omega_i : \mX \dto \mY, \quad \left( \bigcup_{i \in I} \Omega_i \right)(x):=
\bigcup_{i \in I} \Omega_i(x) \text{ for all } x \in \mX.$$
\end{definition}

One can check that 
every channel can be written as the union of deterministic channels.
Moreover, union and concatenation relate to each other as follows. 
The proof can be found in Appendix~\ref{Appe}.

\begin{proposition}\label{proprunionfamily}
Let $\{\Omega_i\}_{i \in I}$ be as in Definition~\ref{defunionc}.
 Let $\Omega_1$, $\Omega_2$ be channels for which the concatenation
$\Omega_1 \concat \Omega_i \concat \Omega_2$ is defined for all $i \in I$ 
(we do not require $\{1,2\} \cap I = \emptyset$). Then
$$\bigcup_{i \in I}\,(\Omega_1 \concat \Omega_i \concat \Omega_2) = \Omega_1 \concat 
\left( \bigcup_{i \in I} \Omega_i \right) \concat \Omega_2.$$
\end{proposition}

%%%%%%%%%%%%%%%%%%%%%%%%%%%%%%%%%%%

\section{Hamming-Type Channels}
\label{sechammingtype}

In this section we study channels whose input alphabet 
is of the form $\mA^s$, where 
$\mA$ is a finite set  with $|\mA| \ge 2$ (the \textbf{alphabet}), and $s \in \N_{\ge 1}$.
We call these channels ``Hamming-type channels''.
In the sequel we work with a fixed alphabet $\mA$ and a fixed integer $s$. If $u \ge 1$ is any integer,  we write $[u]$ for the set $\{1,...,u\}$.
We also define the \textbf{extended alphabet} $\hat{\mA}:=\mA \cup \{\star\}$, where $\star \notin \mA$ is a symbol denoting an erasure.

\begin{notation} \label{a-capacity}
It is convenient for Hamming-type channels to express the capacity as a logarithm in base 
$|\mA|$, rather than in base 2. In this paper, if $\Omega: \mA^s \dto \mY$ is a channel, where $\mY$ is any output alphabet, we abuse notation and write 
$\CA(\Omega)$ for $\CA(\Omega) \cdot \log_2|\mA|$, and 
$\CAz(\Omega)$ for $\CAz(\Omega) \cdot \log_2|\mA|$.

The components
of a vector $x \in \mA^s$ are denoted by $x=(x_1,...,x_s)$.
If $n \ge 1$ is an integer, we denote by
$(x^1,...,x^n)$ a generic element of $(\mA^s)^n$.
\end{notation}

\subsection{Error-and-Erasure Adversaries}

We start by recalling some well known concepts from classical Coding Theory.

\begin{notation} \label{notazbeta}
If $u \ge 1$ is an integer, then the \textbf{minimum} (\textbf{Hamming}) \textbf{distance} of a set 
$\mC \subseteq \mA^u$ with cardinality 
$|\mC| \ge 2$ is $\dH(\mC):= \min\{\dH(x,x') : x,x' \in \mC, \, x \neq x'\}$.

For all $1 \le d \le u$, we set 
$\beta(\mA,u,d):=0$ if there is no $\mC \subseteq \mA^u$ with $|\mC| \ge 2$ and $\dH(\mC) \ge d$, and $\beta(\mA,u,d):= \max \{ \log_{|\mA|}|\mC| : \mC \subseteq  \mA^u, \, |\mC| \ge 2, \, \dH(\mC) \ge d \}$ otherwise.
\end{notation}

We will also need the following definitions.

\begin{definition} \label{defomegastar}
Let $U \subseteq [s]$ be a set. The $U$-\textbf{discrepancy} between vectors $y \in \hat{\mA}^s$ and  $x \in \mA^s$ is the integer $\delta(y,x;U):=|\{i \in U  : y_i \in \mA \mbox{ and }  y_i \neq x_i  \}|$. The $U$-\textbf{erasure weight} of a vector $y \in \hat{\mA}^s$ is
$\omega_\star(y;U):=|\{i \in U : y_i = \star\}|$.
\end{definition}

Evidently, if $x \in \mA^s$ is sent and $y \in \hat{\mA}^s$ is received, the $U$-discrepancy measures the number of errors that
occurred in positions indexed by $U$, while the $U$-erasure weight
measures the number of erasures that occurred in those positions.

We now describe an adversary having access to a certain set of coordinates $U \subseteq [s]$,
and  with limited error and erasure power.

\begin{definition}\label{elementarychannels}
Let $U \subseteq [s]$ and $t, e \ge 0$ be integers.
The channel $\HP{t}{e}{U} : \mA^s \dto \hat{\mA}^s$ is defined by
$$\HP{t}{e}{U} (x) := \{ y \in \hat{\mA}^s  :   y_i=x_i \mbox{ for } i \notin U, \,\delta(y,x;U) \le t, \,\omega_\star(y;U) \le e\} \text{ for all } x \in \mA^s.$$
\end{definition}

Note that $\HP{t}{e}{U}$ models the scenario where an adversary can erase up to $e$ 
components with indices from the set $U$, and change up to $t$ such components into different symbols from $\mA$.

It is well known from classical Coding Theory that 
 $\CA(\HP{t}{e}{[s]}) = \beta(\mA, s, 2t+e+1)$. 
 We now generalize this upper bound to the case where the adversary can only operate on a 
 subset $U \subseteq [s]$.
 
 \begin{proposition}\label{boundext}
 Let $U \subseteq [s]$ be a set of cardinality $u:=|U|$, and let $t, e \ge 0$ be integers. Then
 $$\CA(\HP{t}{e}{U}) = s-u+\beta(\mA, u, 2t+e+1).$$
 \end{proposition} 
 \begin{proof}
 The result is immediate if $U= \emptyset$ or $U=[s]$. Assume 
 $0< |U| < s$. Denote by $\pi:\mA^s \to \mA^{s-u}$ the projection on the coordinates 
 outside $U$, and let $\mC \subseteq \mA^s$ be a capacity-achieving good code for
 $\HP{t}{e}{U}$. By restricting the domain and the codomain of $\pi$,
 we obtain a surjective map
 $\overline{\pi}: \mC \to \pi(\mC) \subseteq \mA^{s-u}$. It is easy to see that
 for all $z \in \pi(\mC)$ we have $\log_{|\mA|}|\overline{\pi}^{-1}(z)| \le \beta(\mA,u,2t+e+1)$.
 We can write $\mC$ as a disjoint union  
 $\mC= \bigsqcup_{z \in \pi(\mC)} \overline{\pi}^{-1}(z),$
 from which we see that  $|\mC| \le |\pi(\mC)| \cdot |\mA|^\beta$, where $\beta:=\beta(\mA,u,2t+e+1)$. Therefore
 \begin{equation}
 \log_{|\mA|}|\mC| \le \log_{|\mA|}|\pi(\mC)| + \beta(\mA,u,2t+e) \le
 s-u+\beta(\mA,u,2t+e). \label{uppbnd}
 \end{equation} Finally, the upper bound in (\ref{uppbnd}) is achieved by any code of the form
 $\mC=\mD \times \mA^{s-u}$, where $\mD \subseteq \mA^u$ is a code with $|\mD|=1$ if $u \le 2t+e$, 
 and a code with $\log_{|\mA|}|\mD|=\beta(\mA,s,u)$ and minimum distance at least $2t+e+1$ otherwise.
 Note that not all codes $\mC$ that achieve (\ref{uppbnd}) with equality are of this form.  
 \end{proof}
 
  A more general scenario consists of multiple adversaries acting on pairwise disjoint 
  sets of coordinates with different powers. Such a collection of adversaries can be 
  described as follows.

 \begin{definition} \label{defhat}
 Let $L \ge 1$ and $U_1,...,U_L \subseteq [s]$ be pairwise disjoint subsets. Let $t_1,...,t_L$
and $e_1,...,e_L$  be non-negative integers.
Set  $\bm{U}:=(U_1,...,U_L)$, $\bm{t}:=(t_1,...,t_L)$ and $\bm{e}:=(e_1,...,e_L)$. The channel
$\HB{t}{e}{U}: \mA^s \dto \hat{\mA}^s$ is defined, for all $x \in \mA^s$, by
$$\HB{t}{e}{U}(x) := \left\{ y \in \hat{\mA}^s :   y_i=x_i \mbox{ for } i 
\notin \bigcup_{\ell=1}^L U_\ell, \, \delta(y,x;U_\ell) \le t_\ell, \,
\omega_\star(y;U_\ell) \le e_\ell \mbox{ for all $1 \le \ell \le L$}\right\}.$$
 \end{definition}
 
 Definition~\ref{defhat} models a channel with $L$ adversaries associated with 
 pairwise disjoint sets of coordinates $U_1,...,U_L \subseteq [s]$. 
 The $\ell$-th adversary can erase up to $e_\ell$ components of a 
 vector $x \in \mA^s$, and change up to $t_\ell$ such components of $x$ 
 into different symbols from $\mA$.

\begin{remark}
The channel $\HB{t}{e}{U}$ is isomorphic, in the sense of Definition~\ref{defisoch}, to 
a product of channels. To see this, assume without loss of generality that 
$U_\ell \neq \emptyset$ for all $1 \le \ell \le L$. Let $u_\ell:=|U_\ell|$ for all $1 \le \ell \le L$, and denote by
$\pi_\ell:\hat\mA^s \to \hat\mA^{u_\ell}$ the projection on the coordinates in $U_\ell$. 
For any integer $r \ge 1$, we let $\Id_r:\mA^r \dto \hat\mA^r$ be the identity channel (see Example~\ref{exid}). 
For $1 \le \ell \le L$, define the channel $\restHP{t_\ell}{e_\ell}{U_\ell}: \mA^{u_\ell} \dto  \hat\mA^{u_\ell}$ by  
$\restHP{t_\ell}{e_\ell}{U_\ell} (x) := \{ \pi_\ell(y) : y \in \HP{t_\ell}{e_\ell}{U_\ell}(x)\}$, for all $x \in \mA^s$. Finally,
let $u:=|u_1| + \cdots + |u_L|$. One can directly check that $\HB{t}{e}{U} \cong 
\restHP{t_1}{e_1}{U_1} \times \cdots \times \restHP{t_L}{e_L}{U_L}$ if
$u=s$, and  that $\HB{t}{e}{U} \cong \restHP{t_1}{e_1}{U_1} \times 
\cdots \times \restHP{t_L}{e_L}{U_L}\times \Id_{s-u}$ if $u<s$.
\end{remark}

\begin{remark}
The upper bound of Proposition~\ref{boundext} does not extend additively to multiple adversaries.
More precisely, if $L$,  $\bm{U}$, $\bm{t}$ and $\bm{e}$
are as in Definition~\ref{defhat}, then in general
$$\CA (\HB{t}{e}{U}) \not\le s - \sum_{\ell=1}^L|U_\ell| + 
\sum_{\ell=1}^L \beta(\mA,u_\ell,2t_\ell+e_\ell),$$
as the following example shows. This reflects the fact that $L$ adversaries
acting on different sets $U_\ell$'s, with error and erasure powers $t_\ell$'s and $e_\ell$'s, are not equivalent to a single adversary acting on the set $U_1 \cup \cdots \cup U_L$ with error and erasure powers
$t_1+ \cdots + t_\ell$ and $e_1+ \cdots + e_\ell$, respectively.
In particular, upper bounds from classical Coding Theory do not 
extend in any obvious way from one to multiple adversaries. 
%However, as we will see in Theorem \ref{mainthhamming}, the Singleton bound does extend additively to
%multiple adversaries, and in the more general context of compound models.
\end{remark}

\begin{example}
Take $\mA:=\F_2$, $L:=2$, $s:=8$, $U_1:=\{1,2,3,4\}$, $U_2:=\{5,6,7,8\}$, $t_1:=t_2:=1$ and
$e_1:=e_2:=0$. Let $\HH$ denote the channel of Example~\ref{firstexample}.
Then $\beta(\F_2,4,3)=\CA(\HH)=1$. By definition of $\HH$, we have
$\HB{t}{e}{U} \cong \HH \times \HH$. Therefore Example~\ref{fee} implies
\begin{equation*}\CA(\HB{t}{e}{U}) \ge \log_2 5 > 2= s - \sum_{\ell=1}^L|U_\ell| + \sum_{\ell=1}^L \beta(\mA,u_\ell,2t_\ell+e_\ell). \qedhere
\end{equation*}
\end{example}

\subsection{Compound Channels} \label{subsectioncompound}

In Section~\ref{firstsection} we introduced the concepts 
of (one-shot) capacity and zero-error capacity for general adversarial channels. We now define a third 
notion of capacity for the specific class of Hamming-type channels, which we call
the compound zero-error capacity.

\begin{notation}\label{notaztuple}
Let $L \ge 1$ be an integer, and let $V_\ell \subseteq [s]$, for 
 $1 \le \ell \le L$, be sets. If $\bm{V}=(V_1,...,V_L)$, we define 
$|\bm{V}|:=(|V_1|,...,|V_L|) \in \N^L$. If $U_\ell \subseteq [s]$ for all $1 \le \ell \le L$ and 
$\bm{U}=(U_1,...,U_\ell)$, then we write $\bm{V} \subseteq \bm{U}$
when $V_\ell \subseteq U_\ell$ for all $1 \le \ell \le L$.
Given integer vectors $\bm{v}=(v_1,...,v_L), \bm{u}=(u_1,...,u_L) \in \N^L$, we write 
$\bm{v} \le \bm{u}$ if $v_\ell \le u_\ell$ for all $1 \le \ell \le L$.
\end{notation}

 Let $L \ge 1$ and $U_1,...,U_L \subseteq [s]$ be pairwise disjoint subsets. Let $t_1,...,t_L$
and $e_1,...,e_L$  be non-negative integers.
Define the $L$-tuples  $\bm{U}:=(U_1,...,U_L)$, $\bm{t}:=(t_1,...,t_L)$ and $\bm{e}:=(e_1,...,e_L)$.

We consider the situation where a channel 
$\HB{t}{e}{U}$ is used $n$ times, but the adversaries are forced
to act on the same sets of coordinates in every channel use.
By this we mean that the $\ell$-th adversary freely chooses a set 
of 
components $V_\ell \subseteq U_\ell$ whose size does not exceed $t_\ell+e_\ell$, and 
operates on those components in each channel use, according to 
its error/erasure power.
The set of vulnerable components is unknown to the source, 
and no feedback is allowed.
This scenario can be mathematically modeled via the channel union operation (see Subsection~\ref{subunion}) as follows.

\begin{definition}\label{defcompound}
Let $L \ge 1$,  $\bm{U}:=(U_1,...,U_L)$, $\bm{t}:=(t_1,...,t_L)$ and $\bm{e}:=(e_1,...,e_L)$
be as in Definition~\ref{defhat}.
For $n \ge 1$, the \textbf{compound channel}
$\HB{t}{e}{U}^{n,\rest}: (\mA^s)^n \dto (\hat\mA^s)^n$ is defined by 
$$\HB{t}{e}{U}^{n,\rest} := \bigcup_{\substack{ \bm{V} \subseteq \bm{U} \\
|\bm{V}| \le\bm{t}+\bm e}}  \,\underbrace{\HB{t}{e}{V} \times \cdots 
\times \HB{t}{e}{V}}_{n \text{ times}} =
\bigcup_{\substack{ \bm{V} \subseteq \bm{U} \\
|\bm{V}| \le \bm{t}+\bm e}} 
\HB{t}{e}{V}^n.$$ 
The 
\textbf{compound zero-error capacity} of the channel 
 $\HB{t}{e}{U}$ is the number 
$$\CAzr (\HB{t}{e}{U}) :=
\sup\left\{\CA (\HB{t}{e}{U}^{n,\rest})/n : n \in \N_{\ge 1}\right\}.$$
\end{definition}

Note that in the definition of compound channel the union is taken over all
the $\bm{V} \subseteq \bm{U}$ with $|\bm{V}| \le\bm{t}+\bm {e}$, and not with
$|\bm{V}| =\bm{t}+\bm {e}$. This is because, by Definition~\ref{defhat}, we allow 
$t_\ell+e_\ell > |U_\ell|$ for some $\ell$ (so there may be no 
$\bm{V} \subseteq \bm{U}$ with $|\bm{V}| =\bm{t}+\bm {e}$). 
This choice may seem unnatural at this point, but it will simplify the discussion
on Hamming-type channels induced by network adversaries
in Subsection~\ref{subportinglemmas}.

The following proposition shows how one-shot capacity, 
zero-error capacity and compound zero-error capacity relate to each other.
The proof is left to the reader.

\begin{proposition}\label{relationcapacities}
Let $L \ge 1$,  $\bm{U}:=(U_1,...,U_L)$, 
$\bm{t}:=(t_1,...,t_L)$ and $\bm{e}:=(e_1,...,e_L)$
be as in Definition~\ref{defhat}.
For all $n \ge 1$ we have 
\begin{equation*}\label{eqrelationcapacities}
n \cdot \CA(\HB{t}{e}{U})  \le   
\CA(\HB{t}{e}{U}^n)  \le   
\CA(\HB{t}{e}{U}^{n,\rest}).
\end{equation*}
In particular,
$$\CA(\HB{t}{e}{U})  \le   
\CAz(\HB{t}{e}{U})  \le   
\CAzr(\HB{t}{e}{U}).$$
\end{proposition}
%\begin{proof}
%By Proposition~\ref{relate}, it suffices to show 
%the latter inequality of (\ref{eqrelationcapacities}).
%We will prove that
%for all $n \ge 1$ we have 
%$\HB{t}{e}{U}^{n,\rest} \le 
% \HB{t}{e}{U}^{n}$, i.e., that the compound channel is finer than the product channel.
%As observed in Definition~\ref{deffiner}, this implies the desired upper bound.
%
%
%Let $(x^1,...,x^n) \in (\mA^s)^n$ be arbitrary, and take any
%$(y^1,...,y^n) \in \HB{t}{e}{U}^{n,\rest}(x^1,...,x^n)$.
%Then by definition of $\HB{t}{e}{U}^{n,\rest}$ there exists $\bm{V}=(V_1,...,V_L)$
%with the following properties:
%$${V}_\ell \subseteq U_\ell \mbox{ for all $1 \le \ell \le L$}, 
%|V_\ell| \le t_\ell+e_\ell 
%\mbox{ for all $1 \le \ell \le L$},
%(y^1,...,y^n) \in \HB{t}{e}{V}^n(x^1,...,x^n).$$ 
%Therefore $y^k \in\HB{t}{e}{V}(x^k)$ for all
%$1 \le k \le n$. In particular, $y^k \in\HB{t}{e}{U}
%(x^k)$ for all $1 \le k \le n$. As a consequence, 
%$(y^1,...,y^n) \in \HB{t}{e}{U}^n
%(x^1,...,x^n)$.
%\end{proof}

\subsection{Capacities of Hamming-Type Channels}
The goal of this subsection is to establish the following general theorem
on the capacities of a Hamming-type channel of the form $\HB{t}{e}{U}$.

\begin{theorem} \label{mainthhamming}
Let $L \ge 1$,  $\bm{U}:=(U_1,...,U_L)$, $\bm{t}:=(t_1,...,t_L)$ and $\bm{e}:=(e_1,...,e_L)$
be as in Definition~\ref{defhat}.
For all $n \ge 1$ we have
\begin{equation}\label{eqmainthhamming}
n \cdot \CA(\HB{t}{e}{U})  \le  
\CA(\HB{t}{e}{U}^n)  \le   
\CA(\HB{t}{e}{U}^{n,\rest})  \le    
n \left(s-\sum_{\ell=1}^L \min\{2t_\ell+e_\ell, |U_\ell|\} \right).
\end{equation}
In particular,
$$\CA(\HB{t}{e}{U})  \le  
\CAz(\HB{t}{e}{U})  \le   
\CAzr(\HB{t}{e}{U})  \le    
s-\sum_{\ell=1}^L \min\{2t_\ell+e_\ell, |U_\ell|\}.$$
Moreover, all the above 
inequalities are achieved with equality if $\mA=\F_q$ and $q$ is sufficiently large.
\end{theorem}

We will need the following preliminary result, whose proof can be found in Appendix~\ref{Appe}.

\begin{lemma} \label{keylong}
Let $L \ge 1$,  $\bm{U}:=(U_1,...,U_L)$, $\bm{t}:=(t_1,...,t_L)$ and 
$\bm{e}:=(e_1,...,e_L)$
be as in Definition~\ref{defhat}. Define $\sigma_\ell:=
\min\{2t_\ell+e_\ell, |U_\ell|\}$ for all $1 \le \ell \le L$,
and $\sigma:=\sigma_1+\cdots +\sigma_L \le s$.

\begin{enumerate} \setlength\itemsep{0em}
\item \label{keylongitem1}There exist sets 
$\overline{U}_\ell^{1}, \overline{U}_\ell^{2}, 
\overline{U}_\ell^{\star} \subseteq U_\ell$, for 
$1 \le \ell \le L$, with the following properties:
$$|\overline{U}_\ell^{1}|,  |\overline{U}_\ell^{2}|  \le t_\ell, \quad
|\overline{U}_\ell^{\star}| \le e_\ell, \quad
 |\overline{U}_\ell^{1} \cup \overline{U}_\ell^{2} \cup 
 \overline{U}_\ell^{\star}|=\sigma_\ell \mbox{ for all $1 \le \ell \le L$}.$$

\item  \label{keylongitem2}We have $\sigma=|\overline{U}|$, 
where $$\overline{U}:=\left|  \left(  \bigcup_{\ell=1}^L 
\overline{U}_\ell^{1} \right)  \cup   \left(  \bigcup_{\ell=1}^L \overline{U}_\ell^{2} \right)
 \cup  \left(  \bigcup_{\ell=1}^L \overline{U}_\ell^{\star} \right)  \right| \subseteq [s].$$
 
 \item  \label{keylongitem3}There exist $\bm{V}=(V_1,...,V_L) \subseteq \bm{U}$
and $\bm{V'}=(V'_1,...,V'_L) \subseteq \bm{U}$
with the following properties:
\begin{itemize}\setlength\itemsep{0em}
\item $|\bm{V}|, |\bm{V'}| \le \bm{t} + \bm e$;
\item for any $x,x' \in \mA^s$,  if $x_i=x'_i$ for all $i \in [s] \setminus \overline{U}$, 
then
$\HB{t}{e}{V}   (x) \cap \HB{t}{e}{V'}  (x') \neq \emptyset$.
\end{itemize}
\end{enumerate}
\end{lemma}

We can now prove the main result of this section.

\begin{proof}[Proof of Theorem~\ref{mainthhamming}]
Let $\sigma_\ell:=\min\{2t_\ell+e_\ell, |U_\ell|\}$ for all $1 \le \ell \le L$,
and $\sigma:=\sigma_1+\cdots +\sigma_L$. We only show
the theorem for $\sigma <s$. The case $\sigma=s$ is similar and in fact easier.
By Proposition~\ref{relationcapacities}, it suffices to show that
$\CA(\HB{t}{e}{U} ^{n,\rest})  \le    
n(s-\sigma)$ for all $n \ge 1$, and that for 
$\mA=\F_q$ and sufficiently large $q$ we have $\CA(\HB{t}{e}{U} ) \ge    
s-\sigma$.

Let
$\overline{U}_\ell^{1}$, $\overline{U}_\ell^{2}$, $\overline{U}_\ell^{\star}$ 
(for $1 \le \ell \le L$) and $\overline{U}$ be as in Lemma
\ref{keylong}.
We have 
$\sigma=|\overline{U}|$. Denote by $\pi:\mA^s \to
 \mA^{s-\sigma}$ the projection on the coordinates outside $\overline{U}$.
Since $\sigma < s$, the map $\pi$ is well defined.

Let $n \ge 1$ be an integer. Then $\pi$ extends component-wise to a map
$\Pi:(\mA^s)^n \to (\mA^{s-\sigma})^n$. Let $\mC\subseteq (\mA^s)^n$ 
be a capacity-achieving good code for $\HB{t}{e}{U} ^{n,\rest}$.
To obtain the upper bound, it suffices to show that 
the restriction of $\Pi$ to $\mC$ is injective. 

Take $x,x' \in \mC$, and assume
$\Pi(x)=\Pi(x')$. We will show that $x=x'$. Write $x=(x^1,...,x^n)$ and $x'=(x'^1,...,x'^n)$. 
By definition of $\Pi$, we have $\pi(x^k)=\pi(x'^k)$ for all $1 \le k \le n$.
By Lemma~\ref{keylong}, there exist
$L$-tuples of sets
$\bm{V}=(V_1,...,V_L) \subseteq \bm{U}$ and
 $\bm{V'}=(V'_1,...,V'_L) \subseteq \bm{U}$ 
such that:
\begin{equation} \label{nonvuoto}
|\bm{V}|, |\bm{V'}| \le \bm{t}+\bm e \quad \mbox{ and } \quad
\HB{t}{e}{V}  (x^k) \cap \HB{t}{e}{V'}  (x'^k) \neq \emptyset \quad \mbox{
for all $1 \le k \le n$}.
\end{equation}
By definition of $\HB{t}{e}{U}^{n,\rest}$ we have
\begin{eqnarray*}
\HB{t}{e}{U}^{n,\rest}(x) \cap \HB{t}{e}{U}^{n,\rest}(x') 
&\supseteq& \HB{t}{e}{V}^n(x^1,...,x^n) \cap \HB{t}{e}{V'}^n (x'^1,...,x'^n) \\
&=& \prod_{k=1}^n \left( \HB{t}{e}{V} (x^k) \cap \HB{t}{e}{V'}(x'^k) \right) \neq \emptyset,
\end{eqnarray*}
where the last inequality follows from (\ref{nonvuoto}).
Since $\mC$ is good for $\HB{t}{e}{U}^{n,\rest}$,
we conclude $x=x'$. This shows that the restriction of $\Pi$ to $\mC$ is injective, as desired.

We now prove that the upper bounds in the theorem are tight for $\mA=\F_q$ and $q$ 
sufficiently large.
As already stated at the beginning of the proof, it suffices to show that $\CA(\HB{t}{e}{U}) \ge    
s-\sigma$.

Let
 $\mC \subseteq \mA^s$ be any code with minimum distance $\dH(\mC)=\sigma+1$
and cardinality $q^{s-\sigma}$. We will prove that $\mC$ is 
good for $\HB{t}{e}{U}$. Let $x,x' \in \mC$ with $x \neq x'$, and assume by way 
of contradiction that there exists 
$z \in \HB{t}{e}{U}(x)
\cap \HB{t}{e}{U}(x')$.
For $1 \le \ell \le L$ construct  the sets
$$U_\ell^{1}:=\{i \in U_\ell : z_i \in \mA, \,z_i \neq x_i\},
\quad  U_\ell^{2}:=\{i \in U_\ell : z_i \in \mA, \,z_i \neq x'_i\}, 
\quad  U_\ell^{\star}:=\{i \in U_\ell : z_i =\star\}.$$
By definition of $\HB{t}{e}{U}$, for all $1 \le \ell \le L$ we have
$|U_\ell^{1} \cup U_\ell^{2} \cup U_\ell^{\star}| \le 
\min\{2t_\ell+e_\ell, |U_\ell|\} = \sigma_\ell$. 
Note moreover that
$U_\ell^{1} \cup U_\ell^{2} \cup
U_\ell^{\star} =\{ i \in U_\ell : 
z_i \neq x_i \mbox{ or } z_i \neq x'_i \}$ for all $1 \le \ell \le L$.
Since the sets $U_\ell$'s are pairwise disjoint, we have 
$$|\{1 \le i \le s : z_i=x_i=x'_i\}| = s- \sum_{\ell=1}^L |U_\ell^{1} \cup 
U_\ell^{1} \cup U_\ell^{\star}| \ge s-
\sum_{\ell=1}^L \sigma_\ell=s-\sigma.$$
Thus $z$, $x$ and $x'$ coincide in at least $s-\sigma$ components, 
and therefore $\dH(x,x') \le \sigma$, a contradiction. 
\end{proof}

\subsection{Hamming-Type Channels over Product Alphabets}

We now consider channels of the form $(\mB^m)^s \dto (\hat\mB^m)^s$,  where $\mB$ 
is a finite set 
with $|\mB| \ge 2$, $\hat\mB=\mB \cup \{\star\}$, $\star \notin \mB$, $m \in \N_{\ge 2}$ 
and $s \in \N_{\ge 1}$.
 Therefore the input alphabet is $\mA=\mB^m$.
For $x=(x_1,...,x_s) \in (\hat\mB^m)^s$ and $i \in [m]$, we denote by 
$x_i=(x_{i,1},...,x_{i,m})$ the 
sub-components of $x_i$.

\begin{definition}
For $y \in \hat{\mB}^m$ and $x \in \mB^m$, we let $\delta(y,x):=|\{i \in [m]  : y_i \in \mB \mbox{ and }  y_i \neq x_i  \}|$ be the \textbf{discrepancy} between $y$ and $x$.
 The \textbf{erasure weight} of $y \in \hat{\mB}^m$ is
$\omega_\star(y):=|\{i \in [m] : y_i = \star\}|$.
\end{definition}

Consider an adversary who has access to all the $s$ components 
of $x=(x_1,...,x_s) \in (\mB^m)^s$. For each component $x_i \in \mB^m$, 
the adversary can corrupt up to $t$ sub-components  of $x_i$, and erase at most $e$ of them.  
This scenario is modeled as follows.

\begin{definition}
Let $t,e \ge 0$ be integers. The channel $\HP{t}{e}{\mB,m,s}:(\mB^m)^s \dto (\hat\mB^m)^s$
is defined, for all $x=(x_1,...,x_s) \in (\mB^m)^s$, by 
$$\HP{t}{e}{\mB,m,s}(x):= \{y \in (\hat\mB^m)^s : \delta(y_i,x_i) \le t \mbox{ and } \omega_\star(y_i) \le e \mbox{ for all $1 \le i \le s$}\}.$$
\end{definition}

We conclude this section computing the capacity and the zero-error capacity
of a channel of the form $\HP{t}{e}{\mB,m,s}$. In the following 
theorem capacities are expressed as logarithms in base $|\mA|=|\mB|^m$.

\begin{theorem} \label{productalphabetthm}
Let $t,e \ge 0$ be integers. For all $n \ge 1$ we have
\begin{equation*}
n \cdot \CA(\HP{t}{e}{\mB,m,s})  \le  
\CA(\HP{t}{e}{\mB,m,s}^n) \le 
n \cdot s/m \cdot \max\{0,m-2t-e\}.
\end{equation*}
In particular,
$$\CA(\HP{t}{e}{\mB,m,s})  \le  
\CAz(\HP{t}{e}{\mB,m,s})  \le   
s/m \cdot  \max\{0,m-2t-e\}.$$
Moreover, all the above 
inequalities are achieved with equality if $\mA=\F_q$ and $q$ is sufficiently large.
\end{theorem}

\begin{proof}
It follows from the definitions that the channel $\HP{t}{e}{\mB,m,s}$ coincides with the 
$s$-th power of the Hamming-type channel $\HP{t}{e}{[m]}: \mB^m \dto \hat{\mB}^m$ 
of Definition~\ref{elementarychannels}
 (over the alphabet $\mB$ and with input symbols from $\mB^m$).
Similarly, the $n$-th power of  $\HP{t}{e}{\mB,m,s}$ coincides with the 
$ns$-th power of $\HP{t}{e}{[m]}: \mB^m \dto \hat{\mB}^m$.
Therefore the desired statement follows from Theorem~\ref{mainthhamming}.
\end{proof}

%%%%%%%%%%%%%%%%%%%%%%%%%%%%%%%%%%%

\section{Networks, Adversaries, and Capacity Regions} \label{secnetwork}

In this section we start the analysis of networks by 
studying combinational networks, network alphabets,
network codes, and adversaries. We show that these objects 
naturally induce families of adversarial channels, which determine the capacity 
regions of multi-source network under various adversarial models.

\subsection{Combinational Networks} \label{subnetworks}

In the sequel, $\mA$ denotes a finite set with cardinality $|\mA| \ge 2$, which we call the \textbf{network alphabet}. 
A network
is defined as follows.

\begin{definition} \label{defnetwork}
A (\textbf{combinational}) \textbf{network} 
 is a 4-tuple $\mN=(\mV,\mE,{\bf S},{\bf T})$ where:
\begin{enumerate}[label=(\Alph*)]\setlength\itemsep{0em} 
 \item $(\mV,\mE)$ is a finite directed acyclic multigraph,
\item ${\bf S} \subseteq \mV$ is the set of \textbf{sources},
\item ${\bf T} \subseteq \mV$ is the set of \textbf{terminals} or \textbf{sinks}.
\end{enumerate}
Note that we allow multiple parallel directed edges. We also assume that the following hold.
\begin{enumerate}[label=(\Alph*)]\setlength\itemsep{0em} \setcounter{enumi}{3}
 \item $|{\bf S}| \ge 1$, $|{\bf T}| \ge 1$, ${\bf S} \cap {\bf T} = \emptyset$.
\item \label{prnE} For any $S \in {\bf S}$ and $T \in {\bf T}$ there exists a directed path from $S$ to $T$.
\item Sources do not have incoming edges, and terminals do not have outgoing edges. \label{prnF}
\item For every vertex $V \in \mV \setminus ({\bf S} \cup {\bf T})$ there exists a directed path from
$S$ to $V$ for some $S \in {\bf S}$, and a directed path from $V$ to $T$ for some $T \in {\bf T}$. \label{prnG}
\end{enumerate}
The elements of $\mV$ are called \textbf{vertices}. The elements of 
$\mV \setminus ({\bf S} \cup {\bf T})$ are the \textbf{intermediate} vertices.
We denote the set of incoming and outgoing edges of a $V \in \mV$ by 
$\inn(V)$ and $\out(V)$, respectively.
\end{definition}

We are interested in multicast problems over networks of type $\mN=(\mV,\mE,{\bf S},{\bf T})$. 
Our model will encompass the presence of one or multiple adversaries, capable of 
corrupting the values of some network edges, or erasing them, according to certain restrictions.

The sources attempt to transmit
information packets to all the terminals simultaneously, sharing the network resources. 
The packets are drawn from the network alphabet $\mA$. 
The intermediate vertices emit packets that belong to the same alphabet $\mA$, but collect over the incoming edges 
packets that belong a priori to the \textbf{extended network alphabet} $\hat{\mA}=\mA \cup \{\star\}$, 
where $\star \notin \mA$ is a symbol denoting an erasure.

\begin{remark}
The symbol $\star$ from the extended alphabet $\hat{\mA}$
should not be regarded as an ordinary alphabet symbol, but more
as an ``erasure warning'' symbol. This is the reason why we force the 
intermediate vertices to emit packets from  
$\mA$, rather than $\hat{\mA}$. If vertices were allowed to emit 
any symbol from $\hat{\mA}$, the whole theory would simply reduce to that of a 
network with a larger alphabet. 
\end{remark}
We assume that 
every edge of $\mN$ can carry precisely one element from $\mA$,
or possibly the  erasure symbol $\star$. For examples, if the sources transmit vectors of length $m$
with entries from a finite field $\F_q$, then $\mA=\F_q^m$, and 
each edge carries $m \log_2(q)$ bits (and possibly 
the symbol $\star$). 

In every network use, an intermediate vertex collects packets over the 
incoming edges, processes them,
and sends out packets over the outgoing edges. The outgoing packets are a 
\textit{function} of the incoming packets.
We assume that the intermediate vertices are \textit{memoryless},
and all network
transmissions are \textit{delay-free}.

Removing idle edges if necessary, we may also assume without loss of generality 
that every vertex sends information packets over every 
outgoing edge. Along with property~\ref{prnG} of Definition~\ref{defnetwork}, this  
ensures in particular that every non-source vertex 
receives a packet on each of its incoming edges.

\begin{notation}
In the sequel we work with a fixed combinational 
network $\mN=(\mV,\mE,{\bf S},{\bf T})$. We choose an enumeration for 
the sources of $\mN$, say ${\bf S}=\{S_1,...,S_N\}$, and 
let $I:=\{1,...,N\}$ be the set of source indices.
For $J \subseteq I$, we let ${\bf S}_J:=\{S_i : i \in J\}$.
For all $i \in I$, the \textbf{local input alphabet} of source $S_i \in 
{\bf S}$ is $\mX_i:= \mA^{|\out(S_i)|}$.
The \textbf{global input alphabet} of
 the network $\mN$ is
$\mX:= \mX_1 \times \cdots \times \mX_N$. 
\end{notation}

The topology of $\mN$ induces a partial order $\preceq$ on $\mE$
as follows. For 
 $e, e' \in \mE$, we have $e \preceq e'$ if and only if there exists 
a directed path in $\mN$ whose first edge is $e$ and whose last edge is $e'$.
Since $\mN$ is a directed and acyclic graph, $\preceq$ is well defined. Such partial 
order can be extended to a total order on $\mE$
(see~\cite[Section 22]{sort}). Throughout this paper, we fix such an order extension 
and denote it by $\le$. 

\begin{remark}
The choice of the linear order $\le$ will prevent ambiguities
in the interpretation of the objects associated with $\mN$.
None of the results in this paper depends on the specific choice of $\le$.
\end{remark}

\subsection{Network Codes and Network Channels}

As mentioned in Subsection~\ref{subnetworks}, each intermediate vertex 
of $\mN$ applies some function to the incoming packets, 
and transmits the value of this function on the outgoing edges.
The collection of operations that the intermediate vertices perform
is called the network code.

\begin{definition} \label{defK}
A \textbf{network code} $\mF$ for $\mN$ is a family of functions $\{\mF_V : 
V \in \mV \setminus ({\bf S} \cup {\bf T})\}$, where 
$$\mF_V: {\hat\mA}^{|\inn(V)|} \to \mA^{|\out(V)|} \quad \mbox{ for all $V \in \mV \setminus({\bf S} \cup {\bf T})$}.$$

Assume $\mA=\F_q^m$ for some prime power $q$ and some integer $m \ge 1$. Then we say that $\mF$ is 
\textbf{linear} if for all vertices $V \in \mV \setminus ({\bf S} \cup {\bf T})$
there exists a $|\inn(V)| \times |\out(V)|$ matrix $L_V$ over $\F_q$
such that $\mF_V(x_1,...,x_r) = (y_1,...,y_s)$
for all $(x_1,...,x_r) \in (\F_q^m)^r$, where
$r:= |\inn(V)|$, $s:=|\out(V)|$, and 
$$\begin{pmatrix} y_1^\top & \cdots & y_s^\top \end{pmatrix}:= \begin{pmatrix} x_1^\top & \cdots & x_r^\top \end{pmatrix} \cdot L_V.$$
In other words, we require that the restriction of $\mF_V: \hat{\mA}^r \to \mA^s$
to $\mA^r = (\F_q^m)^r$ is a function that returns linear combinations of the $r$
input vectors from $\F_q^m$. 
When $\mF$ is a linear network code and $V \in \mV \setminus ({\bf S} \cup {\bf T})$, we sometimes denote the matrix $L_V$ simply by $\mF_V$.
\end{definition}

A collection of functions as in Definition~\ref{defK} fully specifies the  
operations performed by the vertices of $\mN$ with no ambiguity, thanks to the 
choice of the linear order $\le$. Note that in this paper we do not require our network codes to be linear (the alphabet $\mA$ may not even be of the form $\mA=\F_q^m$).

In the remainder of the section we define a series of deterministic channels 
that are naturally induced by a network code 
$\mF$ for $\mN$.

\begin{definition} \label{def11}
Let $\mF$ be a network code for $\mN$, and let $\mE' \subseteq \mE$ be a non-empty set of edges.
In an error-free context, by definition of network code, the values of the edges from $\mE'$ can be 
expressed as function of the sources input $x \in \mX=\mX_1 \times \cdots \mX_N$. 
This defines a deterministic channel
$$\Omega_\mF[{\bf S} \to \mE'] : \prod_{i \in I} \mX_i \dto \mA^{|\mE'|},$$
that associates to $x \in \mX$ the alphabet packets observed over the edges in $\mE'$, ordered according to $\le$. 
Note that $\Omega_\mF[{\bf S} \to \mE'](x) \subseteq \mA^{|\mE'|}$ is a set of cardinality one for all $x \in \mX=\mX_1 \times \cdots \mX_N$.

Now assume that $T \in {\bf T}$, and that $\mE'$ is a cut that separates all sources $S_1,...,S_N$ from $T$. We will construct a deterministic channel that describes the transfer from the edges in $\mE'$ to the terminal $T$.

Recall from Poset Theory that $e \in \mE$ \textbf{covers} $e' \in \mE$ when $e' \precneq e$, and there is no $e'' \in \mE$ with $e' \precneq e'' \precneq e$. We recursively define the sets of edges
$$\mE'_0:=\inn(T), \quad  \mE'_{k+1}:=(\mE' \cap \mE'_k) \cup \{e' \in \mE :  
\mbox{ there exists  } e \in \mE'_k \setminus \mE' \mbox{ that covers } e' \}
\mbox{ for $k \in \N_{\ge 1}$}.$$
For all $k \in \N$, $\mE'_{k+1}$ is the union of $\mE' \cap \mE'_k$ and the set of network edges
that \textit{immediately precede} the edges in $\mE'_k \setminus \mE'$.
Let $\overline{k}:= \min \{ k \in \N: \mE'_k \subseteq \mE'\}$, where the minimum is well defined 
by property~\ref{prnG} of Definition~\ref{defnetwork} and the fact that $\mE'$ is a cut that separates $\bf S$ and $T$.

For all $0 \le k \le \overline{k}-1$, the values of the edges in $\mE'_k$ can be expressed, via the network code $\mF$, as
a function of the values of the edges in $\mE'_{k+1}$. This defines $\overline{k}$ functions
$\Psi_\mF[k]: \hat\mA^{|\mE_{k+1}|} \to  \mA^{|\mE_{k}|}$,
for $0 \le k \le \overline{k}-1$
 (see also the following Example~\ref{extransfer}). Note that we define $\Psi_\mF[k]$ to be the identity on the edges in the set $\mE'_k \cap \mE'_{k+1}$. Now the 
composition  
$\Psi:=\Psi_\mF[{\overline{k}-1}] \circ \cdots \circ \Psi_\mF[1] \circ \Psi_\mF[0]$
expresses the values of the edges of $\inn(T)$ as a function of the values of the 
edges from $\mE'_{\overline{k}} \subseteq \mE'$. We trivially 
extend $\Psi$ to a function
$$\Psi:\mA^{|\mE'|} \to \mA^{|\inn(T)|}$$ that expresses the values of the edges from $\inn(T)$ as a function of the values of the 
edges from the edge-cut $\mE'$ (by ``trivially'' we mean that the value of the edges from $\mE' \setminus \mE'_{\overline{k}}$ play no role in the definition of $\Psi$). This induces 
a deterministic channel denoted by $$\Omega_\mF[\mE' \to \inn(T)] : {\hat\mA}^{|\mE'|}  \dto \mA^{|\inn(T)|}.$$
\end{definition}

\begin{example} \label{extransfer}
Let $\mN$ be the network in Figure~\ref{figextransfer}. The edges
of $\mN$ are ordered according to their indices. The set
$\mE':=\{e_2,e_5\}$ is an edge-cut between ${\bf S}=\{S\}$ and $T$. 
The edges in $\mE'$ are represented by a 
dashed arrow in Figure~\ref{figextransfer}. Fix a network code $\mF$ for $\mN$.

\begin{figure}[htbp]
\centering
\begin{tikzpicture}[>=stealth,
vertex/.style = {shape=circle,draw,inner sep=0pt,minimum size=1.9em},
nnode/.style = {shape=circle,fill=myg,draw,inner sep=0pt,minimum size=1.9em}]

\node[vertex] (S)  {$S$};
\node[nnode,right=\mynodespace of S] (V1) {$V_1$};
\node[nnode,right=\mynodespace of V1] (V2) {$V_2$};
\node[nnode,right=\mynodespace of V2] (V3)  {$V_3$};
\node[vertex,right=\mynodespace of V3] (T)  {$T$};
\draw[->,bend left=20] (S)  to node[fill=white,inner sep=1pt]{\small $e_1$} (V2);
\draw[->,dashed] (S) to  node[fill=white,inner sep=1pt]{\small $e_2$} (V1);
\draw (V1) to  node[fill=white,inner sep=1pt]{\small $e_3$} (V2);
\draw[->,bend right=20] (V1) to node[fill=white,inner sep=1pt]{\small $e_4$} (V3);
\draw[->,dashed] (V2) to node[fill=white,inner sep=1pt]{\small $e_5$} (V3);
\draw[->,bend left=15] (V3) to node[fill=white,inner sep=1pt]{\small $e_6$} (T);
\draw[->,bend right=15] (V3) to node[fill=white,inner sep=1pt]{\small $e_7$} (T);
\end{tikzpicture} 
\caption{{{Network for Example~\ref{extransfer}.}}}\label{figextransfer}
\end{figure}
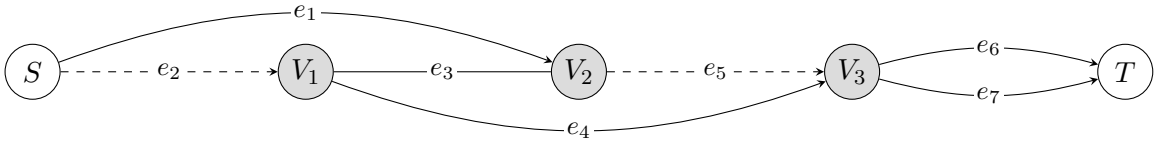

Following the notation of Definition~\ref{def11}, we have 
$\mE'_0=\inn(T)=\{e_6,e_7\}$, $\mE'_1=\{e_4,e_5\}$, and $\mE'_2=\{e_2,e_5\}=\mE'$.
Hence $\overline{k}=2$. The functions $\Psi_\mF[0]$ and $\Psi_\mF[1]$ are as follows:
$$
\Psi_\mF[0] (x_4,x_5)= \mF_{V_3}(x_4,x_5) \in \mA^2 \quad \mbox{ and } \quad \Psi_\mF[1] (x_2,x_5) = (\mF''_{V_1}(x_2),x_5) \in \mA^2$$
 for all $(x_4,x_5) \in \hat\mA^{|\mE'_1|} = \hat\mA^2$ and  $(x_1,x_5) \in 
\hat\mA^{|\mE'_2|} = \hat\mA^2$, and 
where $\mF''_{V_1}(x_2)$ is the second component of the vector $\mF_{V_1}(x_2) \in \mA^2$.
Now observe that $\left( \Psi_\mF[1] \circ \Psi_\mF[0] \right)(x_2,x_5)= \mF_{V_3} \left( \mF_{V_1}''(x_2),x_5 \right)  \in \mA^2$
for all $(x_2,x_5) \in  \hat\mA^2$,
which expresses the values of the edges of $\inn(T)$ as a function of the values of the edges of  $\mE'$. Therefore one has
$\Omega_\mF[\mE' \to \inn(T)](x_2,x_5)=\{\mF_{V_3} \left( \mF_{V_1}''(x_2),x_5 \right)\} \subseteq \mA^2$ for all $(x_2,x_5) \in 
\hat\mA^2$.
\end{example}

\begin{remark}\label{noantichain}
The deterministic channels introduced in Definition~\ref{def11} are well defined for any 
 edge-cut $\mE'$ that separates the network sources from $T$. In particular, we do not require $\mE'$ to be a minimum 
cut or an antichain cut (i.e., a cut where every two distinct edges are not comparable with respect to the order $\preceq$).
For example, the cut $\mE'$ of Example~\ref{extransfer} is not an antichain, 
as $e_2 \preceq e_5$. Observe moreover that whenever 
 $\mE'$ is a non-antichain cut,  
the channel $\Omega_\mF[\mE' \to \inn(T)]$ automatically gives ``priority'' to the edges in $\mE'$ that are 
``closer'' to the destination $T$ in the network topology.

The specific definition of the channel
$\Omega_\mF[\mE' \to \inn(T)]$ makes it so
that all the results of this paper (such as decomposition results and upper bounds) also hold for edge-cuts that do not form
an antichain with respect to $\preceq$. This is different, e.g., from the general approach of~\cite{YeungCai},
where some of the cut-set bounds are derived only for the special case of antichain edge-cuts (see
in particular~\cite[Theorem 3]{CaiYeung}).
\end{remark}

\begin{remark} \label{important1}
If $T \in {\bf T}$ and 
$\mE'$ is a cut between $\bf S$  and $T$, then the channel describing the transfer 
from $\bf S$ to $T$ 
factors through the cut $\mE'$, i.e., we have 
$$\Omega_\mF[{\bf S} \to \inn(T)]  = \Omega_\mF[
{\bf S} \to \mE']  \concat   \Omega_\mF[
\mE' \to \inn(T)].$$
Notice that for a non-antichain cut $\mE'$ the above channel decomposition does not simply follow
from the fact that $\mE'$ is a cut between $\bf S$ and $\inn(T)$, and 
heavily relies on the specific 
construction of the channel
$\Omega_\mF[
\mE' \to \inn(T)]$ given in Definition \ref{def11}.
\end{remark}

\begin{example}
Following the notation of Example~\ref{extransfer},  it is easy to see that 
for all $(x_1,x_2) \in \mA^2$
 we have 
$\Omega_\mF[{\bf S} \to \mE'] (x_1,x_2)=\{(x_2,\mF_{V_2}(x_1,\mF'_{V_1}(x_2)))\}$, 
where $\mF'_{V_1}(x_2)$ is the first component of $\mF_{V_1}(x_2) \in \mA^2$.
Using the expression for $\Omega_\mF[
\mE' \to \inn(T)]$ found in Example~\ref{extransfer} we obtain
$$(\Omega_\mF[
{\bf S} \to \mE']  \concat \Omega_\mF[
\mE' \to \inn(T)])\,(x_1,x_2) = \{\mF_{V_3}(\mF''_{V_1}(x_2), 
\mF_{V_2}(x_1,\mF'_{V_1}(x_2)))\} \mbox{ for all $(x_1,x_2) \in \mA^2$},$$
which expresses the values of the edges from $\inn(T)$ as a function of the values 
of the source edges.
\end{example}

We now examine channels that are obtained by ``freezing''
the packets emitted by some of the network sources. 
These deterministic channels
will play an important role in the study of the capacity region of adversarial networks
in Section~\ref{secportinglemma}. 

\begin{definition} \label{def22}
Let $\mF$ be a network code for $\mN$, and let $J \subsetneq I$ be
 a proper non-empty subset of source indices.
Fix an element $\overline{x} \in \prod_{i \in I \setminus J} \mX_i$, and 
let $\mE' \subseteq \mE$ be non-empty. We denote by 
$$\Omega_\mF^J[{\bf S}_J \to \mE' \mid \overline{x}]: \prod_{i \in J} \mX_i \dto \mA^{|\mE'|}$$
the deterministic channel that associates to an input $x \in \prod_{i \in J} \mX_i$
the values of the edges from $\mE'$ when
the sources in ${\bf S}_J$ emit the corresponding packets from $x \in \prod_{i \in J} \mX_i$, and 
the sources in ${\bf S}_{I \setminus J}$ emit the corresponding (fixed) packets 
from $\overline{x} \in \prod_{i \in I \setminus J} \mX_i$.

Let $T \in {\bf T}$ be a terminal, and let $\mE'$ be a cut that separates ${\bf S}_J$ from $T$. Reasoning as in Definition~\ref{def11}, it is easy to see that 
the values of the 
edges from $\inn(T)$ can be expressed, given the network code $\mF$, as a function of the 
values of the edges from $\mE'$ and by $\overline{x}$.
This defines a deterministic channel
 $$\Omega_\mF^J[\mE' \to \inn(T) \mid \overline{x}] : {\hat\mA}^{|\mE'|}  \dto \mA^{|\inn(T)|}.$$
 \end{definition}

\begin{remark} \label{important2}
The channel describing the transfer from ${\bf S}_J$ to $T$, given the 
 input $\overline{x}$, again factors through any cut
$\mE'$ between ${\bf S}_J$ and $T$, i.e., 
 $$\Omega_\mF^J[{\bf S}_J \to \inn(T) \mid \overline{x}]  =  
 \Omega_\mF^J[{\bf S}_J \to \mE' \mid \overline{x}]  \concat  \Omega_\mF^J[
\mE' \to \inn(T) \mid \overline{x}].$$
\end{remark}

\subsection{Network Adversaries}
We now model combinatorially an adversary capable of corrupting the values
of the edges of $\mN$, or erase them, according to certain restrictions.
We denote such an adversary by $\adv$, and call the pair $(\mN,\adv)$
an \textbf{adversarial network}. The vertices still process the 
incoming packets using  a network code $\mF$.

We start by defining some adversarial analogues of the deterministic channels
introduced in Definitions~\ref{def11} and~\ref{def22}.

\begin{definition} \label{defadvchN}
Let $\mF$ be a network code for $\mN$, and let $T \in {\bf T}$ be a terminal.
We denote by 
$$ \Omega_\mF[\adv;{\bf S} \to \inn(T)] : \prod_{i \in I} \mX_i
\dto {\hat\mA}^{|\inn(T)|}$$
the channel that associates to an input $x \in \prod_{i \in I} \mX_i$ the variety
of packets that can be observed over the edges of $\inn(T)$ in the presence of the adversary  
$\adv$. Similarly, when $J \subseteq I$ is a non-empty subset of source indices, we denote by 
$$\Omega_\mF^J[\adv;{\bf S}_J \to \inn(T) \mid \overline{x}]: \prod_{i \in J} \mX_i 
\dto \hat\mA^{|\inn(T)|}$$
the channel that associates to $x \in \prod_{i \in J} \mX_i$
the variety of packets that can be observed over the edges from $\mE'$ when, in the presence of the adversary $\adv$, the sources in ${\bf S}_J$ emit the corresponding packets from $x \in \prod_{i \in J} \mX_i$,  and the sources in ${\bf S}_{I \setminus J}$ emit the corresponding (fixed) packets 
from $\overline{x} \in \prod_{i \in I \setminus J} \mX_i$.
\end{definition}

Clearly, the adversarial channels introduced in 
Definition~\ref{defadvchN} 
are not deterministic in general.

%We can still factor the channel describing the transfer from a set 
%of sources to a given terminal $T \in {\bf T}$ through a 
%cut $\mE' \subseteq \mE$ as follows:
% \begin{eqnarray}
%  \Omega^\adv_\mF[{\bf S} \to \inn(T)]  &=&  \Omega^\adv_\mF
%  [{\bf S} \to \mE']  \concat  \Omega^\adv_\mF[\mE' \to \inn(T)], 
%  \label{eee1}\\
% \Omega^\adv_\mF^J[\overline{x};{\bf S}_J \to \inn(T)]  
%&=& \Omega^\adv_\mF^J[\overline{x};
%{\bf S}_J \to \mE'] \concat \Omega^\adv_\mF^J[\overline{x};
%\mE' \to \inn(T)], \label{eee2}
%\end{eqnarray}
%where in equation (\ref{eee1}) $\mE'$ is a cut 
%between $\{S_1,...,S_N\}$ and 
%$T$, and in equation (\ref{eee2}) $\mE'$ is a cut 
%between $\{S_i : i \in J\}$ and 
%$T$.

\begin{example} \label{dex1}
Let $\mN$ be the network depicted in 
 Figure~\ref{figex1}. We order the edges 
of $\mN$ according to their indices. The input alphabet is 
$\mX=\mX_1 \times \mX_2=\mA^2 \times \mA$.
A network code $\mF$ for $\mN$ is the assignment of a function 
$\mF_V:{\hat\mA}^2 \times {\hat\mA} \to \mA^4$.

\setlength{\mynodespace}{7.5em}
\begin{figure}[htbp]
\centering
\begin{tikzpicture}
\tikzset{vertex/.style = {shape=circle,draw,inner sep=0pt,minimum size=1.9em}}
\tikzset{nnode/.style = {shape=circle,fill=myg,draw,inner sep=0pt,minimum
size=1.9em}}
\tikzset{edge/.style = {->,>= stealth}}
\tikzset{dedge/.style = {densely dotted,->,> = stealth}}

\node[nnode] (N) at (0,0) {$V$};
\node[vertex, right=\mynodespace of N] (R) {$T$};
\node[vertex] (S1) at (165:\mynodespace) {$S_1$};
\node[vertex] (S2) at (195:\mynodespace) {$S_2$};
\draw[edge, bend left=10] (S1)  to node[fill=white,sloped,inner sep=1pt]{\small $e_1$} (N);
\draw[edge, bend right=10] (S1) to  node[fill=white,sloped,inner sep=1pt]{\small $e_2$} (N);
\draw[edge] (S2) to node[fill=white,sloped,inner sep=1pt]{\small $e_3$} (N);
\draw[dedge, bend left=30] (N) to node[fill=white,sloped,inner sep=1pt]{\small $e_4$} (R);
\draw[edge, bend left=10] (N) to node[fill=white,sloped,inner sep=1pt]{\small $e_5$} (R);
\draw[dedge, bend right=10] (N) to node[fill=white,sloped,inner sep=1pt]{\small $e_6$} (R);
\draw[dedge, bend right=30] (N) to node[fill=white,sloped,inner sep=1pt]{\small $e_7$} (R);
\end{tikzpicture} 
\caption{{{Network for Example~\ref{dex1}.}}}\label{figex1}
\end{figure}
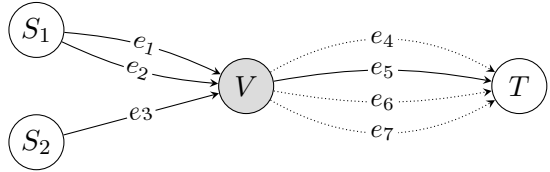

Consider an adversary $\adv$ able to erase at most one of the 
values of the dotted edges
of $\mN$, i.e., $\{e_4,e_6,e_7\}$. This scenario is modeled by a 
channel $\Omega_\mF[\adv;{\bf S} \to \inn(T)]: \mA^2 \times \mA 
\dto {\hat\mA}^4$. We can  describe this channel as follows. If $(x_1,x_2,x_3) \in 
\mA^2 \times \mA$ and $z:=\mF_V(x_1,x_2,x_3) \in \mA^4$, then
$$\Omega_\mF[\adv;{\bf S} \to \inn(T)](x_1,x_2,x_3)= \{y \in {\hat\mA}^4 : y_2 = z_2
  \mbox{ and  }  \omega_\star(y;\{1,2,3,4\}) \le 1\},$$
 where $\omega_\star$ is the erasure weight of
 Definition~\ref{defomegastar}.
 Now let $J:=\{1\}$, and assume that $S_2$ emits a fixed element 
 $\overline{x}_3 \in \mA$. The channel $\Omega_\mF^J[\adv;
 {\bf S}_J \to \inn(T) \mid \overline{x}_3]$ is as follows. 
 If $(x_1,x_2) \in \mA^2$, and $\overline{z}:=\mF_V(x_1,x_2,\overline{x}_3) \in \mA^4$,
 then  $
 \Omega^J_\mF[\adv;{\bf S} \to \inn(T)](x_1,x_2)= \{y \in {\hat\mA}^4 : y_2 =\overline{z}_2
  \mbox{ and  }  \omega_\star(y;\{1,2,3,4\}) \le 1\}$.
 \end{example}
 
 We conclude this subsection by introducing
 a general class of network adversaries, acting on sets of edges 
 with limited error and erasure power.

\begin{notation}\label{edgespecific}
Let $\mU \subseteq \mE$ be a set of network edges, and let $t,e \ge 0$ be 
integers. We denote by $\AP{t}{e}{\mU}$
an adversary having access only to edges from $\mU$, being able to erase 
up to $e$ of them, and change the values of up to $t$ of them with different
symbols from $\mA$.

Let $L \ge 1$ be an integer, and let $\mU_1,...,\mU_L \subseteq \mE$ be pairwise 
disjoint subsets of network edges. Let $t_1,...,t_L$
and $e_1,...,e_L$  be non-negative integers.
Define  the $L$-tuples $\bm{\mU}:=(\mU_1,...,\mU_L)$, $\bm{t}:=(t_1,...,t_L)$ and $\bm{e}:=(e_1,...,e_L)$.
We denote by $\AB{t}{e}{\mU}$ the adversary representing the scenario where all the 
$\AP{t_\ell}{e_\ell}{\mU_\ell}$'s act simultaneously on the network, possibly with coordination. 
\end{notation}

\subsection{Capacity Regions of Adversarial Networks} \label{subcapacityregions}

In this section, in analogy with the theory of point-to-point adversarial channels,
we propose definitions of one-shot, zero-error, and  
compound zero-error capacity regions
of adversarial networks. 

The one-shot capacity region, as the name suggests,
measures the number of alphabet symbols that the sources can 
multicast to all terminals in a single use of the network.
The zero-error capacity region measures the average number 
of packets that can be multicasted to the terminals per network use, in the limit
where the number of uses goes to infinity. Finally, the compound 
zero-error capacity regions measures the average number 
of packets that can be multicasted to the terminals per network use, in the limit
where the number of uses goes to infinity, and in the scenario where the adversary 
is forced to operate on the same (unknown) set of edges in every network use.

\begin{remark}
A definition of capacity region was proposed in~\cite[Section IV]{MANIAC} in the context of 
adversarial random network coding. The definition of~\cite{MANIAC} is probabilistic,
and best suited to study random linear approaches to multicast. 
In particular, it
does not coincide with the notion of capacity originally proposed in~\cite{YeungCai} 
for single-source adversarial networks (see, e.g.,~\cite[Theorem 1]{YeungCai}).
Inspired by the work of Shannon~\cite{shannon_zero}, our approach to adversarial 
network channels has a 
combinatorial flavor, rather than a probabilistic one.
As a consequence,  the three notions of capacity region introduced in this paper 
are all different from the one proposed in~\cite{MANIAC}, and thus require 
an independent analysis.
\end{remark}

\begin{definition} \label{netcap}
The (\textbf{one shot}) \textbf{capacity region} of $(\mN,\adv)$ is 
the set $\reg(\mN,\adv)$ of all the 
$N$-tuples $(\alpha_1,...,\alpha_N) \in \R_{\ge 0}^N$ for which 
there exist a network code 
$\mF$ for $\mN$ and non-empty
sets $\mC_i \subseteq \mX_i$, for $1 \le i \le N$, with the following properties:
\begin{enumerate} \setlength\itemsep{0em}
\item $\log_{|\mA|}|\mC_i|=\alpha_i$,
\item $\mC=\mC_1 \times \cdots \times \mC_N$ is good  
for each channel $\Omega_\mF[\adv;{\bf S} \to \inn(T)]$, $T \in {\bf T}$. 
\end{enumerate}
We say that such a pair $(\mF,\mC)$  
\textbf{achieves} the \textbf{rate} $(\alpha_1,...,\alpha_N)$ {in one shot}.
\end{definition}

The conditions in Definition~\ref{netcap} guarantee that the 
 sources can transmit in one shot to each of the sinks
$\alpha_1+ \cdots + \alpha_N$ alphabet symbols, $\alpha_i$ of which  are emitted
by $S_i$, for $1 \le i \le N$. Note that the code $\mC$ is 
by definition a cartesian product. This models the scenario where
 the sources cannot coordinate.

\begin{remark}
Let $n \ge 1$ be an integer. The $n$-th cartesian power 
$\mX^n = (\mX_1 \times \cdots \times \mX_N)^n$
can be naturally identified with the set $\mX_1^n \times \cdots
 \times \mX_N^n$. Therefore if $T \in {\bf T}$ and $\mF^1,...,\mF^n$ are network codes for 
$\mN$, we can see a product channel of the form 
$\prod_{k=1}^n \Omega_{\mF^k}[\adv;{\bf S} \to \inn(T)]$ as a channel
\begin{equation} \label{interpr}
\prod_{k=1}^n \Omega_{\mF^k}[\adv;{\bf S} \to \inn(T)]:~
\mX_1^n \times \cdots \times \mX_N^n \dto \hat\mA^{|\inn(T)|}.
\end{equation}
\end{remark}

We now introduce the zero-error capacity region.

\begin{definition} \label{netcap0}
The \textbf{zero-error capacity region} of $(\mN,\adv)$ is the  
 closure $\overline{\regz(\mN,\adv)}$ of the set $\regz(\mN,\adv)$ of all the 
$N$-tuples $(\alpha_1,...,\alpha_N) \in \R_{\ge 0}^N$ for which there exist:
\begin{itemize}\setlength\itemsep{0em}
\item an integer $n \ge 1$,
\item network codes $\mF^1,...,\mF^n$ for $\mN$,
\item non-empty sets $\mC_i \subseteq \mX_i^n$, for $1 \le i \le N$,
\end{itemize}
with the following properties:
\begin{enumerate}\setlength\itemsep{-0.8em}
\item $\log_{|\mA|}|\mC_i|=n \cdot \alpha_i$,
\item $\mC=\mC_1 \times \cdots \times \mC_N$ is good  
for each channel 
$\displaystyle \prod_{k=1}^n \Omega_{\mF^k}[\adv;{\bf S} \to \inn(T)], 
\mbox{$T \in {\bf T}$}$, in the sense of (\ref{interpr}). 
\end{enumerate}
\end{definition}

\begin{remark}\label{rmkisdifferent}
When the network vertices operate in a memoryless way (as in our setup) or, more generally, in a causal way,  then 
doubling the length of the network messages
(i.e., assuming that the alphabet is 
$\mA \times \mA$ instead of $\mA$)
 does not model two uses of the network, as the following example shows. 
More generally, taking the $n$-th power of the 
the network alphabet does not model $n$ network uses. 
As a consequence,  the zero-error capacity region is not the same as
the one-shot capacity region, in the limit where 
the ``length'' of the network alphabet goes to infinity. 
These two capacity concepts are different in general, and require 
independent analyses.
\end{remark}

\begin{example} \label{ex_isdifferent}
Let $\mN$ be the network in Figure~\ref{figex_isdifferent}, and let $\mA:=\F_2$ be the network alphabet.
The edges are ordered according to their indices. Consider an adversary on $\mN$
who can only operate on the edges from 
$\mE'=\{e_3, e_4,e_5,e_6\}$ by possibly corrupting the value of one of these edges.
The action of the adversary is described by the channel 
$\HH:\F_2^4 \dto \F_2^4$ of Examples~\ref{firstexample} and~\ref{fee}, given by
$\HH(x)=\{y \in \F_2^4 : \dH(y,x) \le 1\}$
 for all $x \in \F_2^4$.
Observe that a network code for $\mN$ is the assignment of a function
$\mF_V:\F_2^2 \to \F_2^4$, as erasures are excluded from the model.

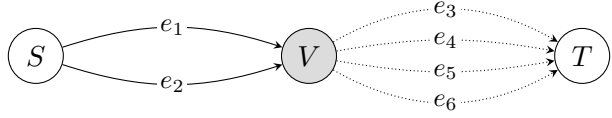
\begin{figure}[htbp]
  \centering
     \begin{tikzpicture}
\tikzset{vertex/.style = {shape=circle,draw,inner sep=0pt,minimum size=1.9em}}
\tikzset{nnode/.style = {shape=circle,fill=myg,draw,inner sep=0pt,minimum
size=1.9em}}
\tikzset{edge/.style = {->,> = stealth}}
\tikzset{dedge/.style = {densely dotted,->,> = stealth}}

\node[nnode] (V) at  (0,0) {$V$};
\node[vertex, left=\mynodespace of V] (S) {$S$};
\node[vertex, right=\mynodespace of V] (T)  {$T$};

\draw[edge, bend left=18] (S)  to node[fill=white,inner sep=1pt,sloped]{\small $e_1$} (V);
\draw[edge, bend right=18] (S)  to node[fill=white,inner sep=1pt,sloped]{\small $e_2$} (V);
\draw[dedge, bend left=30] (V) to  node[fill=white,inner sep=1pt,sloped]{\small $e_3$} (T);
\draw[dedge, bend left=10] (V) to  node[fill=white,inner sep=1pt,sloped]{\small $e_4$} (T);
\draw[dedge, bend right=10] (V) to  node[fill=white,inner sep=1pt,sloped]{\small $e_5$} (T);
\draw[dedge, bend right=30] (V) to  node[fill=white,inner sep=1pt,sloped]{\small $e_6$} (T);
\end{tikzpicture} 
\caption{{{Network for Example~\ref{ex_isdifferent}.}}}\label{figex_isdifferent}
\end{figure}

Assume that the network is used twice, with network codes
$\mF^1$ and $\mF^2$ in the first and second channel use, respectively.
The channel describing the two uses of the network is the product channel 
$$\Omega_1:=(\Omega_{\mF^1}[{\bf S} \to \mE'] \concat \HH) \times 
(\Omega_{\mF^2}[{\bf S} \to \mE'] \concat \HH): 
\F_2^2 \times \F_2^2 \dto \F_2^4 \times \F_2^4.$$
Note that for all $x=(x_1,x_2,x_3,x_4) \in \F_2^2 \times \F_2^2$ we have
$\Omega_1(x)=\HH^2 (\mF_V^1(x_1,x_2),\mF_V^2(x_3,x_4))$.
We will show that $\CA(\Omega_1) < \log_2 5$.
Assume by contradiction that there exists a code $\mC \subseteq \F_2^2 \times \F_2^2$
with $|\mC|=5$ and which is good for $\Omega_1$. Then 
$\overline{\mC}:=\{(\mF_V^1(x_1,x_2),\mF_V^2(x_3,x_4)) : x=(x_1,x_2,x_3,x_4) \in \mC\} \subseteq \F_2^4 \times \F_2^4$
is a good code for the channel $\HH^2$ of cardinality five. Since $|\mC|=5$, there must exist $x,x' \in \mC$ with
$x \neq x'$ and $(x_1,x_2)=(x_1',x_2')$. Therefore $\overline{\mC} \subseteq \F_2^4 \times \F_2^4$ is a good code for 
$\HH^2$ of cardinality five which has two distinct codewords that coincide in the first four components. However, as shown in Example~\ref{fee}, there is no such a code.
Thus $\CA(\Omega_1) < \log_2 5$, as claimed.

Now assume that the edges can carry symbols from $\F_2 \times \F_2$.
Using Example~\ref{fee} one can show that there exists a function
 $\mF_V: \F_2^2 \times \F_2^2 \to \F_2^4 \times \F_2^4$ such that
  $\Omega_2:=\Omega_\mF[{\bf S} \to \mE'] \concat \HH^2 : \F_2^2 \times \F_2^2 \dto \F_2^4 \times \F_2^4$
has capacity $\CA(\Omega_2) \ge \log_2 5$. Therefore  
$\CA(\Omega_2) > \CA(\Omega_1)$.
\end{example}

We now focus on composite adversaries of type $\AB{t}{e}{\mU}$, where 
$L\ge 1$, and $\bm{\mU}$, $\bm{t}$ and $\bm{e}$ are as in Notation~\ref{edgespecific}.
Consider the scenario where the network $\mN$ is used multiple times, and 
where the $L$ network adversaries $\AP{t_1}{e_1}{\mU_1}, ..., \AP{t_L}{e_L}{\mU_L}$ 
defining $\AB{t}{e}{\mU}$ are 
forced to operate on the same set of edges in each use of the network. In analogy with 
Definition~\ref{defcompound}, this scenario naturally motivates the following definition.

\begin{definition}\label{defcompoundregion}
Let $L\ge 1$, $\bm{\mU}$, $\bm{t}$ and $\bm{e}$ be as in Notation~\ref{edgespecific}, 
and assume $\adv=\AB{t}{e}{\mU}$.
The \textbf{compound zero-error capacity region} of $(\mN,\adv)$ is the 
 closure $\overline{\regzr(\mN,\adv)}$ of the set 
 $\regzr(\mN,\adv)$ of all the 
$N$-tuples $(\alpha_1,...,\alpha_N) \in \R_{\ge 0}^N$ for which there exist:
\begin{itemize}\setlength\itemsep{0em}
\item an integer $n \ge 1$,
\item network codes $\mF^1,...,\mF^n$ for $\mN$,
\item non-empty sets $\mC_i \subseteq \mX_i^n$, for $1 \le i \le N$,
\end{itemize}
with the following properties:
\begin{enumerate}\setlength\itemsep{-0.8em}
\item $\log_{|\mA|}|\mC_i|=n \cdot \alpha_i$,
\item $\mC=\mC_1 \times \cdots \times \mC_N$ is good  
for each channel 
$\displaystyle
\bigcup_{\substack{ \bm{\mV} \subseteq \bm{\mU} \\
|\bm{\mV}| \le \bm{t}+\bm e}}  \prod_{k=1}^n \Omega_{\mF^k}[\AB{t}{e}{\mV};{\bf S} \to \inn(T)], \, \mbox{$T \in {\bf T}$}.$
\end{enumerate}
\end{definition}

\begin{remark} \label{rmkregions}
It follows from the definitions that $\mR(\mN,\adv) \subseteq
 \regz(\mN,\adv) \subseteq \overline{\regz(\mN,\adv)}$ for any adversary~$\adv$.
Moreover, if $\adv$ is of the form $\adv=\AB{t}{e}{\mU}$, then 
$\regz(\mN,\adv) \subseteq \regzr(\mN,\adv)$.
In particular, we have 
$\overline{\regz(\mN,\adv)} \subseteq \overline{\regzr(\mN,\adv)}$.
This property formalizes the fact that, in the compound model, the adversaries are forced to operate on the same edges in each channel use, and therefore their power is possibly reduced. 
\end{remark}

%%%%%%%%%%%%%%%%%%%%%%%%%%%%%%%%%%%

\section{Porting Bounds for Hamming-Type Channels to Networks}
\label{secportinglemma}

In this section we show how the properties of channel products, concatenations and unions 
established in Section~\ref{secoperations} can be combined with each other to obtain 
upper bounds for the three capacity regions of an adversarial network.
The idea behind our approach is to derive cut-set bounds by porting, in 
a systematic manner, upper bounds for the three capacities of Hamming-type 
channels to networks.

\begin{notation}
In the sequel we follow the notation of Section~\ref{secnetwork}, and
 fix a network $\mN=(\mV,\mE,{\bf S},{\bf T})$
with alphabet $\mA$.
The sources 
are ${\bf S}=\{S_1,...,S_N\}$, and $I=\{1,...,N\}$ is the set of source indices. The 
alphabet of source $S_i$ is $\mX_i:=\mA^{|\out(S_i)|}$, for
$1 \le i \le N$, and $\mX:=\mX_1 \times \cdots \times \mX_N$.
We work with a generic adversary 
 $\adv$ and a fixed linear extension $\le$
 of the network partial order $\preceq$.
\end{notation}

\subsection{Method Description}
\label{submethod}

Suppose that we are interested in describing the one-shot capacity region 
$\mR(\mN,\adv) \subseteq \R_{\ge 0}^N$ of $(\mN,\adv)$. Let
$(\alpha_1,...,\alpha_N) \in \mR(\mN,\adv)$ be arbitrary. By Definition~\ref{netcap}, 
there exist a network code $\mF$ for $\mN$
and a code $\mC=\mC_1 \times \cdots \times \mC_N
\subseteq \mX_1 \times \cdots \times \mX_N$ such that 
$(\mF,\mC)$ achieves $(\alpha_1,...,\alpha_N)$ in one shot.

Let $J \subsetneq I$ be a non-empty subset of source indices, and let 
$T \in {\bf T}$ be a terminal. Take a cut $\mE' \subseteq \mE$ that separates ${\bf S}_J$
from $T$, and fix an element $\overline{x} \in 
\prod_{i \in I \setminus J} \mC_i$. 
Assume for the moment that there is no adversary $\adv$. As observed in Remark~\ref{important2}, the channel
$\Omega_\mF^J[{\bf S}_J \to \inn(T) \mid \overline{x}]$ decomposes as
follows:
\begin{equation} \label{decident}
\Omega_\mF^J[{\bf S}_J \to \inn(T) \mid \overline{x}]  =  \Omega_\mF^J[
{\bf S}_J \to \mE' \mid \overline{x}]  \concat  \Omega_\mF^J[
\mE' \to \inn(T) \mid \overline{x}].
\end{equation}
Denote 
by $\Id[\mE' \to \mE']:{\mA}^{|\mE'|} \dto {\hat\mA}^{|\mE'|}$ the 
identity channel (see Example~\ref{exid}),
and observe that equation (\ref{decident}) can be trivially re-written as
\begin{equation} \label{decident2}
\Omega_\mF^J[{\bf S}_J \to \inn(T) \mid \overline{x}]  =  \Omega_\mF^J[
{\bf S}_J \to \mE' \mid \overline{x}]  \concat  \Id[\mE' \to \mE']  \concat  \Omega_\mF^J[
\mE' \to \inn(T) \mid \overline{x}].
\end{equation}

We now let the adversary $\adv$ act only on the cut $\mE'$,
and assume that its action on the edges of $\mE'$ can be described by a channel 
$\Omega[\adv;\mE' \to \mE']: \mA^{|\mE'|} \dto {\hat\mA}^{|\mE'|}$. 

For example, if $\adv$ was originally able to corrupt 
the values of up to $t$ edges of $\mN$, we now consider the scenario 
where $\adv$ can corrupt up to $t$ edges from $\mE'$.
This action is described by the Hamming-type channel 
$\HP{t}{0}{U}: \mA^{|\mE'|} \dto 
{\hat\mA}^{|\mE'|}$ introduced in Definition 
\ref{elementarychannels}, with $s:=|\mE'|$ and $U:=[s]=\{1,...,s\}$. 
Therefore we can simply take 
$\Omega[\adv;\mE' \to \mE']:=\HP{t}{0}{U}$ in this case.

Going back to our general discussion, observe that letting the adversary $\adv$ 
operate only on $\mE'$ replaces the channel $\Id[\mE' \to \mE']$ in (\ref{decident2}) 
with the channel $\Omega[\adv;\mE' \to \mE']$. This defines
 a new channel
\begin{equation} \label{newch}
\overline{\Omega}_\mF^J[\adv;{\bf S}_J \to \inn(T) \mid \overline{x}]  :=  \Omega_\mF^J[
{\bf S}_J \to \mE' \mid \overline{x}]  \concat  \Omega[\adv;\mE' \to \mE']  \concat \Omega_\mF^J[\mE' \to \inn(T) \mid \overline{x}],
\end{equation}
whose decomposition is graphically represented in Figure~\ref{diagg}, for a fixed $\overline{x}$.
\begin{figure}[htbp]
  \centering
     \begin{tikzpicture}
\tikzset{vertex/.style = {shape=circle,draw,inner sep=0pt,minimum size=1.9em}}
\tikzset{nnode/.style = {shape=circle,fill=myg,draw,inner sep=0pt,minimum
size=1.9em}}
\tikzset{edge/.style = {->,> = stealth}}
\tikzset{dedge/.style = {densely dotted,->,> = stealth}}
\tikzset{ddedge/.style = {dashed,->,> = stealth}}

\node[vertex] (T) at (13,-1) {$T$};

\node (N1) at  (5,-1) {};
\node (N2) at  (5,-1.5) {};
\node (N3) at  (5,-2.5) {};

\node (N4) at  (7,-1) {};
\node (N5) at  (7,-1.5) {};
\node (N6) at  (7,-2.5) {};
\node (wr) at (6,-2) {$\vdots$};

\draw[dedge] (N1) to (N4);
\draw[dedge] (N2) to (N5);
\draw[dedge] (N3) to (N6);

\draw [decorate,decoration={brace,amplitude=2pt},xshift=0pt,yshift=0pt]
(5,-2.8) -- (5,-0.8) node [black,midway,xshift=0.4cm] 
{};

\draw [decorate,decoration={brace,amplitude=2pt},xshift=0pt,yshift=0pt]
(7,-0.8) -- (7,-2.8) node [black,midway,xshift=0.4cm] 
{};

\draw[ddedge, bend left=15] (0.05,-1.45)  
to node[midway,above]{\small $
\Omega_\mF^J[{\bf S}_J \to \mE' \mid \overline{x}]$}
(4.8,-1.8);

\node (scr) at (6,-0.4) {$\Omega[\adv;\mE' \to \mE']$};
\draw[ddedge, bend right=15] (7.2,-1.8)  to node[midway,below]{\small $
\Omega_\mF^J[\mE' \to \inn(T) \mid \overline{x}]$} (T);

\node (scr) at (6.5,-2.1) {$\mE'$};

\node (SS) at  (-0.30,-1.45) {${\bf S}_J$};
\end{tikzpicture} 
\caption{{{Decomposition of the channel $\overline{\Omega}_\mF^J[\adv;{\bf S}_J \to \inn(T) \mid \overline{x}]$.}}}
\label{diagg}
\end{figure}

Note that $\overline{\Omega}_\mF^J[\adv;{\bf S}_J \to 
\inn(T) \mid \overline{x}]$ is finer than
 $\Omega_\mF^J[\adv;{\bf S}_J \to \inn(T) \mid \overline{x}]$, as in the former channel 
the action of the adversary was restricted to $\mE'$. In symbols (see Definition~\ref{deffiner}) we have
\begin{equation} \label{finerr} 
\overline{\Omega}_\mF^J[\adv;{\bf S}_J \to \inn(T) \mid \overline{x}]
  \le \Omega_\mF^J[\adv;{\bf S}_J \to \inn(T) \mid \overline{x}].
\end{equation}

Now observe that the code
$\prod_{i \in J} \mC_i$ is good for the channel 
$\Omega_\mF^J[\adv;{\bf S}_J \to \inn(T) \mid \overline{x}]$.
Indeed,
$\prod_{i \in I} \mC_i$ is a good code for  
$\Omega_\mF[\adv;{\bf S} \to \inn(T)]$ by definition of one-shot capacity region,
and $\overline{x} \in \prod_{i \in I \setminus J} \mC_i$.

Using~(\ref{finerr}), we see that $\prod_{i \in J} \mC_i$ is good  for the channel 
$\overline{\Omega}_\mF^J[\adv;{\bf S}_J \to \inn(T) \mid \overline{x}]$ as well. In particular,
\begin{equation}\label{bbb}
\sum_{i \in J} \alpha_i  =  \log_{|\mA|} \left| \prod_{i \in J} \mC_i \right| \le \CA\left(\overline{\Omega}_\mF^J[\adv;{\bf S}_J \to \inn(T)
\mid \overline{x}]\right) \le \CA(\Omega[\adv;\mE' \to \mE']),\end{equation}
where $\CA$ denotes the one-shot capacity (expressed as a logarithm in base $|\mA|$), and the last inequality follows combining equation 
(\ref{newch}) and Proposition~\ref{genconc}.
We therefore obtained in (\ref{bbb}) an upper-bound for the capacity region
$\mR(\mN,\adv)$ in terms of the capacity
of the intermediate  channel $\Omega[\adv;\mE' \to \mE']: \mA^{|\mE'|} \dto {\hat\mA}^{|\mE'|}$, which does not depend on the specific choice of the network code
$\mF$.

A similar (and in fact easier) argument can be given for the case $J=I$, i.e., when all sources are considered. In that case there is no $\overline{x}$ to select.

\begin{example} \label{dex2}
Let $\mN$ be the network in Figure~\ref{figex2}. The edges 
of $\mN$ are ordered according to their indices. 
We consider an adversary $\adv$ who can corrupt
up to one packet from the dotted edges, i.e., from $\mU:=\{e_1,e_4,e_6,e_7\}$.
Therefore $\adv=\AP{1}{0}{\mU}$ according to Notation~\ref{edgespecific}
 (erasures are excluded 
from the model). Let $(\alpha_1,\alpha_2) \in \mR(\mN,\adv)$
be arbitrary.
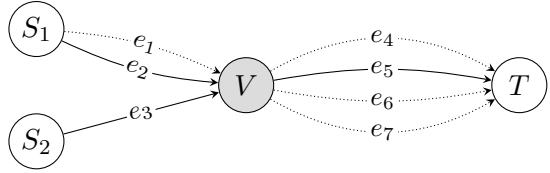
\begin{figure}[htbp]
  \centering
     \begin{tikzpicture}
\tikzset{vertex/.style = {shape=circle,draw,inner sep=0pt,minimum size=1.9em}}
\tikzset{nnode/.style = {shape=circle,fill=myg,draw,inner sep=0pt,minimum
size=1.9em}}
\tikzset{edge/.style = {->,> = stealth}}
\tikzset{dedge/.style = {densely dotted,->,> = stealth}}

\node[nnode] (N) at (0,0) {$V$};
\node[vertex, right=\mynodespace of N] (R) {$T$};
\node[vertex] (S1) at (165:\mynodespace) {$S_1$};
\node[vertex] (S2) at (195:\mynodespace) {$S_2$};

\draw[dedge, bend left=10] (S1)  to node[fill=white,inner sep=1pt,sloped]{\small $e_1$} (N);
\draw[edge, bend right=10] (S1) to  node[fill=white,inner sep=1pt,sloped]{\small $e_2$} (N);
\draw[edge] (S2) to node[fill=white,inner sep=1pt,sloped]{\small $e_3$} (N);
\draw[dedge, bend left=30] (N) to node[fill=white,inner sep=1pt,sloped]{\small $e_4$} (R);
\draw[edge, bend left=10] (N) to node[fill=white,inner sep=1pt,sloped]{\small $e_5$} (R);
\draw[dedge, bend right=10] (N) to node[fill=white,inner sep=1pt,sloped]{\small $e_6$} (R);
\draw[dedge, bend right=30] (N) to node[fill=white,inner sep=1pt,sloped]{\small $e_7$} (R);
\end{tikzpicture} 
\caption{{{Network for Example~\ref{dex2}.}}}\label{figex2}
\end{figure}

Take $J=\{1\}$.
The set $\mE'=\{e_4,e_5,e_6,e_7\}$ is a cut between $S_1$ and $T$. 
Let $U:=\{1,3,4\}$, and let 
$\HP{1}{0}{U}$ denote the channel introduced in Definition
\ref{elementarychannels}. Using (\ref{bbb}) we obtain
$\alpha_1 \le \CA(\HP{1}{0}{U})$, and by 
Theorem~\ref{mainthhamming} we have 
$\CA(\HP{1}{0}{U}) \le 4-2\cdot 1=2$. Therefore 
$\alpha_1 \le 2$.
Reasoning in the same way with the cut 
$\mE'=\{e_1,e_2\}$, one obtains a better bound for $\alpha_1$, i.e., 
$\alpha_1 \le 1$.

Taking $J=\{2\}$ and $J=\{1,2\}$, and combining in the same way (\ref{bbb}) and
Theorem~\ref{mainthhamming}, we obtain
$\alpha_2 \le 1$ and 
$\alpha_1+\alpha_2 \le 2$. Therefore
$\mR(\mN,\adv) \subseteq \{(\alpha_1,\alpha_2) \in \R_{\ge 0}^2 :
\alpha_1 \le 1, \,\alpha_2 \le 2, \,\alpha_1+\alpha_2 \le 1\}$. 
\end{example}

\subsection{Porting Lemmas}\label{subportinglemmas}

We now establish two general lemmas that formalize the argument presented 
in  Subsection~\ref{submethod}, and extend it to the zero-error and to the compound 
zero-error capacity regions. Note that these extensions do not directly 
follow from the discussion in the previous 
subsection, and heavily rely on the properties of channel concatenation and union established in 
Section~\ref{secoperations}. We start with the Porting Lemma for the one-shot and the zero-error capacity.

\begin{lemma}\label{portinglemma}
Let $\{\Omega[\adv; \mE' \to \mE'] : \mE'\subseteq \mE, \, \mE' \neq \emptyset\}$ 
be a family of adversarial channels such that
 $\Omega[\adv; \mE' \to \mE']: \mA^{|\mE'|} \dto \hat\mA^{|\mE'|}$ for all $\mE'$. 
Assume that for all network codes $\mF$ for $\mN$:

\begin{itemize}\setlength\itemsep{0em}
\item $\Omega_\mF[\adv; {\bf S} \to \inn(T)] \le \Omega_\mF[{\bf S} \to \mE'] 
\concat \Omega[\adv;\mE' \to \mE'] \concat \Omega_\mF[\mE' \to \inn(T)]$ 
for all $T \in {\bf T}$ and for all cuts $\mE' \subseteq \mE$ that separate $\bf S$ from $T$;
\item $\Omega_\mF^J[\adv;{\bf S}_J \to \inn(T) \mid \overline{x}] \le 
\Omega_\mF^J[{\bf S}_J \to \mE' \mid \overline{x}] \concat \Omega[\adv;\mE' \to \mE'] 
\concat \Omega_\mF^J[\mE' \to \inn(T) \mid \overline{x}]$ for all non-empty $J \subsetneq I$, all cuts 
 $\mE' \subseteq \mE$ that separate ${\bf S}_J$ from $T$, and all 
 $\overline{x} \in \prod_{i \in I \setminus J} \mX_i$.
 \end{itemize}
  The following hold.
  \begin{enumerate} \setlength\itemsep{0em}
\item \label{ineq1}For all $(\alpha_1,...,\alpha_N) \in \mR(\mN,\adv)$ and 
for all non-empty subset $J \subseteq I$ we have
\begin{equation}\label{region1}
 \sum_{i \in J} \alpha_i \le \min_{T \in {\bf T}} \, \min \{ \CA(\Omega[\adv;\mE' \to \mE']) : 
 \mE' \subseteq \mE \mbox{ is a cut between ${\bf S}_J$ and $T$}\}.
 \end{equation}
\item \label{ineq2}For all $(\alpha_1,...,\alpha_N) \in \overline{\regz(\mN,\adv)}$ 
and for all non-empty subset $J \subseteq I$ we have
 \begin{equation}\label{region2}
 \sum_{i \in J} \alpha_i  \le  \min_{T \in {\bf T}}  \, \min
  \{ \CAz(\Omega[\adv;\mE' \to \mE'])  : \mE' \subseteq \mE \mbox{ is a cut between ${\bf S}_J$ and $T$}\}.
 \end{equation}
 \end{enumerate}
   \end{lemma}

\begin{proof}
We only show the two properties for $J \subsetneq I$.
The proof for the case $J=I$ is similar and easier.
Part~\ref{ineq1} was essentially already shown in Subsection~\ref{submethod}.
More precisely, the bound in (\ref{region1}) follows from~(\ref{bbb}) by minimizing over all 
$T \in {\bf T}$ and $\mE'$.

Let us show part~\ref{ineq2}. Since the inequalities in 
$(\ref{region2})$ define a closed set in $\R^N$, it suffices to 
show the result for $\regz(\mN,\adv)$. Fix an arbitrary
$(\alpha_1,...,\alpha_N) \in \regz(\mN,\adv)$. Let
$n \ge 1$, $\mF^1,...,\mF^n$ and $\mC= \mC_1 \times \cdots 
\times \mC_N\subseteq \mX_1^n \times \cdots \times \mX_N^n$ 
be as in Definition~\ref{netcap0}. Take an element 
$$\overline{x}=(\overline{x}^1,...,\overline{x}^n) \in \prod_{i \in I \setminus J} 
\mC_i \subseteq \left( \prod_{i \in I \setminus J} \mX_i \right)^n, \quad
 \mbox{ where we identify }  
 \prod_{i \in I \setminus J} \mX_i^n  = \left( \prod_{i \in I \setminus J}
  \mX_i \right)^n.$$
 Let $T \in {\bf T}$ be a terminal, and let $\mE'$ be a cut that 
 separates ${\bf S}_J$ from $T$. 
 To simplify the notation throughout the proof, for all $1 \le k \le n$ define
 \begin{eqnarray}
 \nonumber \Omega_{\mF^k}^J[\adv;{\bf S}_J \to T] &:=& 
 \Omega_{\mF^k}^J[\adv;{\bf S}_J \to \inn(T) \mid \overline{x}^k], \\
 \nonumber \Omega_{\mF^k}^J[{\bf S}_J \to \mE']  &:=& 
 \Omega_{\mF^k}^J[{\bf S}_J \to \mE' \mid \overline{x}^k], \\
  \nonumber \Omega_{\mF^k}^J[\mE' \to T] &:=& 
  \Omega_{\mF^k}^J[\mE' \to \inn(T) \mid \overline{x}^k],  \\
 \label{deccc}\overline{\Omega}_{\mF^k}^J[\adv;{\bf S}_J \to T] &:=& 
 \Omega_{\mF^k}^J[{\bf S}_J \to \mE']   \concat 
 \Omega[\adv; \mE' \to \mE']  \concat \Omega_{\mF^k}^J[\mE' \to T].
 \end{eqnarray}
We can now apply Corollary 
\ref{propprel5}  to equation (\ref{deccc}), and obtain
 \begin{equation}\label{intermediate}
 \prod_{k=1}^n \overline{\Omega}_{\mF^k}^J[\adv;{\bf S}_J \to T]  = 
 \left( \prod_{k=1}^n  \Omega_{\mF^k}^J[{\bf S}_J \to \mE']\right) \concat 
 \Omega[\adv; \mE' \to \mE']^n \concat
 \left(\prod_{k=1}^n \Omega_{\mF^k}^J[\mE' \to T]\right).
 \end{equation}
By assumption, we have $\Omega_{\mF^k}^J[\adv;{\bf S}_J \to T] \ge 
\overline{\Omega}_{\mF^k}^J[\adv;{\bf S}_J \to T]$ for all $1 \le k \le n$. Combining 
this with equation (\ref{intermediate}) we deduce
$$\prod_{k=1}^n \Omega_{\mF^k}^J[\adv;{\bf S}_J \to T]  \ge 
  \left( \prod_{k=1}^n  \Omega_{\mF^k}^J[{\bf S}_J \to \mE']\right) \concat 
 \Omega[\adv; \mE' \to \mE']^n  \concat 
 \left(\prod_{k=1}^n \Omega_{\mF^k}^J[\mE' \to T]\right).$$
Since, by definition of $\regz(\mN,\adv)$, the code $\mC$ is good for 
$\prod_{k=1}^n \Omega_{\mF^k} [\adv;{\bf S} \to T]$,  the code 
$\prod_{i \in J} \mC_i$ is good for
$\prod_{k=1}^n \Omega_{\mF^k}^J[\adv;{\bf S}_J \to T]$.
Thus by Proposition~\ref{genconc} we conclude that
$$\sum_{i \in J} \alpha_i  =  \frac{1}{n}  
\log_{|\mA|} \left| \prod_{i \in J} \mC_i \right| \le 
\frac{1}{n}  \CA \left( \Omega[\adv; \mE' \to \mE']^n\right)
 \le  \CAz(\Omega[\adv; \mE' \to \mE']).$$ 
The final result can be obtained by minimizing over all $T \in {\bf T}$ and $\mE'$. 
\end{proof}

In the remainder of the section we establish an analogous porting lemma
 for the compound zero-error capacity. We start by defining Hamming-type channels 
 associated with an adversary of type
$\adv=\AB{t}{e}{\mU}$.

\begin{definition} \label{preldefn}
Let $L \ge 1$, $\bm{\mU}$, $\bm t$ and $\bm e$ be as in 
Notation~\ref{edgespecific}, and let $\bm{\mV}=(\mV_1,...,\mV_L) 
\subseteq \bm{\mU}$ be arbitrary (possibly $\bm \mV=\bm \mU$). 
Take a non-empty subset of edges $\mE' \subseteq \mE$, and set $s:=|\mE'|$. 
Denote by $e_1 < e_2 < \cdots < e_s$ the edges in $\mE$, sorted 
according to the linear order extension $\le$. For all
$1 \le \ell \le L$, let $V_\ell:=\{1 \le i \le s : e_i \in \mU_\ell\}$.
The channel $\HH_{\bm t, \bm e} \langle \bm{\mV}, \mE' \rangle$ is defined by
$\HH_{\bm t, \bm e} \langle \bm{\mV}, \mE' \rangle:= \HB{t}{e}{V} 
: \mA^s \dto \hat\mA^s$,
where $\bm V:=(V_1,...,V_L)$ and  $\HB{t}{e}{V}$ is the Hamming-type 
channel introduced in Definition~\ref{defhat}.
\end{definition}
 
 The channel $\HH_{\bm t, \bm e} \langle \bm{\mV}, \mE' \rangle$ can be interpreted as
 the ``projection'' of the adversary $\adv_{\bm t, \bm e} \langle \bm{\mV} \rangle$ on 
 $\mE'$. In particular, note that $\mE'$  determines 
 the size of the input/output alphabets
of the channel $\HH_{\bm t, \bm e} \langle \bm{\mV}, \mE' \rangle$. 

We can now state the Compound Porting Lemma.

\begin{lemma} \label{portinglemma1}
Let $\bm{\mU}$, $\bm t$ and $\bm e$ be as in 
Notation~\ref{edgespecific}, and assume that $\adv=\AB{t}{e}{\mU}$.
Then for all $(\alpha_1,...,\alpha_N) \in 
\overline{\regzr(\mN,\adv)}$ and for all non-empty subset $J \subseteq I$ we have
 \begin{equation*}\label{region3}
 \sum_{i \in J} \alpha_i  \le  \min_{T \in {\bf T}} \,  \min \{ 
 \CAzr(\HH_{\bm t, \bm e}
  \langle \bm\mU, \mE' \rangle)  : \mE' \subseteq \mE \mbox{ is a cut between ${\bf S}_J$ and $T$}\}.
 \end{equation*}
 \end{lemma}
\begin{proof}[Sketch of the proof] We proceed as in the proof of Lemma~\ref{portinglemma}. 
Take 
$(\alpha_1,...,\alpha_N) \in \regzr(\mN,\adv)$, and let
$n \ge 1$, $\mF^1,...,\mF^n$ and $\mC= \mC_1 \times \cdots \times \mC_N\subseteq 
\mX_1^n \times \cdots \times \mX_N^n$ be as in Definition~\ref{defcompoundregion}. 
Fix any element 
$\overline{x}=(\overline{x}^1,...,\overline{x}^n) \in \prod_{i \in I \setminus J} \mC_i.$
 Let $T \in {\bf T}$ be a terminal, and let $\mE'$ be a cut between ${\bf S}_J$ and $T$. 
 To simplify the notation throughout the proof, for $\bm \mV \subseteq \bm \mU$ 
 and for $1 \le k \le n$ we define
 \begin{eqnarray}
 \nonumber \Omega_{\mF^k}^J[\adv\langle \bm \mV \rangle;{\bf S}_J \to T] &:=& 
 \Omega_{\mF^k}^J[\adv_{\bm t, \bm e}\langle \bm \mV \rangle;
 {\bf S}_J \to \inn(T) \mid \overline{x}^k], \\
 \nonumber \Omega_{\mF^k}^J[{\bf S}_J \to \mE']  &:=&
  \Omega_{\mF^k}^J[{\bf S}_J \to \mE' \mid \overline{x}^k], \\
  \nonumber \Omega_{\mF^k}^J[\mE' \to T] &:=& 
  \Omega_{\mF^k}^J[\mE' \to \inn(T) \mid \overline{x}^k],  \\
 \label{deccc2}\overline{\Omega}_{\mF^k}^J[\adv\langle \bm \mV \rangle;{\bf S}_J \to T] &:=& 
 \Omega_{\mF^k}^J[{\bf S}_J \to \mE']   \concat \HH_{\bm t, \bm e} \langle \bm\mV, \mE' \rangle
  \concat \Omega_{\mF^k}^J[\mE' \to T].
 \end{eqnarray}
 Now observe that $\Omega_{\mF^k}^J[\adv \langle \bm \mV \rangle;{\bf S}_J \to T] \ge 
 \overline{\Omega}_{\mF^k}^J[\adv\langle \bm \mV \rangle;{\bf S}_J \to T]$ for all $1 \le k \le n$.
 Thus combining equation (\ref{deccc2}),  Corollary~\ref{propprel5} and Proposition
 ~\ref{proprunionfamily}, we obtain
\begin{multline*}
\bigcup_{\substack{ \bm{\mV} \subseteq \bm{\mU} \\
|\bm{\mV}| \le \bm{t}+\bm e}} \prod_{k=1}^n 
\Omega_{\mF^k}^J[\adv \langle \bm \mV \rangle;{\bf S}_J \to T]  \\ \ge
 \left( \prod_{k=1}^n  \Omega_{\mF^k}^J[{\bf S}_J \to \mE']\right) \concat 
 \left(\bigcup_{\substack{ \bm{\mV} \subseteq \bm{\mU} \\
|\bm{\mV}| \le \bm{t}+\bm e}} \HH_{\bm t, \bm e} 
\langle \bm\mV, \mE' \rangle^n \right) \concat 
 \left(\prod_{k=1}^n \Omega_{\mF^k}^J[\mE' \to T]\right).
 \end{multline*}
 As in the proof of Lemma~\ref{portinglemma}, by Proposition~\ref{genconc} we conclude
$$\sum_{i \in J} \alpha_i  =  \frac{1}{n} 
\log_{|\mA|} \left| \prod_{i \in J} \mC_i \right|  \le 
\frac{1}{n}  \CA\left(\bigcup_{\substack{ \bm{\mV} \subseteq \bm{\mU} \\
|\bm{\mV}| \le \bm{t}+\bm e}} \HH_{\bm t, \bm e} \langle \bm\mV, \mE' \rangle^n \right)
 \le \CAzr(\HH_{\bm t, \bm e} \langle \bm\mU, \mE' \rangle),$$
where the last inequality  can be shown using the 
definition of 
$\HH_{\bm t, \bm e} \langle \bm\mV, \mE' \rangle$
(see Definition~\ref{preldefn}). 
\end{proof}

%%%%%%%%%%%%%%%%%%%%%%%%%%%%%%%%%%%

\section{Capacity Regions: Upper Bounds}
\label{secbounds}

In this section we apply the theoretical bounds established in Section~\ref{secportinglemma} 
to concrete networking 
 contexts. Our results study one-shot models, zero-error models, and compound zero-error models.
Moreover, they cover several classes of network adversaries, including 
multiple adversaries, restricted adversaries, different types of 
error/erasure adversaries, and 
rank-metric adversaries (see Section~\ref{extensions}). The bounds presented in this section 
do not follow in any obvious from known results in the context of network communications under 
probabilistic error models.

We follow the notation of Section~\ref{secnetwork}, and
 work with a fixed network $\mN=(\mV,\mE,{\bf S},{\bf T})$ over an alphabet $\mA$.
The network sources are ${\bf S}=\{S_1,...,S_N\}$, and $I=\{1,...,N\}$ is the set of source indices.
We let $\beta$ be the parameter introduced in Notation~\ref{notazbeta}.

\begin{notation} \label{nnnn}
If $J \subseteq I$ is a non-empty subset and $T \in {\bf T}$ is a terminal, we denote by 
$\mu({\bf S}_J,T)$ the min-cut between ${\bf S}_J$ and $T$, i.e., the minimum size of an edge-cut that separates all the sources in ${\bf S}_J$ from $T$.
By the edge-connectivity version of Menger's Theorem (see~\cite{menger}), $\mu({\bf S}_J,T)$ 
coincides with the maximum number of edge-disjoint directed paths connecting ${\bf S}_J$ to $T$.
Note moreover that, by property~\ref{prnE} of Definition~\ref{defnetwork}, one has $\mu({\bf S}_J,T) \ge 1$ for all non-empty $J \subseteq I$ and for all $T \in {\bf T}$.
\end{notation}

\subsection{One-Shot Capacity Region}

We start by considering a single adversary $\adv$ restricted to a set of edges $\mU \subseteq \mE$, with limited error and erasure powers. The following result shows that any bound from classical Coding Theory translates into a bound for the one-shot capacity region of
$(\mN,\adv)$
via the Porting Lemma.

\begin{theorem} \label{THEO1}
Let $\mU \subseteq \mE$ be a subset of edges, and let $t,e \ge 0$ be integers. 
Assume $\adv=\AP{t}{e}{\mU}$. For all 
$(\alpha_1,...,\alpha_N) \in \mR(\mN,\adv)$ and for all non-empty $J \subseteq I$ we have 
\begin{equation}
\sum_{i \in J} \alpha_i \le \min_{T \in {\bf T}} \, \min\left\{ |\mE' \setminus \mU| + \beta(\mA,|\mE' \cap \mU|,2t+e+1) : \mE' \subseteq \mE \mbox{ is a cut between ${\bf S}_J$ and $T$}\right\}. \label{NNice}
\end{equation}
 In particular, if $\mU=\mE$ then for all 
$(\alpha_1,...,\alpha_N) \in \mR(\mN,\adv)$ and for all non-empty $J \subseteq I$ we have
 $$\sum_{i \in J} \alpha_i \le \min_{T \in {\bf T}}  \beta(\mA,\mu({\bf S}_J,T),2t+e+1).$$
\end{theorem}

\begin{proof}
Combine Proposition~\ref{boundext} and part~\ref{ineq1} of Lemma~\ref{portinglemma}.
The second part of the statement follows from the fact that
$\beta(\mA,u',2t+e+1) \le \beta(\mA,u,2t+e+1)$ for all $u \ge u' \ge 1$.
\end{proof}

\begin{remark}
When $\mU \subsetneq \mE$ in Theorem~\ref{THEO1}, the second minimum in (\ref{NNice}) is not realized, in general, by a 
minimum cut $\mE'$ between ${\bf S}_J$ and $T$ (examples can be easily found). Thus the  topology of 
the set $\mU$ of ``vulnerable'' edges plays an important role in the evaluation of the bound in (\ref{NNice}).
\end{remark}

As special cases of Theorem~\ref{THEO1}, we obtain generalizations 
of the Singleton-type and 
Hamming-type bounds established in~\cite{YeungCai} for single-source networks. Note that \textit{any} other
classical bound from Coding Theory can be ported to the networking context via Theorem~\ref{THEO1}, in the general case where the adversary is possibly restricted to operate on a subset $\mU \subseteq \mE$ of vulnerable edges. Observe moreover that no extra property of the edge-cuts $\mE'$ is needed in Theorem~\ref{THEO1}. In particular, we do not require that the sets $\mE'$ are antichain cuts (see also Remark
\ref{noantichain}).

\begin{corollary}[Singleton-type and Hamming-type bound] \label{corohasi}
Let $t,e \ge 0$ be integers, and assume $\adv=\AP{t}{e}{\mE}$. Then for all 
$(\alpha_1,...,\alpha_N) \in \mR(\mN,\adv)$ and for all non-empty $J \subseteq I$ we have 
$$\displaystyle \sum_{i \in J} \alpha_i  \le \min_{T \in {\bf T}}  \, \max \left\{ 0,  \mu({\bf S}_J, T)-2t-e \right\}$$
and
$$\displaystyle \sum_{i \in J} \alpha_i  \le  \min_{T \in {\bf T}} \,   \max \left\{ 0,  {\mu({\bf S}_J,T)} - \log_{|\mA|} \left(\displaystyle \mathlarger{\sum}_{h=0}^{t'} \binom{\mu({\bf S}_J,T)}{h} (|\mA|-1)^h \right) \right\}, \quad\mbox{ where $t':= \lfloor t+e/2 \rfloor$}.$$
\end{corollary}

\begin{proof}
The result can be shown by combining Theorem~\ref{THEO1} and the well known Singleton bound and Hamming bound from classical Coding Theory (see~\cite{plessbook} for a general reference). Note that the proofs of these two bounds do not require the alphabet $\mA$ to be a finite field.
\end{proof}

\begin{example} \label{1shotexample}
Let $\mN$ be the network in Figure~\ref{fig1shotexample}. The edges are ordered according to their indices. Consider the adversary
$\adv:=\adv_{1,0} \langle \mE \rangle$ who can corrupt at most one of
the values of the edges from $\mE$. 
Let $\mA$ be the network alphabet, and let $(\alpha_1,\alpha_2) \in \mR(\mN,\adv)$ be arbitrary. 
\begin{figure}[htbp]
\centering
\begin{tikzpicture}
\tikzset{vertex/.style = {shape=circle,draw,inner sep=0pt,minimum size=1.9em}}
\tikzset{nnode/.style = {shape=circle,fill=myg,draw,inner sep=0pt,minimum
size=1.9em}}
\tikzset{edge/.style = {->,> = stealth}}
\tikzset{dedge/.style = {densely dotted,->,> = stealth}}
\tikzset{ddedge/.style = {dashed,->,> = stealth}}

\node[vertex] (S1) at  (0,0) {$S_1$};
\node[vertex,below=0.75\mynodespace of S1] (S2) {$S_2$};
\node[nnode,right=\mynodespace of S1] (V1)  {$V_1$};
\node[nnode,right=\mynodespace of S2] (V2)  {$V_2$};
\coordinate (A) at (barycentric cs:V1=0.5,V2=0.5);
\node[vertex,right=\mynodespace of A] (T)  {$T$};

\draw[dedge,bend left=10] (S1)  to node[fill=white,inner sep=1pt,sloped]{\small $e_1$} (V1);
\draw[dedge,bend right=10] (S1)  to node[fill=white,inner sep=1pt,sloped]{\small $e_2$} (V1);
\draw[dedge] (S1)  to node[near start,fill=white,inner sep=1pt,sloped]{\small $e_3$} (V2);

\draw[dedge] (S2)  to node[near start,fill=white,inner sep=1pt,sloped]{\small $e_4$} (V1);
\draw[dedge,bend left=10] (S2)  to node[fill=white,inner sep=1pt,sloped]{\small $e_5$} (V2);
\draw[dedge,bend right=10] (S2)  to node[fill=white,inner sep=1pt,sloped]{\small $e_6$} (V2);

\draw[dedge,bend left=10] (V1)  to node[fill=white,inner sep=1pt,sloped]{\small $e_7$} (T);
\draw[dedge,bend right=10] (V1)  to node[fill=white,inner sep=1pt,sloped]{\small $e_8$} (T);

\draw[dedge,bend left=10] (V2)  to node[fill=white,inner sep=1pt,sloped]{\small $e_9$} (T);
\draw[dedge,bend right=10] (V2)  to node[fill=white,inner sep=1pt,sloped]{\small $e_{10}$} (T);
\end{tikzpicture} 
\caption{{{Network for Example~\ref{1shotexample}.}}}\label{fig1shotexample}
\end{figure}
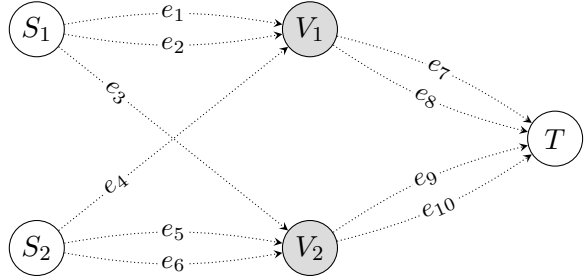

Using the Singleton-Type bound of Corollary~\ref{corohasi} with
$I=\{1\}$ we get $\alpha_1 \le 1$. Applying the same bound with
$I=\{2\}$ and $I=\{1,2\}$ one obtains, respectively, $\alpha_2 \le 1$ and
$\alpha_1 + \alpha_2 \le 2$. Therefore
$\mR(\mN,\adv) \subseteq \{(\alpha_1,\alpha_2) \in \R_{\ge 0}^2 : 
\alpha_1 \le 1,\, \alpha_2 \le 1, \,\alpha_1+\alpha_2 \le 2\}$
for any alphabet $\mA$.

Now assume $\mA=\F_2$. Applying in the same way the Hamming-Type bound  we obtain 
\begin{equation} \label{eqhamm}
\mR(\mN,\adv) \subseteq \{(\alpha_1,\alpha_2) \in \R_{\ge 0}^2 : \alpha_1 \le 1, \,\alpha_2 \le 1, \,\alpha_1+\alpha_2 \le 4-\log_2 5\} \quad \mbox{for $\mA=\F_2$}.
\end{equation}
Using the definition of one-shot capacity region (Definition~\ref{netcap}), it is easy to see that (\ref{eqhamm}) implies that
$\mR(\mN,\adv) \subseteq \{(0,0), (1,0), (0,1)\}$ for $\mA=\F_2$.
Therefore for $\mA=\F_2$ the Hamming-Type Bound is better than the Singleton-type Bound. 
\end{example}

\subsection{Other Capacity Regions}

We now consider multiple adversaries acting on pairwise disjoint sets of network edges, with different error and erasure powers. In this context, the three capacity regions defined in Subsection~\ref{subcapacityregions} can be upper-bounded by combining the Compound Porting Lemma with the results established in Section~\ref{sechammingtype}.

\begin{theorem} \label{THEO2}
Fix an integer $L \ge 1$. Let $\mU_1,...,\mU_L \subseteq \mE$ be pairwise disjoint subsets of edges, and let $t_1,...,t_L \ge 0$ and 
$e_1,...,e_L \ge 0$ be non-negative integers. Define $\bm{\mU}:=(U_1,...,U_L)$, 
$\bm t:=(t_1,...,t_L)$ and $\bm e:=(e_1,...,e_L)$, and assume
$\adv=\AB{t}{e}{\mU}$.
Then for all $(\alpha_1,...,\alpha_N) \in 
\overline{\regzr(\mN,\adv)}$ and for all non-empty subset $J \subseteq I$ we have
\begin{equation} \label{eq1THEO2}
\sum_{i \in J} \alpha_i  \le  \min_{T \in {\bf T}}  \, \min \left\{ 
|\mE'| - \sum_{\ell=1}^L \min \left\{2t_\ell+e_\ell, |\mE' \cap \mU_\ell|\right\}  : \mE' \subseteq \mE \mbox{ is a cut between ${\bf S}_J$ and $T$}\right\}.
\end{equation}
In particular, the inequality in (\ref{eq1THEO2}) holds for all non-empty  $J \subseteq I$ 
and for any $(\alpha_1,...,\alpha_N) \in 
\overline{\mR_{0}(\mN,\adv)}$ and  
$(\alpha_1,...,\alpha_N) \in 
{\mR(\mN,\adv)}$.
\end{theorem}

\begin{proof} It suffices to combine Theorem~\ref{mainthhamming} and Lemma~\ref{portinglemma1}. 
The second part of the statement follows from the fact that 
$\mR(\mN,\adv) \subseteq \overline{\regz(\mN,\adv)} \subseteq \overline{\regzr(\mN,\adv)}$,
as observed in Remark~\ref{rmkregions}.
\end{proof}

We now show how to apply 
Theorem~\ref{THEO2} in a concrete example.

\begin{example} \label{exfinale}
Let $\mN$ be the network in Figure~\ref{figexfinale}. The network edges are ordered according to their indices. We consider two adversaries associated
with the sets of edges $\mU_1=\{e_5,e_6,e_7\}$ and 
$\mU_2=\{e_1,e_8,e_9,e_{10},e_{11},e_{12}\}$. 
Both the adversaries have error power 1 and erasure power 0. 
In our notation, the adversary is therefore $\adv=\AB{t}{e}{\mU}$, 
where $\bm{\mU}=(\mU_1,\mU_2)$, $\bm{t}=(1,1)$ and $\bm{e}=(0,0)$. 

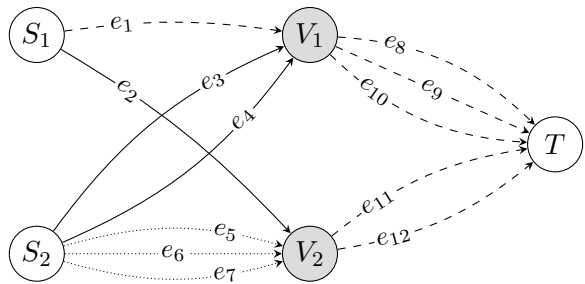
\begin{figure}[htbp]
\centering
\begin{tikzpicture}
\tikzset{vertex/.style = {shape=circle,draw,inner sep=0pt,minimum size=1.9em}}
\tikzset{nnode/.style = {shape=circle,fill=myg,draw,inner sep=0pt,minimum
size=1.9em}}
\tikzset{edge/.style = {->,> = stealth}}
\tikzset{dedge/.style = {densely dotted,->,> = stealth}}
\tikzset{ddedge/.style = {dashed,->,> = stealth}}

\node[vertex] (S1) {$S_1$};
\node[vertex,below=0.75\mynodespace of S1] (S2) {$S_2$};

\node[nnode,right=\mynodespace of S1] (V1) {$V_1$};
\node[nnode,right=\mynodespace of S2] (V2) {$V_2$};
\coordinate (A) at (barycentric cs:V1=0.5,V2=0.5);
\node[vertex,right=\mynodespace of A] (T) {$T$};

\draw[ddedge,bend left=8] (S1)  to node[near start,sloped,fill=white, inner sep=1pt]{\small $e_1$} (V1);

\draw[edge,bend left=8] (S1) to  node[near start,sloped,fill=white, inner sep=1pt]{\small $e_2$} (V2);

\draw[edge,bend left=15] (S2)  to node[near end,sloped,fill=white, inner sep=1pt]{\small $e_3$} (V1);

\draw[edge,bend right=15] (S2) to  node[near end,sloped,fill=white, inner sep=1pt]{\small $e_4$} (V1);

\draw[dedge,bend left=16] (S2)  to node[near end,sloped,fill=white, inner sep=1pt]{\small $e_5$} (V2);

\draw[dedge] (S2)  to node[sloped,fill=white, inner sep=1pt]{\small $e_6$} (V2);

\draw[dedge,bend right=16] (S2)  to node[near end,sloped,fill=white, inner sep=1pt]{\small $e_7$} (V2);

\draw[ddedge,bend left=20] (V1)  to node[near start,sloped,fill=white, inner sep=1pt]{\small $e_8$} (T);

\draw[ddedge] (V1)  to node[sloped,fill=white, inner sep=1pt]{\small $e_9$} (T);

\draw[ddedge,bend right=20] (V1)  to node[near start,sloped,fill=white, inner sep=1pt]{\small $e_{10}$} (T);

\draw[ddedge,bend left=15] (V2)  to node[near start,sloped,fill=white, inner sep=1pt]{\small $e_{11}$} (T);

\draw[ddedge,bend right=15] (V2)  to node[near start,sloped,fill=white, inner sep=1pt]{\small $e_{12}$} (T);

\end{tikzpicture} 
\caption{{{Network for Example~\ref{exfinale}.}}}\label{figexfinale}
\end{figure}

We want to describe the three capacity regions of
$(\mN,\adv)$. Let $(\alpha_1,...,\alpha_N) \in 
\overline{\regzr(\mN,\adv)}$ be arbitrary.
Applying Theorem~\ref{THEO2} with $I=\{1\}$ and $\mE'=\{e_1,e_2\}$ we obtain $\alpha_1 \le 2-1=1$. Using the same theorem with
$I=\{2\}$ and $\mE'=\{e_5,e_6,e_7,e_8,e_9,e_{10}\}$ we find 
$\alpha_2 \le 6-2-2=2$. Finally, applying again Theorem~\ref{THEO2}
with $I=\{1,2\}$ and $\mE'=\{e_2,e_5,e_6,e_7,e_8,e_9,e_{10}\}$ we obtain
$\alpha_1+\alpha_2 \le 3$.
Therefore by Remark~\ref{rmkregions} we deduce
$$\mR(\mN,\adv) \subseteq \overline{\regz(\mN,\adv)} \subseteq 
\overline{\regzr(\mN,\adv)} \subseteq \{(\alpha_1,\alpha_2) \in \R_{\ge 0}^2 : 
\alpha_1 \le 1, \,\alpha_2 \le 2, \,\alpha_1+\alpha_2 \le 3\}.$$

We now explicitly give a communication scheme 
that achieves the rate $(\alpha_1,\alpha_2)=(1,2)$ in one shot
for $\mA=\F_5$. It can be shown that such a scheme exists for
$\mA=\F_q$ whenever $q \ge 4$.

We start by assigning the network code functions. Since erasures 
are excluded from the model, it suffices to assign functions
$\mF_{V_1}: \F_5^3 \to \F_5^3$ and $\mF_{V_2}: \F_5^4 \to \F_5^2$.
For all $(x_1,x_3,x_4) \in \F_5^3$, define $\mF_{V_1}(x_1,x_3,x_4):=(x_1+2x_3+3x_4,x_3,x_4)$. To define $\mF_2$,
we first select any function $f:\F_5^3 \to \F_5$ with the following property:
for all $(x_5,x_6,x_7) \in \F_5^3$ and $x \in \F_5$,  
$f(x_5,x_6,x_7)=x$ if at least two among $x_5,x_6,x_7$ equal $x$.
The function $f$ is an arbitrary extension of a ``majority vote'' decoding function. Now for all 
$(x_2,x_5,x_6,x_7) \in \F_5^4$ define $\mF_{V_2}(x_2,x_5,x_6,x_7):=(x_2+f(x_5,x_6,x_7),2x_2)$.

Let us construct the local codes $\mC_1 \subseteq \F_5^2$ and $\mC_2 \subseteq \F_5^5$ for the two sources
$S_1$ and $S_2$.
We take $\mC_1:=\{(a,3a) : a \in \F_5\}$ as code for source $S_1$, and 
$\mC_2:=\{(b,c,2b+c,2b+c,2b+c) : (b,c) \in \F_5 \times \F_5\}$ as code for source $S_2$. 
Note that $\log_5 |\mC_1|=1=\alpha_1$ and $\log_5 |\mC_2|=2=\alpha_2$.

It is easy to see that for any transmitted $(\overline{a},3\overline{a}) \in \mC_1$ and 
$(\overline{b},\overline{c},2\overline{b}+\overline{c},2\overline{b}+\overline{c},2\overline{b}+\overline{c}) \in \mC_2$,
 terminal $T$ observes over the incoming edges the vector
\begin{equation}\label{ricev}
(\overline{a}+2\overline{b}+3\overline{c}, \,\overline{b}, \,\overline{c}, \,3\overline{a}+2\overline{b}+\overline{c},\,\overline{a})  +  w,
\end{equation}
where $w \in \F_5^5$ is a vector of Hamming weight at most $1$. Note moreover that
$$(\overline{a}+2\overline{b}+3\overline{c}, \,\overline{b}, \,\overline{c}, \,3\overline{a}+2\overline{b}+\overline{c},\,\overline{a}) = 
(\overline{a},\overline{b},\overline{c}) \cdot G, \quad \mbox{ where } 
G:=\begin{pmatrix}  1& 0 & 0 & 3 & 1 \\ 2 & 1 & 0 & 2 & 0 \\ 3 & 0 & 1 & 1 & 0\end{pmatrix} \in \F_5^{3 \times 5}.$$
One can check that $G$ is the generator matrix of a
 code $\mD \subseteq \F_5^5$ of dimension 3 
and minimum Hamming distance 3.
Thus $T$ can recover $\overline{a}, \overline{b}$ and $\overline{c}$ from (\ref{ricev}) using a minimum distance decoder for $\mD$.
\end{example}

Theorem 
\ref{THEO2} implies the following corollary describing special types of adversaries.

\begin{corollary}\label{corosingle}  
Let $t \ge 0$ be an integer, and assume $\adv=\AP{t}{0}{\mE}$. For all 
$(\alpha_1,...,\alpha_N) \in \overline{\regzr(\mN,\adv)}$ and for all non-empty $J \subseteq I$ we have 
$$\displaystyle \sum_{i \in J} \alpha_i  \le  \min_{T \in {\bf T}} \,  \max \left\{ 0, \,\mu({\bf S}_J, T)-2t \right\}.$$
In particular, if $\mN$ is adversary-free then for all 
$(\alpha_1,...,\alpha_N) \in \overline{\regzr(\mN,\adv)}$ and all $\emptyset \neq J \subseteq I$ we have 
$$\displaystyle \sum_{i \in J} \alpha_i  \le  \min_{T \in {\bf T}}  \mu({\bf S}_J, T).$$
Moreover, the two bounds above hold for all non-empty  $J \subseteq I$ 
and for any $(\alpha_1,...,\alpha_N) \in 
\overline{\mR_{0}(\mN,\adv)}$ and  
$(\alpha_1,...,\alpha_N) \in 
{\mR(\mN,\adv)}$.
\end{corollary}

We conclude this section studying the scenario where the network alphabet is of the form $\mA=\mB^m$, where 
$\mB$ is a finite set with $|\mB| \ge 2$ and $m \ge 2$ is an integer.
Consider an adversary $\adv$ who can erase up to $e$ components and corrupt up to $t$ components 
of each edge of the network, where $t,e \ge 0$ are non-negative integers.
We denote such an adversary by $\adv_{t,e} \langle \mB,m \rangle$. 
Combining Theorem~\ref{productalphabetthm} and 
Lemma~\ref{portinglemma} one easily obtains the following result.

\begin{theorem}\label{bndproduct}
Let $t,e \ge 0$ be non-negative integers, and assume
$\adv=\AP{t}{e}{\mB,m}$.
Then for all $(\alpha_1,...,\alpha_N) \in 
\overline{\mR_{0}(\mN,\adv)}$ and for all non-empty subset $J \subseteq I$ we have
\begin{equation} \label{eq1bndproduct}
m \cdot \sum_{i \in J} \alpha_i  \le  \min_{T \in {\bf T}}  \mu({\bf S}_J,T) \cdot \max\{0, \, m-2t-e\}.
\end{equation}
In particular, the inequality in (\ref{eq1bndproduct}) holds for all non-empty  $J \subseteq I$ 
and for all $(\alpha_1,...,\alpha_N) \in 
{\mR(\mN,\adv)}$.
\end{theorem}

%%%%%%%%%%%%%%%%%%%%%%%%%%%%%%%%%%%

\section{Capacity Regions: Constructions} \label{secconstructions}

This section studies the achievability of some of the upper bounds established in 
Section~\ref{secbounds}, both for one-shot and (compound) zero-error models. 
We describe communication schemes that achieve any integer point of the capacity region
of certain adversarial networks, most of which are inspired by
previously proposed approaches (in particular, by~\cite{MANIAC} and~\cite{koettermedard}).
We also show that, for the case of restricted adversaries, linear network coding does not suffice in 
general
to achieve every point in the capacity region of adversarial networks.
In the sequel we follow the notation of Section~\ref{secbounds}.

\subsection{Adversary-Free Scenario}
We start by investigating the adversary-free scenario, showing how in this case
the algebraic approach of~\cite{koettermedard} can be extended to achieve, 
in one shot, any integer point in the rate region 
described by the bound of Corollary~\ref{corosingle}, over sufficiently large fields.

\begin{remark}
Let  $(a_1,...,a_N) \in \N^N$ be an integer vector that satisfies $\sum_{i \in J} a_i  \le  \min_{T \in {\bf T}}  \mu({\bf S}_J, T)$
for all  $\emptyset \neq J \subseteq I$ (cf. Corollary~\ref{corosingle}).
To show that $(a_1,...,a_N)$ can be achieved in one shot in an adversary-free context, it does not suffice to
directly apply the approach of~\cite{koettermedard} to the network obtained
from $\mN$
by adding a super-source connected to all the sources $S_1,...,S_N$ with a sufficiently large number of edges.
Indeed, this approach would only show that there exists a communication scheme that allows 
the \textit{set} of sources $\{S_1,...,S_N\}$ to 
transmit to all the terminals at a \textit{global rate} of $a_1+a_2+ \cdots +a_N$. However, 
 such a scheme does not necessarily allow
each source $S_i$ to transmit at the prescribed rate $a_i$, for all $i \in I$. In other words, in the notation 
of Definition~\ref{netcap}, such a scheme does not necessarily induce a global code $\mC$ for the sources that decomposes as a cartesian product, defining local codes $\mC_1,...,\mC_N$.
This issue can be solved by extending the algebraic approach 
of~\cite{koettermedard} to multi-source networks via the following Graph Theory lemma.
\end{remark}

\begin{lemma} \label{lemmagt}
Let $(a_1,...,a_N) \in \N^N$ be an integer $N$-tuple such that
$\sum_{i \in J} a_i  \le  \min_{T \in {\bf T}}  \mu({\bf S}_J,T)$
for all non-empty  $J \subseteq I$.
Then for each terminal $T \in {\bf T}$ there exist $a_1+a_2+ \cdots + a_N$ edge-disjoint directed paths connecting 
${\bf S}=\{S_1,...,S_N\}$ to $T$,  $a_i$ of which originate in $S_i$ for all $i \in I$.
\end{lemma}

\begin{proof}
We start by adding a vertex $S \notin \mV$ to the digraph $\mG:=(\mV,\mE)$. For all $i \in I$, connect $S$ to $S_i$
with exactly $a_i$ directed edges. This defines a new digraph $\mG':=(\mV',\mE')$. Fix any terminal $T \in {\bf T}$. 
By the edge-connectivity version of Menger's Theorem,
it suffices to show that the minimum size of an edge-cut between $S$ and $T$ (in the graph $\mG'$) is $a:=a_1+a_2+\cdots +a_N$, i.e., that
$\mu(S,T)=a$, where this time $\mu$ denotes the min-cut in the new graph $\mG'$.

It is clear that $\mu(S,T) \le a$. Assume by contradiction that there exists a cut $\Gamma \subseteq \mE'$ 
with $|\Gamma|<a$ that separates $S$ from $T$.
For all $i \in I=\{1,...,N\}$, denote by $\mE_i$ the set of directed edges connecting $S$ to $S_i$, and let $n_i:=|\Gamma \cap \mE_i|$. Since
$\mE_i \cap \mE_j = \emptyset$ for all $i,j \in I$ with $i \neq j$, we have 
\begin{equation} \label{idimp}
\sum_{i \in I} n_i = \sum_{i \in I}|\Gamma \cap \mE_i| \le |\Gamma| < a=\sum_{i \in I} a_i.
\end{equation}
Moreover, for all $i \in I$ one clearly has $n_i =|\Gamma \cap \mE_i| \le |\mE_i|=a_i$. Define  
$J:=\{i \in I : n_i<a_i\} \subseteq I$.

If $J=\emptyset$ then $a_i=n_i$ for all $i \in I$. This contradicts (\ref{idimp}). Now assume $J \neq \emptyset$,
and observe that for all $i \in J$ the cut $\Gamma$ does not disconnect $S$ from $S_i$. Hence 
$\Gamma$ must be a cut between ${\bf S}_J$ and  $T$. As a consequence, the set
$\Gamma \setminus \bigcup_{i \in I \setminus J} \mE_i$
is a cut between ${\bf S}_J$ and $T$. Therefore 
\begin{equation*}\mu({\bf S}_{J},T)  \le  \left| \Gamma \setminus \bigcup_{i \in I \setminus J} \mE_i  \right|  =  |\Gamma| - 
\sum_{i \in I \setminus J} n_i  = |\Gamma| - \sum_{i \in I \setminus J} a_i  <  a -\sum_{i \in I \setminus J} a_i
 =  \sum_{i \in J} a_i,
\end{equation*} a contradiction. This concludes the proof.
\end{proof}

Lemma~\ref{lemmagt} can now be employed to extend the approach of~\cite{koettermedard} from one to multiple sources,
obtaining the following achievability result for adversary-free 
networks. 

\begin{theorem} \label{THM1}
Assume $\adv= \AP{0}{0}{\mE}$. We have
$$\mR(\mN,\adv) \supseteq 
\left\{(a_1,...,a_N) \in \N^N \, :\, \sum_{i \in J} a_i  \le   
\min_{T \in {\bf T}}  \mu({\bf S}_J, T) \mbox{ for all $\emptyset \neq J \subseteq I$}\right\},$$
provided that $\mA=\F_q$ and $q$ is sufficiently large. Moreover,
every integer rate vector $(a_1,...,a_N)$ as above can be achieved in one shot employing 
linear network coding.
\end{theorem}

\begin{proof}
Let $(a_1,...,a_N) \in \N^N$ be an integer vector such that 
$\sum_{i \in J} a_i \le   
\min_{T \in {\bf T}}  \mu({\bf S}_J, T)$ for all $\emptyset \neq J \subseteq I$. By Lemma~\ref{lemmagt}, 
for each terminal $T \in {\bf T}$ there exist
$a_1+ \cdots + a_N$ edge-disjoint paths connecting ${\bf S}$ to
$T$, of which $a_i$ originate in $S_i$ for all $i \in I$. Without loss of generality, 
we may assume that there is neither an edge nor a vertex of $\mN$ which is not on at least one 
of these $|{\bf T}| \cdot (a_1+ \cdots + a_N)$ paths. In particular,  source $S_i$
has at least $a_i$ outgoing edges for all $i \in I$, and every terminal 
$T$ has exactly $a:=a_1+ \cdots + a_N$ incoming edges
(we are using property~\ref{prnF} of Definition~\ref{defnetwork}).

For all $i \in I$, we let the local code of source $S_i$ 
to be of the form 
$\mC_i:=\{v \cdot E_i : v \in \F_q^{a_i}\}$,
where $E_i$ is a matrix of size $a_i \times |\out(S_i)|$ over $\F_q$ to 
be determined. 

Since erasures are excluded from the model, to
construct a linear network code 
$\mF$ for $\mN$ it suffices to assign
an $(|\inn(V)| \times |\out(V)|)$-matrix 
$\mF_V$ over $\F_q$ for every $V \in \mV \setminus ({\bf S} \cup {\bf T})$. See 
Definition~\ref{defK}.

In analogy with~\cite{koettermedard}, we introduce a variable for each entry of each of the matrices
$E_i$ and $\mF_V$, for $i \in I$ and $V \in \mV \setminus ({\bf S} \cup {\bf T})$. We denote these variables
by $\zeta_1,...,\zeta_M$, and let $\zeta:=(\zeta_1,...,\zeta_M)$.
In the sequel, the $E_i$'s and $\mF_V$'s are denoted by 
$E_i(\zeta)$ and $\mF_V(\zeta)$, to stress the dependency on $\zeta$.
An evaluation  
$\overline{\zeta}=(\overline{\zeta}_1,...,\overline{\zeta}_M) \in \F_q^M$ of the variables
induces matrices $E_i(\overline{\zeta})$, for $i \in I$, and a linear network code
$\mF(\overline{\zeta})$ for $\mN$. Moreover, 
for all $T \in {\bf T}$ and for all 
$(x_1,...,x_N) \in \F_q^{a_1} \times \cdots \times \F_q^{a_N}$ we have
$$\Omega_{\mF(\overline{\zeta})}[{\bf S} \to \inn(T)](x_1 \cdot E_1(\overline{\zeta}) ,...,x_N \cdot E_N(\overline{\zeta})) 
= (x_1,...,x_N) \cdot M_T(\overline{\zeta})$$
for some $a \times a$ transfer matrix $M_T(\overline{\zeta})$  
whose entries are polynomials in $\zeta_1,...,\zeta_M$ evaluated in 
$\overline{\zeta}$.  Note that the matrix $M_T(\overline{\zeta})$ has size $a \times a$ by property \ref{prnF} of Definition \ref{defnetwork}.

By definition of one-shot capacity region, to prove the theorem it suffices to show that,
for a large enough $q$,
there exists an evaluation $\overline{\zeta}$ of the variables such that
each matrix $M_T(\overline{\zeta})$, $T \in {\bf T}$, is invertible. As in~\cite{koettermedard}, 
this fact follows
from the Sparse Zeros Lemma (see, e.g.,~\cite[Lemma 1]{netbook}) and the existence 
of the routing solutions. The details of this part of the argument are left to the reader. 

The invertibility of the 
$M_T(\overline{\zeta})$'s implies that each 
$E_i(\overline{\zeta})$, $i \in I$, is injective as a linear map.
This shows that  $\log_{|\mA|}|\mC_i|=a_i$ for all $i \in I$,  
defining local codes for the sources of the expected cardinalities.
\end{proof}

\subsection{Single Error-Adversary With Access to All Edges}
In this subsection we focus on adversaries 
of the form $\adv=\AP{t}{0}{\mE}$, and show that every integer vector
$(a_1,...,a_N) \in \N^N$ in the region described by Corollary~\ref{corosingle} can be achieved in one shot,
provided that $\mA=\F_q^m$, and that $q$ and $m$ are sufficiently large. Our scheme is a simple modification of an idea from~\cite{MANIAC}, in which the authors propose a scheme in the context of
random linear network coding. 

Note that the following Theorem~\ref{ACHIEV1} 
cannot be directly obtained from classical results in the context of 
network communications under probabilistic error/erasure models. 

\begin{theorem} \label{ACHIEV1}
Let $t \ge 0$ be an integer, and assume $\adv=\AP{t}{0}{\mE}$. We have
$$\mR(\mN,\adv) \supseteq 
\left\{(a_1,...,a_N) \in \N^N \,: \, \sum_{i \in J} a_i   \le  \min_{T \in {\bf T}}   \,
\max \left\{ 0, \,\mu({\bf S}_J, T)-2t \right\}\mbox{ for all $\emptyset \neq J \subseteq I$}
 \right\},$$
provided that $\mA=\F_q^m$, $q$ is sufficiently large, and $m= \prod_{i \in I} (a_i+2t)$. Moreover,
every integer rate vector $(a_1,...,a_N)$ as above can be achieved in one shot employing 
linear network coding.
\end{theorem}

\begin{proof}
We will show the theorem only for $N=2$ sources and in the case where, for any terminal $T \in {\bf T}$,
we have $\mu(S_1,T), \mu(S_2,T), \mu({\bf S},T) > 2t$. 

Fix a pair $(a_1,a_2) \in \N^2$ with $a_1 \le \mu(S_1,T)-2t$,
$a_2 \le \mu(S_2,T)-2t$ and $a_1+a_2 \le \mu({\bf S},T)-2t$
for all $T \in {\bf T}$. 
Observe that $a_1 \le \mu(S_1,T)$, $ a_2+2t \le \mu(S_2,T)$, and 
$a_1+(a_2+2t) \le \mu({\bf S},T)$.
By Lemma~\ref{lemmagt}, 
for each $T \in {\bf T}$ there exist
$a_1+a_2+2t$ edge-disjoint paths connecting ${\bf S}$ to
$T$, of which $a_1$ originate in $S_1$, and $a_2+2t$ originate in $S_2$. Moreover, 
as $a_1+2t \le \mu(S_1,T)$, there exist $a_1+2t$ edge-disjoint paths connecting $S_1$ to
$T$. Without loss of generality, 
we may assume that there is neither an edge nor a vertex of $\mN$ which is not on at least one 
of these $|{\bf T}|\cdot (a_1+a_2+2t+a_1+2t) = |{\bf T}|\cdot (2a_1+a_2+4t)$ paths.
Finally, for ease of notation define:
$$n_1:=a_1+2t, \quad  n_2:=a_2+2t,  \quad q_1:=q^{n_1},  \quad q_2:=q_1^{n_2},  \quad b_1:=|\out(S_1)|,  \quad b_2:=|\out(S_2)|.$$

Before describing a communication scheme that achieves $(a_1,a_2)$ in one shot, following~\cite{MANIAC} 
we introduce some maps
needed in the sequel. Fix finite fields $\F_q \subseteq \F_{q_1} \subseteq \F_{q_2}$ and bases
$\{\beta_1^1,...,\beta_1^{n_1}\}$, $\{\beta_2^1,...,\beta_2^{n_2}\}$ for $\F_{q_1}$ and $\F_{q_2}$
over $\F_q$ and $\F_{q_1}$, respectively. We denote by $\varphi_1: \F_{q_1} \to \F_q^{n_1 \times 1}$
the $\F_q$-isomorphism that expands an element of $\F_{q_1}$ over the basis
$\{\beta_1^1,...,\beta_1^{n_1}\}$. Similarly, we denote by $\varphi_2: \F_{q_2} \to \F_{q_1}^{n_2 \times 1}$
the $\F_{q_1}$-isomorphism that expands an element of $\F_{q_2}$ over the basis
$\{\beta_2^1,...,\beta_2^{n_2}\}$. Extend the maps $\varphi_1$ and $\varphi_2$ entry-wise
to matrices or arbitrary size over $\F_{q_1}$ and $\F_{q_2}$, respectively. 
Note that the entries of matrices are always expanded as column vectors.

We now describe the communication scheme. Set $m:=n_1n_2$, so that $\mA=\F_q^{n_1n_2}$
is the alphabet.
Let $G_1 \in \F_{q_1}^{a_1 \times n_1}$ be the generator matrix 
of a rank-metric code $\mD_1 \subseteq \F_{q_1}^{n_1}$ with 
$\dim_{\F_{q_1}}(\mD_1)=a_1$ and minimum rank distance $2t+1$ over $\F_q$
(see~\cite{gabidulin}).
Similarly, let $G_2 \in \F_{q_2}^{a_2 \times n_2}$ be the generator matrix 
of a rank-metric code $\mD_2 \subseteq \F_{q_2}^{n_2}$ with 
$\dim_{\F_{q_2}}(\mD_2)=a_2$ and minimum rank distance $2t+1$ over $\F_{q_1}$.

Source $S_1$ chooses an arbitrary matrix $X_1 \in \F_{q_1}^{n_2 \times a_1}$, computes
$X_1G_1 \in \F_{q_1}^{n_2 \times n_1}$, and sends over the outgoing edges the columns of the matrix
$\varphi_1(X_1G_1)E_1 \in \F_q^{n_1n_2 \times b_1}$,
where $E_1 \in \F_q^{n_1 \times b_1}$ is a local encoding matrix to be determined. Source $S_2$ chooses an arbitrary  $X_2 \in \F_{q_2}^{1 \times a_2}$, computes
$X_2G_2 \in \F_{q_2}^{1 \times n_2}$, and sends over the outgoing edges the columns of the matrix
$\varphi_1(\varphi_2(X_2G_2))E_2 \in 
\F_q^{n_1n_2 \times b_2}$,
where $E_2 \in \F_q^{n_2 \times b_2}$ is another local encoding matrix that will be chosen later in the proof.

The vertices in 
$\mV\setminus({\bf S} \cup {\bf T})$ process the incoming packets using a linear network code $\mF$. According to Definition~\ref{defK},
we therefore need to assign to every  $V \in \mV\setminus({\bf S} \cup {\bf T})$
an $|\inn(V)| \times |\out(V)|$ matrix $\mF_V$ over $\F_q$ (erasures are excluded from the model).

As in the proof of Theorem~\ref{THM1}, we introduce a variable for each entry of each of the matrices
$E_i$ and $\mF_V$, for $i \in \{1,2\}$ and $V \in \mV \setminus ({\bf S} \cup {\bf T})$. We denote these variables
by $\zeta_1,...,\zeta_M$, and let $\zeta:=(\zeta_1,...,\zeta_M)$.
In the sequel, the $E_i$'s and $\mF_V$'s are denoted by 
$E_i(\zeta)$ and $\mF_V(\zeta)$.
Note that an evaluation  
$\overline{\zeta}=(\overline{\zeta}_1,...,\overline{\zeta}_M) \in \F_q^M$ of the variables
induces matrices $E_i(\overline{\zeta})$, for $i \in \{1,2\}$, and a linear network code
$\mF(\overline{\zeta})$ for $\mN$.  Moreover, 
for all $T \in {\bf T}$ and for all 
$X_1 \in \F_{q_1}^{n_2 \times a_1}$ and $X_2 \in \F_{q_2}^{1 \times a_2}$ we have 
\begin{multline*}
\Omega_{\mF(\overline{\zeta})}[{\bf S} \to \inn(T)](\varphi_1(X_1G_1)E_1(\overline{\zeta}), \varphi_1(\varphi_2(X_2G_2))E_2(\overline{\zeta})) \\ =
\varphi_1(X_1G_1)E_1(\overline{\zeta})\cdot M_T^1(\overline{\zeta}) + 
\varphi_1(\varphi_2(X_2G_2))E_2(\overline{\zeta})\cdot M_T^2(\overline{\zeta}),
\end{multline*}
where the packets are organized as column vectors,
and $M_T^1(\overline{\zeta})$, $M_T^2(\overline{\zeta})$ are the transfer matrices
from source $S_1$ and $S_2$, respectively, to the terminal $T$. These two matrices are well defined in the context of erasure-free linear network coding (see~\cite[Section~VII.B]{MANIAC}), and their entries are polynomials in $\zeta_1,...,\zeta_M$ evaluated in $\overline{\zeta}$. Note that the size of $M_T^i(\overline{\zeta})$
is $b_i \times |\inn(T)|$, for all $T \in {\bf T}$ and $i \in \{1,2\}$.

Using the existence of routing solutions and the Sparse Zeros Lemma (e.g.,~\cite[Lemma 1]{netbook}),
it can be shown that there exists an evaluation $\overline{\zeta} \in \F_q^M$ of the variables 
such that, for all $T \in {\bf T}$, the matrices 
\begin{equation} \label{twomatr}
A_T(\overline{\zeta}):=\begin{pmatrix}  G_1E_1(\overline{\zeta})M_T^1(\overline{\zeta}) \\ E_2(\overline{\zeta})M_T^2(\overline{\zeta})\end{pmatrix} \in \F_{q_1}^{(a_1+n_2) \times |\inn(T)|},\
\quad B_T(\overline{\zeta}):=\begin{pmatrix}  E_1(\overline{\zeta}) M_T^1(\overline{\zeta}) \end{pmatrix} \in \F_q^{n_1 \times |\inn(T)|}
\end{equation}
are both right-invertible (or equivalently full-rank), provided that $q$ is sufficiently large (for details about the matrix $A_T(\overline{\zeta})$, see the proof of~\cite[Lemma 2]{MANIAC}). In the sequel we fix such an evaluation 
$\overline{\zeta}$, and simply write $E_1$, $E_2$, $\mF$, $M_T^1$, $M_T^2$, $A_T$, $B_T$ for
$E_1(\overline{\zeta})$, $E_2(\overline{\zeta})$, $\mF(\overline{\zeta})$, $M_T^1(\overline{\zeta})$, $M_T^2(\overline{\zeta})$, $A_T(\overline{\zeta})$, $B_T(\overline{\zeta})$.
The local codes for sources $S_1$ and $S_2$ are, respectively, 
$$\mC_1=\{\varphi_1(X_1G_1)E_1 : X_1 \in \F_{q_1}^{n_1 \times a_1}\} \subseteq \F_q^{n_1n_2 \times b_1}, \quad
\mC_2=\{\varphi_1(\varphi_2(X_2G_2))E_2 : X_2 \in \F_{q_2}^{1 \times a_2}\}\subseteq \F_q^{n_1n_2 \times b_2},$$
where the network packets are again organized  as column vectors.
Since the matrices in (\ref{twomatr}) are both full-rank, the matrices $E_1$ and $E_2$ are full-rank as well (and thus right-invertible). As a consequence, we have
$|\mC_1|=|\F_{q_1}^{n_2 \times a_1}|=q^{ma_1}$ and 
$|\mC_2|=|\F_{q_2}^{1 \times a_2}|
=q^{ma_2}$. Therefore it remains to prove that $\mC_1 \times \mC_2$ is  good for 
each channel $\Omega_\mF[\adv;{\bf S} \to \inn(T)]$, $T \in {\bf T}$. We will show this by explicitly giving a decoding 
procedure. In the remainder of the proof the packets will be always organized as column vectors.

A given terminal $T \in {\bf T}$ receives
\begin{equation} \label{fff}
R_T=\varphi_1(X_1G_1)E_1M_T^1 + \varphi_1(\varphi_2(X_2G_2))E_2M_T^2 + Z_T \in \F_q^{n_1n_2 \times |\inn(T)|},
\end{equation}
where $Z$ is an error matrix such that $\mbox{rk}_{\F_q}(Z_T) \le t$ (see~\cite[Section IV]{onmetrics} for details).
Applying $\varphi_1^{-1}$ to both sides of (\ref{fff}), and using the fact that such map is $\F_q$-linear, the terminal computes
$$\varphi_1^{-1}(R_T) X_1G_1E_1M_T^1+ \varphi_2(X_2G_2)E_2M_T^2 + \varphi_1^{-1}(Z_T)
 = \begin{pmatrix} X_1 & \varphi_2(X_2G_2) \end{pmatrix} 
  A_T + \varphi_1^{-1}(Z_T).$$
%\begin{eqnarray*}
%\varphi_1^{-1}(R_T) &=& X_1G_1E_1M_T^1+ \varphi_2(X_2G_2)E_2M_T^2 + \varphi_1^{-1}(Z_T)
% \\ &=& \begin{pmatrix} X_1 & \varphi_2(X_2G_2) \end{pmatrix} 
% \cdot \begin{pmatrix} G_1E_1M_T^1 \\ E_2M_T^2 \end{pmatrix} + \varphi_1^{-1}(Z_T) \\
% &=& \begin{pmatrix} X_1 & \varphi_2(X_2G_2) \end{pmatrix} 
%  A_T + \varphi_1^{-1}(Z_T).
%\end{eqnarray*}
Terminal $T$ can now multiply on the right both members of the previous equality by the right-inverse of $A_T$, which is a matrix over $\F_{q_1}$ 
denoted by $A_T^{-1}$,
and obtain
\begin{equation} \label{interm}
\varphi_1^{-1}(R_T) A_T^{-1} = \begin{pmatrix} X_1 & \varphi_2(X_2G_2) \end{pmatrix} 
 + \varphi_1^{-1}(Z_T)A_T^{-1}.
 \end{equation}

 Observe that  
 $\mbox{rk}_{\F_{q_1}}(\varphi_1^{-1}(Z_T)A_T^{-1}) \le 
 \mbox{rk}_{\F_{q_1}}(\varphi_1^{-1}(Z_T)) \le \mbox{rk}_{\F_q}(Z_T)\le t$,
 where the second inequality follows from~\cite[Lemma 1]{MANIAC}.
 Therefore $T$ can delete the first $a_1$ columns of (\ref{interm}), and recover
 $X_2 \in \F_{q_2}^{1 \times a_2}$ using a minimum rank-distance decoder for the code generated by $G_2$.
 Now $T$ uses~(\ref{fff}) and computes
 \begin{equation}\label{fff1}
 \overline{R}_T:= R_T - \varphi_1(\varphi_2(X_2G_2))E_2M_T^2 = \varphi_1(X_1G_1)E_1M_T^1 + Z_T =
 \varphi_1(X_1G_1) B_T + Z_T.
 \end{equation}
 By our choice of $\overline{\zeta}$, $B_T$ is a right-invertible matrix over $\F_q$, whose right-inverse is denoted by
 $B_T^{-1}$. Multiplying on the right both sides of (\ref{fff1}) by $B_T^{-1}$, the terminal obtains
 \begin{equation}\label{final}
 \overline{R}_T B_T^{-1}= \varphi_1(X_1G_1) + Z_TB_T^{-1} \in \F_q^{n_1n_2 \times n_1}.
 \end{equation}
Define $\overline{X}:=\varphi_1(X_1G_1)$ and $\overline{Z}_T:=Z_TB_T^{-1}$. Organize
the $n_1n_2$ rows of $\overline{X}$ and $\overline{Z}_T$ in $n_2$ blocks of $n_1$ rows, and re-write
(\ref{final}) as
$$\overline{R}_T B_T^{-1} = \overline{X} + \overline{Z}_T = \begin{pmatrix}  \overline{X}^1 \\ \vdots \\ 
\overline{X}^{n_2} \end{pmatrix}   + \begin{pmatrix}  \overline{Z}_T^1 \\ \vdots \\ 
\overline{Z}_T^{n_2} \end{pmatrix}.$$
Since $\mbox{rk}_{\F_q}(Z_T) \le t$, we have $\mbox{rk}_{\F_q}(\overline{Z}_T^i) \le t$
for all $1 \le i \le n_2$. Moreover, $\varphi_1$,
 $\overline{X}^i= \varphi_1(X_1^iG_1)$, where
$X_1^i$ is the $i$-th row of $X_1$. Therefore $T$ can recover $\overline{X}$ by applying $n_2$ times a 
minimum rank-distance decoder for the code generated by $G_1$. Clearly, this allows $T$ to recover $X_1$ as well.  
\end{proof}

\subsection{A Scheme for the Compound Model}

In this subsection we adapt the scheme of Theorem~\ref{ACHIEV1} to 
the compound model, i.e., to the scenario where the adversary is forced to act on the 
same set of edges in each network use. We show that, in this specific context, the use of long network packets can be avoided by 
employing
the network multiple times. Note that this fact does not follow
directly from Theorem~\ref{ACHIEV1}, as a network alphabet 
of the form $\F_q^m$ does not model $m$ uses of the network
(see Remark~\ref{rmkisdifferent}). In fact, our scheme does not 
work if the adversary can act on a different set of edges in each channel use.

\begin{theorem}  \label{ACHIEV2}
Let $t \ge 0$ be an integer, and assume $\adv=\AP{t}{0}{\mE}$. We have
$$\overline{\regzr(\mN,\adv)} \supseteq 
\left\{(a_1,...,a_N) \in \N^N \,:\,\sum_{i \in J} a_i   \le   \min_{T \in {\bf T}}   \, 
\max \left\{ 0, \,\mu({\bf S}_J, T)-2t \right\}\mbox{ for all $\emptyset \neq J \subseteq I$}
 \right\},$$
provided that $\mA=\F_q^m$, $q$ is sufficiently large, and $m:=\min\{a_i : i \in I, a_i \neq 0\}+2t$. Moreover,
every integer rate vector $(a_1,...,a_N)$ as above can be achieved using 
linear network coding.
\end{theorem}

\begin{proof}
As for Theorem~\ref{ACHIEV1}, we show the result only for $N=2$ sources and in the case where, 
for any $T \in {\bf T}$,
we have $\mu(S_1,T), \mu(S_2,T), \mu({\bf S},T) > 2t$. Assume $a_1 \le a_2$ without loss of generality. 

In the sequel we follow the notation
of the proof of Theorem~\ref{ACHIEV1}, and modify the scheme to achieve $(a_1,a_2)$ in
$n_2$ channel uses, provided that $\mA=\F_q^{n_1}$ and $q$ is sufficiently large.
In every network use, each intermediate vertex $V \in \mV \setminus ({\bf S} \cup {\bf T})$ processes the incoming packets according to 
the network code $\mF$ constructed in the proof of Theorem~\ref{ACHIEV1} (thus $\mF$ is the same in every network use).

After choosing $X_1$ and $X_2$, sources $S_1$ and $S_2$ compute
$\varphi_1(X_1G_1), \varphi_1(\varphi_2(X_2G_2)) \in \F_{q}^{n_1n_2 \times n_1}$, respectively.
The rows of these matrices are then organized in $n_2$ blocks of
$n_1$ rows:
$$\varphi_1(X_1G_1) = \begin{pmatrix}  Y_1^1 \\ \vdots \\ Y_1^{n_2} \end{pmatrix}, \quad
\varphi_1(\varphi_2(X_2G_2)) = \begin{pmatrix}  Y_2^1 \\ \vdots \\ Y_2^{n_2} \end{pmatrix}.$$
In the $j$-th network use, $S_i$ sends over the outgoing edges the columns of
$Y_1^jE_i \in \F_q^{n_1 \times b_i}$, for $1 \le j \le n_2$ and  $i \in \{1,2\}$.
This defines codes $\mC_1$, $\mC_2$ for $S_1$ and $S_2$ of cardinality
$q^{n_1n_2a_1}$ and $q^{n_1n_2a_2}$, respectively.

The decoding is as follows. In the $j$-th network use, terminal $T \in {\bf T}$
collects packets from $\F_q^{n_1}$ over the incoming edges, and organizes them as the columns of a 
$n_1 \times |\inn(T)|$ matrix over $\F_q$,  which is denoted by $R_T^j$. 
Observe that
$R_T^j = Y_1^j E_1M_T^1 + Y_2^j E_2M_T^2 + Z_T^j$,
where $Z_T^j$ is the error matrix. Recall that $Z_T^j$ is defined by $Z_T^j: = W^j U_T$, where
$W^j \in \F_q^{n_1 \times |\mE|}$ is the matrix whose columns are the error packets, and 
$U_T$ is the $|\mE| \times |\inn(T)|$ transfer matrix from the edges of $\mN$ to the destination $T$.
Note that  $U_T$ is well defined in the context
of erasure-free linear network coding 
(see~\cite[Section I and IV]{onmetrics} for details).
After $n_2$ channel uses, terminal $T$ constructs the matrix 
\begin{equation} \label{fff22}
R_T:= \begin{pmatrix}  R_T^1 \\ \vdots \\ R_T^{n_2} \end{pmatrix} = \varphi_1(X_1G_1)E_1M_T^1 + 
\varphi_1(\varphi_2(X_2G_2))E_2M_T^2 +  Z_T, \quad \mbox{where }
Z_T:=\begin{pmatrix}  Z_T^1 \\ \vdots \\ Z_T^{n_2} \end{pmatrix}.
\end{equation}
Now observe that $Z_T$ can be written as
$$Z_T=\begin{pmatrix}  Z_T^1 \\ \vdots \\ Z_T^{n_2} \end{pmatrix} = WU_T, \quad
\begin{pmatrix} W^1 \\ \vdots \\ W^{n_2} \end{pmatrix} \in \F_q^{n_1n_2 \times |\mE|}.$$
Since the adversary acts on the same edges in each network use, the matrix $W$ has at most $t$ non-zero columns, which implies
$\mbox{rk}_{\F_q}(Z_T) \le \mbox{rk}_{\F_q}(W) \le t$ (this fact would not be true if the adversary was able to act on a different set 
of edges in each use of the channel). Decoding can therefore be completed starting from equation (\ref{fff22}) as in the proof of 
Theorem~\ref{ACHIEV1} (cf. equation (\ref{fff})).
\end{proof}

\subsection{Product Network Alphabets}
Throughout this subsection we assume that $\mA=\mB^m$ is a product alphabet, where 
$\mB$ is a finite set with $|\mB| \ge 2$ and $m \ge 2$ is a fixed integer. The following theorem describes a capacity-achieving scheme for adversarial networks of the form $(\mN, \adv)$,
where $\adv=\adv_{t,e} \langle \mB,m \rangle$ is the adversary 
of Theorem~\ref{bndproduct}.

\begin{theorem}
Let $t,e \ge 0$ be integers, and let $\adv=\adv_{t,e} \langle \mB,m \rangle$. Assume $m \ge 2t+e+1$, and let $k:=m-2t-e$. Then
$$\mR(\mN,\adv) \supseteq 
\left\{\frac{k}{m} (a_1,...,a_N) \in \R_{\ge 0}^N  : a_i \in \N \mbox{ for all $i \in I$, } \sum_{i \in J} a_i  \le  \min_{T \in {\bf T}}   
\mu({\bf S}_J, T) \mbox{ for all $\emptyset \neq J \subseteq I$}\right\},$$
provided that $\mB=\F_q$ and $q$ is sufficiently large. 
\end{theorem}

\begin{proof}
Let $(a_1,...,a_N) \in \N^N$ be an integer vector such that 
$\sum_{i \in J} a_i \le   
\min_{T \in {\bf T}}  \mu({\bf S}_J, T)$ for all non-empty $J \subseteq I$. 
By Lemma~\ref{lemmagt}, 
for each terminal $T \in {\bf T}$ there exist
$a_1+ \cdots + a_N$ edge-disjoint paths connecting ${\bf S}$ to
$T$, of which $a_i$ originate in $S_i$ for all $i \in I$. Without loss of generality, 
we may assume that there is neither an edge nor a vertex of $\mN$ which is not on at least one 
of these paths.
Moreover, we let $E_i(\overline{\zeta})$ and 
$\mF_V(\overline{\zeta})$, for $i \in I$ and $V \in \mV 
\setminus ({\bf S} \cup {\bf T})$, be as in the proof of
Theorem~\ref{THM1}.

We will give a communication scheme (and therefore construct a pair $(\mF,\mC)$) that achieves
the rate 
$k/m \cdot (a_1,...,a_N)$ in one shot.
Let 
$\mD \subseteq \F_q^m$ be a code of cardinality $q^k$ and minimum Hamming distance $d=2t+e+1$.
Let $\Enc:\F_q^k \to \F_q^m$ and $\Dec:({\F_q} \cup \{\star\})^m \to \F_q^k$
be an encoder and a decoder for $\mD$, respectively.
For all $i \in I$, let $b_i:=|\out(S_i)|$. Source $S_i$ chooses $x_i =(x_{i,1},...,x_{i,a_i}) \in (\F_q^{k})^{a_i}$, 
computes
$$\begin{pmatrix} x^\top_{i,1} & \cdots & x^\top_{i,a_i} \end{pmatrix} \cdot E_i(\overline{\zeta})=: \begin{pmatrix} y^\top_{i,1} & 
\cdots & y^\top_{i,b_i} \end{pmatrix} \in \F_q^{k \times b_i},$$
and sends $\Enc(y_{i,1}),...,\Enc(y_{i,b_i}) \in \F_q^m$ over the outgoing edges. Since the $E_i(\overline{\zeta})$'s are injective as linear maps 
(see the proof of Theorem~\ref{THM1}), this defines local codes
$\mC_1,...,\mC_N$ for the sources with $\mC_i \subseteq (\F_q^m)^{b_i}$ and 
$\log_{|\mA|}|\mC_i|=\log_{q^m}q^{ka_i} =k/m \cdot a_i$, for all
$i \in I$.

Let $V \in \mV \setminus ({\bf S} \cup {\bf T})$ be a vertex, and set
$r:=|\inn(V)|$, $s:=|\out(V)|$ for ease of notation.  The vertex $V$ collects $x_1,...,x_r \in (\F_q \cup \{\star\})^m$ over the incoming edges, and computes
$$\begin{pmatrix}\Dec(x_1)^\top & \cdots & \Dec(x_r)^\top \end{pmatrix} \cdot \mF_V(\overline{\zeta})=: \begin{pmatrix} 
y^\top_1 & \cdots & y^\top_s \end{pmatrix} \in \F_q^{m \times d}.$$
The vectors $\Enc(y_1),...,\Enc(y_s) \in \F_q^m$ are then sent over the outgoing edges of $V$.
This defines a network code $\mF$ for $\mN$ via
$\mF_V: (x_1,...,x_r) \mapsto (\Enc(y_1),...,\Enc(y_s))$, following the notation above.
It is easy to see that 
the pair $(\mF,\mC_1 \times \cdots \times \mC_N)$ achieves the rate
$k/m \cdot (a_1,...,a_N)$ in one shot.
\end{proof}

\subsection{Linear and Non-Linear Network Coding}

Theorem~\ref{ACHIEV1} shows that linear network coding suffices to achieve any integer point in the capacity region of an adversarial network of the form $(\mN,\AP{t}{0}{\mE})$, provided
that the network alphabet is of the form $\F_q^m$, with $q$ and $m$ sufficiently large.
We now show that this is not the case in general 
if the adversary has the form  $\adv=\AP{t}{0}{\mU}$, where
$\mU \subsetneq \mE$ is a proper subset of vulnerable edges.

%The following example shows a network where, in order to achieve capacity,
% an intermediate vertex needs to decode the packets received over the
%  incoming edges. This forces in particular the use of a non-linear 
%  network code at the intermediate vertex.

\begin{example} \label{dexNON}
Let $\mN$ be the network in Figure~\ref{figexNON}. The edges 
of $\mN$ are ordered according to their indices. 
Let $\mU:=\{e_1,e_2,e_3\} \subsetneq \mE$ and $\adv:=\AP{1}{0}{\mU}$.

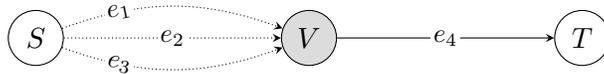
\begin{figure}[htbp]
  \centering
     \begin{tikzpicture}
\tikzset{vertex/.style = {shape=circle,draw,inner sep=0pt,minimum size=1.9em}}
\tikzset{nnode/.style = {shape=circle,fill=myg,draw,inner sep=0pt,minimum
size=1.9em}}
\tikzset{edge/.style = {->,> = stealth}}
\tikzset{dedge/.style = {densely dotted,->,> = stealth}}

\node[nnode] (N) at (0,0) {$V$};
\node[vertex, right=\mynodespace of N] (R) {$T$};
\node[vertex, left=\mynodespace of N] (S) {$S$};

\draw[dedge, bend left=20] (S)  to node[near start,fill=white, sloped, inner sep=1pt]{\small $e_1$} (N);
\draw[dedge, bend right=20] (S) to  node[near start,fill=white, sloped, inner sep=1pt]{\small $e_3$} (N);
\draw[dedge] (S) to node[fill=white, sloped, inner sep=1pt]{\small $e_2$} (N);
\draw[edge] (N) to node[fill=white, sloped, inner sep=1pt]{\small $e_4$} (R);
\end{tikzpicture} 
\caption{{{Network for Example~\ref{dexNON}.}}}\label{figexNON}
\end{figure}

It is easy to see that $1 \in \mR(\mN,\adv)$ for any network alphabet $\mA$. 
We now show that the rate $1$ cannot be achieved employing a linear network code at the intermediate vertex $V$. 

Assume that $\mA=\F_q^m$ for some prime power $q$ and some $m \ge 1$, and that the vertex
 $V$ processes the incoming packets according to a linear network code $\mF_V:\mA^3 \to \mA$ 
 (erasures are excluded from the model). By definition of linear network code (Definition~\ref{defK}), 
 there exist $\lambda_1,\lambda_2,\lambda_3 \in \F_q$ such that 
 $\mF_V(x)=\lambda_1x_1+\lambda_2x_2+\lambda_3x_3$ for all
 $x=(x_1,x_2,x_3) \in \mA^3$. Moreover, by definition of $\adv$ we have
 $$
 \Omega_\mF[\adv;{\bf S} \to \inn(T)](x) \\  = 
 \{\lambda_1y_1 + \lambda_2y_2 + \lambda_3y_3 \,:\, y \in \mA^3  \mbox{ and } y_i \neq x_i 
 \mbox{ for at most one value of $i$}\}.
$$
Assume by contradiction that there exists a good code $\mC \subseteq \mA^3$ for the channel $\Omega_\mF[\adv;{\bf S} \to \inn(T)]$ with
$|\mC|=|\mA|=q^m \ge 2$. Then at least one among $\lambda_1, \lambda_2, \lambda_3$ must be non-zero. 
Without loss of generality, we may assume 
$\lambda_1 \neq 0$. Let $x,x' \in \mC$ with $x \neq x'$, and define 
$$y_1:= \lambda_1^{-1}(-\lambda_2x_2-\lambda_3x_3),~ y_2:=x_2, ~ y_3:=x_3,  \quad\quad
y'_1:= \lambda_1^{-1}(-\lambda_2x'_2-\lambda_3x'_3),~  y'_2:=x'_2,~ y'_3:=x'_3.$$
Then 
$0= \lambda_1y_1+\lambda_2y_2+\lambda_3y_3 = \lambda_1y'_1+\lambda_2y'_2+\lambda_3y'_3 \in  
\Omega_\mF[\adv;{\bf S} \to \inn(T)](x) \cap \Omega_\mF[\adv;{\bf S} \to \inn(T)](x')$, 
contradicting the fact that $\mC$ is good for the channel $\Omega_\mF[\adv;{\bf S} \to \inn(T)]$.
\end{example}

%%%%%%%%%%%%%%%

\section{Other Adversarial Models} \label{extensions}

Using the combinatorial framework developed in this work, other adversarial
models can be investigated.  As already shown for certain  of
adversaries, Lemma~\ref{portinglemma} and~\ref{portinglemma1} allow to port to
the network context any upper bound for the capacity of channels of the form
$\mA^s \dto \hat\mA^s$.  We include the analysis of
error-adversaries having access to overlapping sets of coordinates, and of
rank-metric adversaries.  See~\cite{del1,alb1} for a general reference on
rank-metric codes in matrix representation.

In the sequel we only consider erasure-free adversarial models, and study the capacities 
of certain channels of the form $\mA^s \dto \mA^s$, where 
$\mA$ is a finite set with $|\mA| \ge 2$ and $s \ge 1$. All the capacities are expressed as logarithms in base 
 $|\mA|$.
 
 The upper bounds established in
 this section can be ported to the networking context using the Porting
  Lemmas established in Section \ref{secportinglemma}. The details are left to the reader.

\subsection{Error-Adversaries Acting on Overlapping Sets of Coordinates}

We start by considering $L \ge 1$ adversaries having access to possibly intersecting sets of
coordinates $U_1,...,U_L \subseteq [s]$. The adversaries have error powers
$t_1,...,t_L \ge 0$, and zero erasure powers.

\begin{definition}\label{EXTdefff3}
Let $L \ge 1$ be an integer, and let $U_1,...,U_L \subseteq [s]$. Let $t_1,...,t_L$
be integers with $t_\ell \ge 0$ for all $1 \le \ell \le L$.
Set $\bm{U}:=(U_1,...,U_L)$ and $\bm{t}:=(t_1,...,t_L)$. The channel
$\HHB{t}{U}: \mA^s \dto \mA^s$ is defined by
$$\HHB{t}{U} := \bigcup_{\substack{ p \textnormal{ permutation} \\
\textnormal{of } \{1,...,s\} }} \HHP{t_{p(1)}}{U_{p(1)}} \concat \cdots \concat\HHP{t_{p(L)}}{U_{p(L)}},$$
where $\HHP{t_{p(i)}}{U_{p(i)}} := \HP{t_{p(i)}}{0}{U_{p(i)}}: \mA^s \dto \mA^s \subseteq \hat{\mA}^s$ for all $i \in [s]$. See Definition~\ref{elementarychannels} for details.
\end{definition}

Before proceeding with the analysis of channels of type
 $\HHB{t}{U}$,
we observe that the order in which the adversaries act is in fact irrelevant.

\begin{lemma} \label{lemcomm}
Let $U_1,U_2 \subseteq [s]$ be sets, and let $t_1,t_2 \ge 0$ be integers. Then
$$\HHP{t_1}{U_1} \concat \HHP{t_2}{U_2} = \HHP{t_2}{U_2} \concat \HHP{t_1}{U_1}.$$
\end{lemma}
\begin{proof}
To simplify the notation, we set $\HH_1:=\HHP{t_1}{U_1}$ and 
$\HH_2:=\HHP{t_2}{U_2}$. We need to show that 
for all $x \in \mA^s$ we have $(\HH_1 \concat \HH_2)(x)=(\HH_2 \concat \HH_1)(x)$.
By symmetry, it suffices to show that for all $x \in \mA^s$ we have $(\HH_1 \concat 
\HH_2)(x) \subseteq (\HH_2 \concat \HH_1)(x)$. To see this, fix an arbitrary $y \in (\HH_1 \concat \HH_2)(x)$. Then by definition of concatenation there exists $z \in \mA^s$ such that $z \in \HH_1(x)$ and $y \in  \HH_2(z)$.

Define the sets $\overline{U}_1:=\{1 \le i \le s : z_i \neq x_i\} \subseteq U_1$ and 
$\overline{U}_2:=\{1 \le i \le s : z_i \neq y_i\} \subseteq U_2$. Now construct $z' \in \mA^s$
as follows:
$$\mbox{for $1 \le i \le s$, } \quad z'_i:= \left\{ \begin{array}{ll} y_i & \mbox{ if $i \in \overline{U}_2$,} \\ x_i & \mbox{ if $i \in \overline{U}_1 \setminus \overline{U}_2$,} \\
z_i & \mbox{ otherwise.} \end{array}   \right.$$
One can directly check that $z' \in \HH_2(x)$ and  $y \in \HH_1(z')$. Therefore
$y \in (\HH_2 \concat \HH_1)(x)$. Since $y$ was arbitrary,  this shows that
$(\HH_1 \concat \HH_2)(x) \subseteq (\HH_2 \concat \HH_1)(x)$ for all $x \in \mA^s$,
and concludes the proof.
\end{proof}

\begin{proposition} \label{permuta}
Let $L \ge 1$, $\bm{U}=(U_1,...,U_L)$ and $\bm{t}=(t_1,...,t_L)$ be  as in Definition~\ref{EXTdefff3}. Then
$$\HHB{t}{U} =  \HHP{t_1}{U_1} \concat \cdots \concat \HHP{t_L}{U_L}.$$
\end{proposition}

\begin{proof}
Apply Lemma~\ref{lemcomm} iteratively.
\end{proof}

The compound zero-error capacity of a channel of type $\HHB{t}{U}$ is defined as follows.

\begin{definition}\label{EXTdefcompound}
Let $L \ge 1$,  $\bm{U}:=(U_1,...,U_L)$ and $\bm{t}:=(t_1,...,t_L)$
be as in Definition~\ref{EXTdefff3}.
For $n \ge 1$, the \textbf{compound channel}
$\HHB{t}{U}^{n,\rest}: (\mA^s)^n \dto (\mA^s)^n$ is defined by 
$$\HHB{t}{U}^{n,\rest} := 
%\bigcup_{\substack{ \bm{V} \subseteq \bm{U} \\
%|\bm{V}| \le\bm{t}+\bm e}}   \underbrace{\HHB{t}{V} \times \cdots 
%\times \HHB{t}{V}}_{n  \textnormal{ times}}  =
\bigcup_{\substack{ \bm{V} \subseteq \bm{U} \\
|\bm{V}| \le \bm{t}}}
\HHB{t}{V}^n,$$ 
and the 
\textbf{compound zero-error capacity} of the channel 
 $\HHB{t}{U}$ is the number 
$$\CAzr (\HHB{t}{U}) :=
\sup\left\{\CA (\HHB{t}{U}^{n,\rest})/n : n \in \N_{\ge 1}\right\}.$$
\end{definition}

We can now state the analogue of  Theorem \ref{mainthhamming} for the case 
of error-adversaries acting on possibly overlapping sets of coordinates.

\begin{theorem} \label{EXTmainthhamming}
Let $L \ge 1$,  $\bm{U}:=(U_1,...,U_L)$ and $\bm{t}:=(t_1,...,t_L)$
be as in Definition~\ref{EXTdefff3}.
For all $n \ge 1$ we have
\begin{equation*}
n \cdot \CA(\HHB{t}{U})  \le  
\CA(\HHB{t}{U}^n)  \le   
\CA(\HHB{t}{U}^{n,\rest})  \le    
n \left(s-\sigma_{\bm t} \langle \bm U \rangle \right),
\end{equation*}
where
$$\sigma_{\bm t} \langle \bm U \rangle:= \max \left\{   \left|\bigcup_{\ell=1}^L 
U_\ell^{(1)} \cup U_\ell^{(2)} \right|   \,:\,  U_\ell^{(j)} \subseteq U_\ell \mbox{ and } \left|U_\ell^{(j)} \right| \le t_\ell  \mbox{ for all $1 \le \ell \le L$ and $j=1,2$} \right\}$$
is a parameter which we call the \textbf{adversarial strength}.
In particular,
$$\CA(\HHB{t}{U})  \le  
\CAz(\HHB{t}{U})  \le   
\CAzr(\HHB{t}{U})  \le    
s-\sigma_{\bm t} \langle \bm U \rangle.$$
Moreover, all the above 
inequalities are achieved with equality if $\mA=\F_q$ and $q$ is sufficiently large.
\end{theorem}

To prove Theorem \ref{EXTmainthhamming},
we will need the following preliminary lemma.

\begin{lemmaapp} \label{EXTkeylong}
Let $L \ge 1$,  $\bm{U}:=(U_1,...,U_L)$ and $\bm{t}:=(t_1,...,t_L)$
be as in Definition~\ref{EXTdefff3}. 
Fix any sets
$\overline{U}_\ell^{1}, \overline{U}_\ell^{2}, 
 \subseteq U_\ell$, for 
$1 \le \ell \le L$,  that achieve the maximum in the definition of
the adversarial strength $\sigma_{\bm t} \langle \bm U \rangle$ (see Theorem~\ref{EXTmainthhamming}).
Let $$\overline{U}:=\bigcup_{\ell=1}^L 
\overline{U}_\ell^{1} \cup \overline{U}_\ell^{2},$$
and define $\bm{V}:=(\overline{U}^1_1,...,\overline{U}^1_L) \subseteq \bm{U}$,
 $\bm{V'}=(\overline{U}^2_1,...,\overline{U}^2_L)\subseteq \bm{U}$.
Then for any $x,x' \in \mA^s$ we have 
$\HHB{t}{V}   (x) \cap \HHB{t}{V'}  (x') \neq \emptyset$
whenever $x_i=x'_i$ for all $i \in [s] \setminus \overline{U}$.
\end{lemmaapp}

\begin{proof}
To simplify the notation, let $V_\ell:=\overline{U}_\ell^1$ and $V'_\ell:=\overline{U}_\ell^2$
for all $\ell \in \{1,...,L\}$. Let $x,x' \in \mA^s$ with 
 $x_i=x'_i$ for all $i \in [s] \setminus \overline{U}$.
  We will explicitly construct a vector $z \in 
\HHB{t}{V}   (x) \cap \HHB{t}{V'}  (x')$. First of all,
define the sets
$$\overline{V}:= \bigcup_{\ell=1}^L V_\ell, \quad
\overline{V}':= \bigcup_{\ell=1}^L V'_\ell, \quad \overline{W}:= \overline{V}
\cup \overline{V}'.$$
Note that, by construction, $\overline{W} \supseteq \overline{U}$.
Therefore $x_i=x_i'$ for all $i \in [s] \setminus \overline{W}$.
To construct $z$,  recursively define vectors
$z[0],...,z[L] \in \mA^s$ and  $z'[0],...,z'[L] \in \mA^s$
as follows. Set $z[0]:=x$ and $z'[0]:=x'$.
Then for all $1 \le \ell \le L$ define
$$z[\ell]_i:= \left\{ \begin{array}{ll} x_i' & \mbox{ if $i \in V_\ell$,} \\ z[\ell-1]_i & \mbox{ otherwise},  \end{array}\right.\quad
z'[\ell]_i:= \left\{ \begin{array}{ll} x_i & \mbox{ if $i \in V'_\ell \setminus \overline{V}$,} \\ z'[\ell-1]_i & \mbox{ otherwise},  \end{array}\right.\quad \mbox{ for $1 \le i \le s$.} $$
By construction, for all $1 \le \ell \le L$ we have 
$$z[\ell] \in \HHP{t_\ell}{V_\ell}(z[\ell-1])
\quad \mbox{ and } \quad z'[\ell] \in \HHP{t_\ell}{V'_\ell}(z'[\ell-1]).$$ 
By definition of $\HHB{t}{V}$ and 
$\HHB{t}{V'}$ 
we have $z[L] \in \HHB{t}{V}(x)$
and $z'[L] \in \HHB{t}{V'}(x')$. 
However, it is easy to see that $z[L]=z'[L]$. Therefore 
we can take $z:=z[L]=z'[L]$. 
\end{proof}

\begin{proof}[Proof of Theorem~\ref{EXTmainthhamming}]
The fact that $n \cdot \CA(\HHB{t}{U})  \le  
\CA(\HHB{t}{U}^n)  \le   
\CA(\HHB{t}{U}^{n,\rest})$ can be easily  shown.
We only establish
the theorem for $\sigma:=\sigma_{\bm u} \langle \bm{U} 
\rangle <s$.
The case $\sigma=s$ is similar.
It suffices to show that $\CA(\HHB{t}{U}^{n,\rest}) \le n(s-\sigma)$ for all $n \ge 1$, and that for 
$\mA=\F_q$ and sufficiently large $q$ we have $\CA(\HHB{t}{U} ) \ge    
s-\sigma$.

We let
$\overline{U}_\ell^{1}$, $\overline{U}_\ell^{2}$ 
(for $1 \le \ell \le L$) and $\overline{U}$ be as in Lemma
\ref{EXTkeylong}.
Note that, in particular,  
$\sigma=|\overline{U}|$. In the sequel we denote by $\pi:\mA^s \to
 \mA^{s-\sigma}$ the projection on the coordinates outside $\overline{U}$.

Let $n \ge 1$ be an integer. Then $\pi$ extends component-wise to a map
$\Pi:(\mA^s)^n \to (\mA^{s-\sigma})^n$. Let $\mC\subseteq (\mA^s)^n$ 
be a capacity-achieving good code for $\HHB{t}{U} ^{n,\rest}$.
To obtain the upper bound, it suffices to show that 
the restriction of $\Pi$ to $\mC$ is injective. 

Take $x,x' \in \mC$, and assume
$\Pi(x)=\Pi(x')$. We will show that $x=x'$. Write $x=(x^1,...,x^n)$ and $x'=(x'^1,...,x'^n)$. 
By definition of $\Pi$, we have $\pi(x^k)=\pi(x'^k)$ for all $1 \le k \le n$.
By Lemma~\ref{EXTkeylong}, the
$L$-tuples 
$\bm{V}:=(\overline{U}^1_1,...,\overline{U}^1_L) \subseteq \bm{U}$ and
 $\bm{V'}:=(\overline{U}^2_1,...,\overline{U}^2_L) \subseteq \bm{U}$ 
satisfy:
\begin{equation} \label{nonvuotoEXT}
|\bm{V}|, |\bm{V'}| \le \bm{t} \quad \mbox{ and }  \quad
\HHB{t}{V}  (x^k) \cap \HHB{t}{V'}  (x'^k) \neq \emptyset \quad \mbox{
for all $1 \le k \le n$}.
\end{equation}
By definition of $\HHB{t}{U}^{n,\rest}$ we have
\begin{eqnarray*}
\HHB{t}{U}^{n,\rest}(x) \cap \HHB{t}{U}^{n,\rest}(x') 
&\supseteq& \HHB{t}{V}^n(x^1,...,x^n) \cap \HHB{t}{V'}^n (x'^1,...,x'^n) \\
&=& \prod_{k=1}^n \left( \HHB{t}{V} (x^k) \cap \HHB{t}{V'}(x'^k) \right) \neq \emptyset,
\end{eqnarray*}
where the last inequality follows from (\ref{nonvuotoEXT}).
Since $\mC$ is good for $\HHB{t}{U}^{n,\rest}$,
we conclude $x=x'$. This shows that the restriction of $\Pi$ to $\mC$ is injective, as desired. 

We now prove that the upper bounds in the theorem are tight for $\mA=\F_q$ and $q$ sufficiently large.
As already stated, it suffices to show that $\CA(\HHB{t}{U}) \ge    
s-\sigma$. 
Let $\mC \subseteq \F_q^s$ be code
with cardinality $q^{s-\sigma}$ and minimum distance $\dH(\mC)=\sigma+1$.
We will show that $\mC$ is good for 
$\HHB{t}{U}$. Let $x,x' \in \mC$ be arbitrary with $x \neq x'$.
Assume by contradiction that there exists 
$z \in  \HHB{t}{U}(x) \cap \HHB{t}{U}(x')$.
Then by definition of $\HHB{t}{U}$ and Proposition~\ref{permuta} there
exist vectors $z[0],z[1],...,z[L] \in \mA^s$ and $z'[0],z'[1],...,z'[L] \in \mA^s$
with the following properties: $z[0]=x$, $z'[0]=x'$, $z[L]=z'[L]=z$, 
$$z[\ell] \in \HHP{t_\ell}{U_\ell} (z[\ell-1])  \mbox{ and }    
z'[\ell] \in \HHP{t_\ell}{U_\ell}(z'[\ell-1]) \quad \mbox{ for all $1 \le \ell \le L$}.$$
Now for $1 \le \ell \le L$ define the sets 
$$U_\ell^{1} := \{ 1 \le i \le s : z[\ell]_i \neq z[\ell-1]_i\} \quad \mbox { and } \quad
U_\ell^{2} := \{ 1 \le i \le s : z'[\ell]_i \neq z'[\ell-1]_i\}.$$ By the construction of the 
 $U_\ell^{j}$'s and the definition of 
$\sigma_{\bm u} \langle \bm U \rangle$, we have 
$$\sigma = \sigma_{\bm u} \langle \bm U \rangle \ge \left|\bigcup_{\ell=1}^L 
U_\ell^{1} \cup U_\ell^{2} \right|.$$
On the other hand, 
$$\{ 1 \le i \le s : z_i \neq x_i \mbox{ or } z_i \neq x_i'\} \subseteq \left( \bigcup_{\ell=1}^L U_\ell^{1} \right) 
\cup \left( \bigcup_{\ell=1}^L U_\ell^{2} \right) = \bigcup_{\ell=1}^L 
U_\ell^{1} \cup U_\ell^{2}.$$
Therefore the vectors $z$, $x$ and $x'$ must agree in at least $s-\sigma$ components. In particular, we have
$\dH(x,x') \le \sigma$, a contradiction.
\end{proof}

Theorem~\ref{EXTmainthhamming} can be ported to the network context to study 
the scenario where multiple adversaries have access to possibly overlapping sets of network edges,
and can corrupt up to a certain number of them.

\subsection{Rank-Metric Adversaries}
In this section we let $\mA:=\F_q^m$, and study adversarial channels whose input and output alphabet is the matrix 
space $\mA^s \cong \F_q^{m \times s}$, where $m$ and $s$ are positive integers. In the sequel, we denote by 
$M_1,...,M_s$ the columns of a matrix $M \in \F_q^{m \times s}$. 
We consider an adversary who can access only some of the columns of a matrix $M \in \F_q^{m \times s}$, and is able
to change $M$
in any matrix $N \in \F_q^{m \times s}$ such that $\mbox{rk}(N-M) \le t$, where $t \ge 0$ is an integer
measuring the adversary's power.

\begin{definition}
Let $m,s \ge 1$ and $t \ge 0$ be integers, and let $U \subseteq [s]$
be a subset. The \textbf{matrix channel} $\textnormal{R}_t \langle U \rangle: \F_q^{m \times s} \dto \F_q^{m \times s}$ is defined,
for all $M \in \F_q^{m \times s}$,  by
$$\textnormal{R}_t \langle U \rangle(M):=\{N \in \F_q^{m \times s} : \mbox{rk}(N-M) \le t \mbox{ and } M_i=N_i \mbox{ for all $i \notin U$}\}.$$
\end{definition}

The one-shot capacity and zero-error capacity of a rank-metric channel 
$\textnormal{R}_t \langle U \rangle)$ are as follows.

\begin{theorem} \label{EXTrank}
Let $m,s \ge 1$ and $t \ge 0$ be integers, and let $U \subseteq [s]$.
For all $n \ge 1$ we have
\begin{equation*}
n \cdot \CA(\textnormal{R}_t \langle U \rangle)  \le  
\CA(\textnormal{R}_t \langle U \rangle^n)  \le    
n \left( s-\min\{2t,|U|\} \right).
\end{equation*}
In particular,
$$\CA(\textnormal{R}_t \langle U \rangle)  \le  
\CAz(\textnormal{R}_t \langle U \rangle)  \le    
s-\min\{2t,|U|\}.$$
Moreover, all the above 
inequalities are achieved with equality, provided that $q \ge m \ge s$.
\end{theorem}

\begin{proof}[Proof of Theorem~\ref{EXTrank}]
Define $\sigma:=\min\{2t,|U|\}$. By Proposition \ref{relate}, it suffices to show that $\CA(\textnormal{R}_t \langle U \rangle^n)  \le    
n \left( s- \sigma \right)$ for all $n \ge 1$, and that
for $q \ge m \ge s$ we have $\CA(\textnormal{R}_t \langle U \rangle) \ge s-\sigma$.
We only prove the theorem for $\sigma<s$. The case $\sigma=s$ is similar.
Fix sets $\overline{U}^1, \overline{U}^2 \subseteq U$ with $|\overline{U}^1|, |\overline{U}^2| \le t$ and 
$|\overline{U}^1\cup \overline{U}^2|=\min\{2t,|U|\}=\sigma$. Let $\overline{U}:=\overline{U}^1\cup \overline{U}^2$.

Denote by $\pi:\F_q^{m \times s} \to \F_q^{m \times (s-\sigma)}$ the projection on the columns 
indexed by $[s] \setminus \overline{U}$, and extend $\pi$ component-wise to a projection map
$$\Pi:\F_q^{m \times ns} \to \F_q^{m \times n(s-\sigma)}.$$
Reasoning as in the proof of Theorem~\ref{EXTmainthhamming}, one can show that the restriction of 
$\Pi$ to any good code $\mC \subseteq \F_q^{m \times ns}$ for
$\textnormal{R}_t \langle U \rangle^n$ is injective. This shows the desired upper bound on
$\CA(\textnormal{R}_t \langle U \rangle^n)$.

Now assume $q \ge m \ge s$. If $2t \ge |U|$, then we trivially have 
$\CA(\textnormal{R}_t \langle U \rangle) \ge s-\sigma$.
If $2t < |U|$, then let $\mC \subseteq \F_q^{m \times s}$ be a rank-metric code
in matrix representation \cite{del1,alb1}
of minimum rank distance $2t+1$ and cardinality $|\mC|=q^{m(s-2t)}=q^{m(s-\sigma)}$.
It is easy to see that $\mC$ is good for $\textnormal{R}_t \langle U \rangle$.
\end{proof}

Theorem~\ref{EXTrank} can now be ported to the network context to study 
the scenario where an adversary has access to the packets on a given subset $\mU \subseteq \mE$ of edges,
and can corrupt them by reducing by at most $t$ the rank of the $m \times |\mU|$ matrix whose columns are the 
packets.

%%%%%%%%%%%%%%%

\section{Conclusions} \label{secconcl}
In this paper, we have proposed a combinatorial framework for adversarial
network coding.  We have derived upper bounds for three notions of
capacity region (by porting results for Hamming-type channels to the
networking context in a systematic way) and have given some
capacity-achieving coding schemes.

The results of this paper determine (for sufficiently large network
alphabets) the three capacity regions of the following multi-source
network adversarial models:
\begin{itemize}\setlength\itemsep{0em}
\item error-free networks,
\item a single error-adversary with access to all the network edges,
\item a single adversary who can corrupt/erase a fraction of the
components of any alphabet symbol.
\end{itemize}
In the three models above, any integer point in the capacity regions can
be achieved using linear network coding, possibly combined with
link-level encoding/decoding.  Moreover, in compound adversarial
models the use of a large network alphabet (whose size grows in general
exponentially with the number of sources) can be avoided by using the
network multiple times.

Even small modifications of these models (such as restricting the
adversary to a subset of edges) give rise to capacity regions whose
integer points cannot be achieved in general with linear network coding.

In this paper, we have derived cut-set upper bounds for the three
capacity regions associated with multiple (possibly coordinated) adversaries, who 
can only operate 
on prescribed subsets of the network edges.  The problem of determining the various 
capacity regions in this generalized adversarial models remains open.

%%%%%%%%%%%%%%%%%%%%%%%%%%%%%%%%%%%

\bibliographystyle{IEEEtran}

\bibliography{NC}

\appendix

%\addcontentsline{toc}{chapter}{APPENDICES}

\section{Proofs} \label{Appe}

\begin{proof}[Proof of Proposition~\ref{concprod}]
For $1 \le k \le n$ and $1 \le i \le m$, we denote by $\mX_{k,i}$ and $\mY_{k,i}$ the input and output alphabets, respectively, of channel $\Omega_{k,i}$.
By assumption, we have $\mY_{k,i} \subseteq \mX_{k,i+1}$ for all $1 \le k \le n$ and for all
$1 \le i \le m-1$. We need to prove that for all
$(x_{1,1},...,x_{n,1}) \in \mX_{1,1} \times \cdots \times \mX_{n,1}$ one has
$$\left(\prod_{k=1}^n \left( \Omega_{k,1} \concat \cdots \concat \Omega_{k,m} \right) \right)(x_{1,1},...,x_{n,1})= \left(\left( \prod_{k=1}^n \Omega_{k,1}\right)
\concat \cdots \concat   \left( \prod_{k=1}^n \Omega_{k,m}\right) \right)(x_{1,1},...,x_{n,1}).$$
We will only show the inclusion $\subseteq$. The other containment can be seen similarly.
Fix an arbitrary 
$$(y_{1,m},...,y_{n,m}) \in \left(\prod_{k=1}^n \left( \Omega_{i,1} \concat \cdots \concat \Omega_{i,m} \right) \right)(x_{1,1},...,x_{n,1}) \subseteq \mY_{1,m} \times \cdots \times \mY_{n,m}.$$
By definition of product, for all $1 \le k \le n$ we have 
$y_{k,m} \in (\Omega_{k,1} \concat \cdots \concat \Omega_{k,m})(x_{k,1})$. Thus, by definition of concatenation, for any integer $1 \le k \le n$ there exist elements $x_{k,2},...,x_{k,m}$ with the following properties:
$$x_{k,i} \in \Omega_{k,i-1}(x_{k,i-1}) \subseteq \mY_{k,i-1} \subseteq \mX_{k,i}  \mbox{ for all
$2 \le i \le m$}, \quad \mbox{and} \quad y_{k,m} \in 
\Omega_{k,m}(x_{k,m}).$$
Therefore we have $(x_{1,i},...,x_{n,i}) \in (\Omega_{1,i-1} \times \cdots \times \Omega_{n,i-1}) (x_{1,i-1},...,x_{n,i-1})$ for all
$2 \le i \le m$, and $(y_{1,m},...,y_{n,m}) \in (\Omega_{1,m} \times \Omega_{n,m})(x_{1,m},...,x_{n,m})$.
This implies
\begin{equation*}
(y_{1,m},...,y_{n,m}) \in\left(\left( \prod_{k=1}^n \Omega_{k,1}\right)
\concat \cdots \concat   \left( \prod_{k=1}^n \Omega_{k,m}\right) \right)(x_{1,1},...,x_{n,1}). \qedhere
\end{equation*}
\end{proof}

\begin{proof}[Proof of Proposition~\ref{proprunionfamily}]
Denote by $\mX_j$ and $\mY_j$ the input (resp., output) alphabets of 
channel $\Omega_j$, for $j \in \{1,2\}$. 
By definition of adversarial channel, we need to show that 
for all $x_1 \in \mX_1$ we have
$$\left( \bigcup_{i \in I} (\Omega_1 \concat \Omega_i \concat \Omega_2) \right)(x_1)= \left(\Omega_1 \concat 
\left( \bigcup_{i \in I} \Omega_i \right) \concat \Omega_2\right)(x_1).$$
Let us show the inclusion $\subseteq$. Take an arbitrary $y_2 \in 
\left( \bigcup_{i \in I} (\Omega_1 \concat \Omega_i \concat \Omega_2) \right)(x_1)$. Then there exist
$\overline{\iota} \in I$, $y_1 \in \mY_1 \subseteq \mX$ and 
$y \in \mY$ such that: $y_1 \in \Omega_1(x_1)$, $y \in \Omega_{\overline{\iota}}(y_1)$, and $y_2 \in \Omega_2(y)$.
Therefore $y_2 \in \left(\Omega_1 \concat 
\left( \bigcup_{i \in I} \Omega_i \right) \concat \Omega_2\right)(x_1)$, as desired. The other inclusion
can be shown similarly.
\end{proof}

\begin{proof}[Proof of Lemma~\ref{keylong}]
Part~\ref{keylongitem1} is straightforward, and part~\ref{keylongitem2} follows
from the fact that the $U_\ell$'s are pairwise disjoint. Let us show
part~\ref{keylongitem3}. 
For all 
$1 \le \ell \le L$ define the sets $V_\ell, V'_\ell \subseteq U_\ell$ by
$$V_\ell := \overline{U}_\ell^{1} \cup \overline{U}_\ell^{\star}, \quad\quad
V'_\ell := \overline{U}_\ell^{2} \cup \overline{U}_\ell^{\star}.
$$ 
We claim that $\bm V:=(V_1,...,V_L)$ and $\bm {V'}:=(V'_1,...,V'_L)$ have the desired properties.
First of all,  $|\bm{V}|, |\bm{V}'| \le \bm{t} + \bm e$ by construction. Now let $x,x' \in \mA^s$,
and assume $x_i=x'_i$ for all $i \in [s] \setminus \overline{U}$. We will explicitly construct a vector 
$z \in \HB{t}{e}{V} (x) \cap \HB{t}{e}{V'}  (x')$.
Define the auxiliary sets
$$\overline{U}^{\star}:= \bigcup_{\ell=1}^L \overline{U}_\ell^{\star}, \quad\quad
\overline{V}:= \bigcup_{\ell=1}^L V_\ell, \quad \overline{V}':= \bigcup_{\ell=1}^L V'_\ell.$$
The previous unions are disjoint unions. Moreover, by construction, we have $\overline{V} \supseteq \overline{U}^{\star}$ and $\overline{V}' \supseteq \overline{U}^{\star}$.
Define $z \in \mA^s$ as follows:
$$\mbox{for $1 \le i \le s$, } \quad z_i:= \left\{ \begin{array}{cl} \star & \mbox{ if $i \in  \overline{U}^{\star}$,} \\ x'_i & \mbox{ if $i \in \overline{V} \setminus \overline{U}^{\star}$,} \\  
x_i &  \mbox{ if $i \in \overline{V}' \setminus \overline{V}$,} \\
x_i=x'_i &  \mbox{ otherwise.}
\end{array}   \right.$$
Note that the vector $z$ is well defined. Indeed, we have 
$x_i=x'_i$ for all $i \notin \overline{V} \cup \overline{V}'$, as
 $\overline{V} \cup \overline{V}' \supseteq \overline{U}$.
One can  directly check that $z \in \HB{t}{e}{V}(x) 
\cap \HB{t}{e}{V'}(x')$.
\end{proof}
\end{document}